\documentclass[12pt,a4paper]{book} 
\usepackage[T1]{fontenc}
\usepackage{lmodern}
\usepackage{graphicx}
\usepackage{float}
\usepackage{booktabs}
\usepackage{enumitem}
\usepackage{subcaption}
\usepackage{multicol}
\usepackage[usenames,dvipsnames]{xcolor}
\usepackage{amsmath, amsfonts, mathtools, amsthm, amssymb}
\usepackage{mathrsfs}
\usepackage{euler}
\usepackage{cancel}
\usepackage{enumitem}
\usepackage{cite}       % better grouping of refs., e.g. [7-20]

\definecolor{dblue}{rgb}{0,0,0.6}
\usepackage[breaklinks = true, colorlinks=true, linkcolor=dblue, citecolor=dblue, urlcolor=dblue]{hyperref}
\usepackage{lipsum}

\usepackage{orcidlink}

\usepackage{tikz}

\usepackage{geometry}
\makeatother
\usepackage{thmtools}
\usepackage[framemethod=TikZ]{mdframed}
\mdfsetup{skipabove=1em,skipbelow=1em}

\theoremstyle{definition}
\declaretheoremstyle[
    headfont=\bfseries\sffamily\color{black!70!black}, bodyfont=\normalfont,
    mdframed={
        linewidth=2pt,
        rightline=true, topline=false, bottomline=false,
        linecolor=black, backgroundcolor=black!3!white,
    }
]{thmbox}

\declaretheoremstyle[
    headfont=\bfseries\sffamily\color{black!3}, bodyfont=\normalfont,
    mdframed={
        linewidth=2pt,
        rightline=false, topline=false, bottomline=false,
        linecolor=black, backgroundcolor=black!3!white,align=center,
    }
]{thmtikz}

\declaretheoremstyle[
    headfont=\bfseries\sffamily\color{black!70!black}, bodyfont=\normalfont,
    numbered=no,
    mdframed={
        linewidth=0pt,
        rightline=false, topline=false, bottomline=false,
        linecolor=black, backgroundcolor=black!2!white,
    },
    qed=\qedsymbol
]{thmproofbox}

\declaretheoremstyle[
    headfont=\tiny\color{black!3}, bodyfont=\normalfont, spaceabove=-10pt,
    mdframed={
        linewidth=2pt,
        rightline=true, topline=false, bottomline=false,
        linecolor=black, backgroundcolor=black!3!white,align=center,
    }
]{eqbox}

\declaretheorem[numberwithin=chapter,style=thmbox, name=Definición]{definition}

\declaretheorem[sibling=definition,style=thmbox, name=Theorem, numbered=yes]{theorem}

\declaretheorem[style=thmproofbox, name=Proof]{replacementproof}
\renewenvironment{proof}[1][\proofname]{\vspace{-10pt}\begin{replacementproof}}{\end{replacementproof}}

\declaretheorem[style=eqbox, numbered=no, name=.]{eqnbox}

\usepackage{epigraph}
\usepackage{fancyhdr}

\usepackage{emptypage}
\newcommand*{\blankpage}{
{\newpage \vspace*{5cm}\thispagestyle{empty}\centering \bfseries  \textit{This page is intentionally left blank.} \par}
\vspace{\fill}}

\DeclareFixedFont{\titlefont}{T1}{ppl}{bx}{n}{0.75in}

\makeatletter                       
\def\printauthor{%                  
    {\large \@author}}              
\makeatother
\author{Names of Authors\\
    }

\newcommand\titlepagedecoration{%
\begin{tikzpicture}[remember picture,overlay,shorten >= -10pt]

\coordinate (aux1) at ([yshift=-15pt]current page.north east);
\coordinate (aux2) at ([yshift=-410pt]current page.north east);
\coordinate (aux3) at ([xshift=-4.5cm]current page.north east);
\coordinate (aux4) at ([yshift=-150pt]current page.north east);

\begin{scope}[black!40,line width=12pt,rounded corners=12pt]
\draw
  (aux1) -- coordinate (a)
  ++(225:5) --
  ++(-45:5.1) coordinate (b);
\draw[shorten <= -10pt]
  (aux3) --
  (a) --
  (aux1);
\draw[opacity=0.6,black,shorten <= -10pt]
  (b) --
  ++(225:2.2) --
  ++(-45:2.2);
\end{scope}
\draw[black,line width=8pt,rounded corners=8pt,shorten <= -10pt]
  (aux4) --
  ++(225:0.8) --
  ++(-45:0.8);
\begin{scope}[black!70,line width=6pt,rounded corners=8pt]
\draw[shorten <= -10pt]
  (aux2) --
  ++(225:3) coordinate[pos=0.45] (c) --
  ++(-45:3.1);
\draw
  (aux2) --
  (c) --
  ++(135:2.5) --
  ++(45:2.5) --
  ++(-45:2.5) coordinate[pos=0.3] (d);   
\draw 
  (d) -- +(45:1);
\end{scope}
\end{tikzpicture}%
}

\usepackage{framed}
\usepackage{titletoc}

\patchcmd{\tableofcontents}{\contentsname}{\sffamily\contentsname}{}{}

\renewenvironment{leftbar}
  {%
    \MakeFramed{\parshape 1 0cm \dimexpr\textwidth-6em\relax\FrameRestore}\vskip2pt%
  }
 {\endMakeFramed}

\titlecontents{chapter}
  [0em]{\Large\bfseries\vspace*{2\baselineskip}}
  {\hspace{-1.2cm}\parbox{6.5em}{%
    \hfill\Huge\sffamily\bfseries\color{black}}% \thecontentspage
   \vspace*{-1.9\baselineskip}\leftbar\sffamily\bfseries}
  {}{\endleftbar}

\titlecontents{section}
  [11.4em]
  {\sffamily\contentslabel{3em}}{}{} 
  {\phantom{    }\dotfill \nobreak\itshape\color{black}\contentspage}

\titlecontents{subsection}
  [13em]
  {\sffamily\contentslabel{3em}}{}{}  
  {\phantom{    }\dotfill \nobreak\itshape\color{black}\contentspage}

\usepackage{titlesec}

\titleformat{\chapter}
{\vspace{-1.2cm}\normalfont\fontsize{35}{25}\bfseries}{}{0pt}{}

\titleformat{\section}
{\normalfont\Large\bfseries}{~\thesection}{1em}{}

\usepackage{fancyhdr}

\usepackage{physics}
\usepackage{amsmath,amssymb,amstext,amsthm,mathtools}
\usepackage{empheq}
\usepackage[most]{tcolorbox}

\usepackage{mathrsfs}
\usepackage{stmaryrd} % for variable-sized delimiters

\newcommand*{\commut}[1]{[ #1 ]}

\newcommand{\BoxedEquation}[1]{\begin{eqnbox}\vspace{-12pt}
                                        \begin{align} #1 \end{align} \end{eqnbox}}

\newcommand{\bs}[1]{\mathbf{#1}} 
\newcommand{\diag}{\operatorname{diag}}

\usepackage{color} %% colored text
\definecolor{purple}{rgb}{0.63,0,1}
\definecolor{pink}{rgb}{1,0,0.9}

\begin{document}

\begin{titlepage}
\vspace*{3cm}

\noindent
\vspace*{0.5cm}
%\titlefont{Numerical\\ Dynamics\\ Simulations\\[10pt] of\\[10pt] Quantum\\ Systems}

\begin{center}

% Algebra
\titlefont{Observables\\[8pt] in Motion:\\[12pt] 
\huge{{A guide to simulating classical\\[5pt] and quantum dynamics}
}}
    
\end{center}

\vspace{5cm}
%\epigraph{<<``Epic Phrase''>>}%
%{ \textsc{Author}}
%\null\vfill
\vspace*{1cm}
\noindent
\hfill
\begin{minipage}{0.7\linewidth}
    \begin{flushright}
        Denys I. Bondar\orcidlink{0000-0002-3626-4804}, 
        Gerard McCaul\orcidlink{0000-0001-7972-456X}, 
        Andrii Sotnikov\orcidlink{0000-0002-3632-4790}
    \end{flushright}
\end{minipage}
\begin{minipage}{0.02\linewidth}
    \rule{1pt}{70pt}
\end{minipage}
\titlepagedecoration

\vspace*{4cm}

\begin{center}
Copyright \textcopyright\ 2025 by Denys I. Bondar, Gerard McCaul,  Andrii Sotnikov. All rights reserved.
\end{center}

\end{titlepage}

\let\cleardoublepage=\clearpage
\tableofcontents
\blankpage

\chapter*{Table of notations}
\addcontentsline{toc}{section}{Table of notations}
%\begin{table}[htbp]\caption{Table of Notation My Research}
\begin{center}% used the environment to augment the vertical space
% between the caption and the table
\begin{tabular}{r c p{10cm} }
\toprule
\multicolumn{3}{c}{}\\
\multicolumn{3}{c}{\underline{Abbreviations}}\\
\multicolumn{3}{c}{}\\
FFT & --- & Fast Fourier Transform \\
FFT & --- & Fast Fourier Transform \\
FRFT & --- & Fractional Fourier Transform \\
\multicolumn{3}{c}{}\\
\multicolumn{3}{c}{\underline{Variables and Functions}}\\
\multicolumn{3}{c}{}\\
$\hat{x}$
	 & --- & classical position operator\\
$\hat{p}$
	 & --- & classical momentum operator\\
$\hat{\bs{x}} = \hat{x} - \hbar \hat{\theta} /2$
	 & --- & quantum position operator\\
$\hat{\bs{p}} = \hat{p} + \hbar \hat{\lambda} /2$
	 & --- & quantum momentum operator\\
$\mathcal{F}$ & --- & direct Fourier transform \\
$\mathcal{F}^{-1}$ & --- & inverse Fourier transform \\
\bottomrule
\end{tabular}
\end{center}
%\label{tab:TableOfNotationForMyResearch}
%\end{table}

\chapter*{Introduction}
\addcontentsline{toc}{section}{Introduction}
\markboth{Introduction}{Introduction}
%\lhead{Introduction}
%\rhead{Introduction}
Part of the magic of physics is its capacity to infer so much from so little. Equipped with a theory for dynamics, it is possible to alchemize the reagents of initial conditions into a prediction of the future. The true astonishment is not that prediction is possible, but that the engines that power it are, when properly expressed, disarmingly simple.

This book is a map of those engines—both the mathematics of dynamics and the craft of simulating it—and an argument that an algebraic, vector-space view unifies what too often is taught as two disjoint worlds: the classical and the quantum. The former is far more easily digested conceptually, but the gap between each regime's usual mathematical representation is typically a source of endless confusion. The description of quantum mechanics often appears entirely alien when compared to classical physics. Abstracted to the point of being conceptually air-gapped from its antecedents. 

This is not a book that “applies” mathematics after the fact. It uses algebra as the native habitat of dynamics. We move deliberately between derivation and device, proof and program, because the point is not only to understand evolution but also to compute it faithfully. Throughout, we treat vector spaces as the common stage. In classical mechanics, they host densities; in quantum mechanics, they host amplitudes. In both cases, evolution is a linear operator—a propagator—and numerical methods are disciplined approximations of its action. The thesis is simple: once you adopt this algebraic view, the ``mysteries” of quantum theory look less like magic and more like the inevitable consequences of representation.

Before embarking on this journey, it is worth taking a little time to orient ourselves. The usual pedagogy drops quantum mechanics into the story like a deus ex Hilbert, but here we follow a different route. As we aim to emphasize the unity of description between quantum and classical physics, it is natural to first follow the thread of the latter's development. We may then follow the natural gradient of progress that runs from measurement to model, from kinematics to evolution, from forces to flows, from trajectories to distributions, and finally into Hilbert space, where probability is not an add-on but the oxygen the theory breathes. To that end, we begin with a short survey of the history of dynamical description up to the birth of quantum mechanics. 

It is in this context we introduce vector space representations as a natural language for such statistical treatments, and begin our exposition in anger. The rest of the book is outlined as follows:
\begin{itemize}
    \item Chapter~\ref{Chapter:1} — Mathematical background.
    \item Chapter~\ref{Chapter:2} — Kinematics. Classical and quantum cases are placed side by side and analyzed using the same framework.
    \item Chapter~\ref{Chapter:3} — Formal evolution. Stone’s theorem to time-ordering, controlled approximations, and integral equations for the propagator.
    \item Chapter~\ref{Chapter:4} — Prerequisites for dynamics. From data to differential equations; numerics for the time-independent Schr{\"o}dinger equation; Fourier and fractional Fourier as complementary lenses.
    \item Chapter~\ref{Chapter:5} — Time-dependent Schrödinger equation. Split-operator methods; diagonalization by time propagation; imaginary time for eigenstates and gaps; absorbing boundaries; internal degrees and lattices.
    \item Chapter~\ref{Chapter:6} — Classical dynamics. Ehrenfest to Liouville to Newton, and back again; symplectic integrators that respect the geometry they compute.
    \item Chapter~\ref{Chapter:7} — Open quantum systems. Phase-space formulations; density matrices and uncertainty; completely positive maps and Lindblad; when and how baths thermalize; stochastic unravellings.
\end{itemize}

Before plunging into the thrilling world of linear algebra, we begin our detour into history.

\section*{Galileo and the birth of dynamics}
Adam, it was said, knew all things. At least all things that it was possible to know. It was \textit{scientia infusa}—infused and intrinsic knowledge. The truths embedded in the created world were known best by the ancients, and in a fallen world could only decay. It was—and in many ways still is—the apocalyptic core of the Christian ethos, which regarded the workings of nature not as a frontier to be mapped but a ruin to be excavated. Sin was not just moral corruption but an epistemic rupture, and therefore the weight given to any source must be in inverse proportion to its proximity to the taint. \textit{Scientia acquisita}—the knowledge we acquire through experience—must therefore be regarded as unreliable, suspect, and liable to malign warpings. Do not believe your eyes, they lie. Do not rely on your reasoning; it is insufficient. Trust Aristotle, and when the evidence of the senses comes into conflict with the wisdom of the beatified past, do not be alarmed. It is simply evidence of the imperfection of immanence, of the incapacity of any mortal to encompass the divine.  

It was into this world that Galileo was born. He was the first to systematically perform what we today recognize as experimental science; and yet, in his thinking, he remained an unreconstructed theorist. In all his writing, he prioritizes deduction first, and experiment second. He was nevertheless guided by the latter in a manner that set the standard of scientific practice up to today. He aimed to explain—by all methods available—not what he had been told, but what he could \textit{see}. Galileo was a knot under tension, a bundle of contradictions. It is a miracle of will that this was alchemized into a vision of new mental horizons. It would be easy to begin a narrative of scientific development before him, but it would be impossible without him. He is, in short, an irresistible starting point. 

 Galileo’s early intellectual labor was embedded in the practical demands of patronage and pedagogy, especially the needs of aristocratic students who demanded experiential, useful knowledge. That early crucible trained him to build insight from what could be touched, not just what could be read. Galileo begins to attack the central question of motion itself. Not as a philosopher, but as a mechanic. How do bodies fall? Not: why do they fall, or to what end, or what metaphysical purpose they serve—but \textit{how}, exactly, do they move? The Aristotelian answer was elegant in its circularity: heavier things fall faster because they contain more of the element Earth, and the nature of Earth is to return to its natural place. There was no error in this schema, because there was no mechanism of correction. What Galileo does, and is the first to do in such a systematic manner, is to treat falling objects as data. To isolate the event so it can be observed, measured, and analyzed. 

The first true dynamical law arises from this, \textit{The law of fall}. It is the principle that distance increases with the square of time and arises only from trial, instrumentation, and iteration. It requires something revolutionary: the ability to mark time precisely.

Galileo has no stopwatch. He has no atomic clock. What he has is the pulse in his neck, and the chandelier in the cathedral. From these, he begins to build a scale. A musician’s sense of rhythm weaponized into proto-physics. He rolls balls down inclined planes and counts the intervals using water clocks and his heartbeat. It is clumsy, ingenious, and deranged. It reflects an underlying belief not only in an underlying regularity to natural phenomena, but that it can be used to interrogate \textit{itself}. The mind may lie, but matter cannot. It works. Theoretical abstractions of idealized systems, allied to careful observation and mechanical cunning. It is a procedure to bootstrap the tools and techniques of scientific inquiry into existence. It has to be created before it can be named. 

From this methodology, Galileo is able to gather the essential ingredients necessary for a new conception of reality. One in which the character of motion is indifferent to purpose, and which draws no distinction between the terrestrial and heavenly spheres. Aristotle is not merely wrong, he is \textit{demonstrably} wrong. Falling objects do not accelerate linearly, but quadratically. The natural state of matter is not to be at rest, but to \textit{resist change}. They are lazy, they are \textit{inertial}. This is not a minor refinement. It is a rupture. The first crack in a two-thousand-year wall. Galileo is not using equations— not yet. The language of calculus does not exist. Vectors are still centuries away. What he is doing is inventing the basic categories of dynamical analysis—position, time, acceleration— before there exists a language to contain them. Piece by piece, he is dragging the grammar of physics into being.

\textit{Eppur si muove} are his (unfortunately apocryphal) final words to us. A last act of defiance in a life fueled by it. What makes Galileo extraordinary, what makes him \textit{necessary} to this story is not merely the ideas that he discovered. It is the purpose that he put them to. It was not enough to understand the inertial quality of motion and rewrite the laws of fall. It is not enough to observe that the moon possesses a topography as crude as the Earth's, or that Jupiter hides satellites of its own. Galileo did not stop at the discovery, he \textit{weaponized} it. With it, he fatally wounded the Aristotelian worldview, and dismantled the distinction between the heavenly and the terrestrial. He turned his contempt for the conventional into something scarcely credible. A \textit{reshaping} of it. More than discovery, it initiated a profound reorientation of the human condition that continues to this day. We owe Galileo and his contemporaries an immense debt of gratitude, for being the first to articulate and advocate for an idea. It may be imperceptible, it may be unalterable, but the world \textit{does} move. And now we need the language to describe it. Enter Newton. 

\section*{Newtonian dynamics}
Newton, despite his revolutionary approach, did not consider his laws to be an innovation in the modern sense. He saw himself as recovering the fundamental principles of motion, bringing clarity to ideas that had always been present in the natural world. As he put it, \textit{“Truth is ever to be found in simplicity, and not in the multiplicity and confusion of things.”} His work was framed within a world that still saw knowledge as rediscovery rather than invention.

This perspective was deeply embedded in the intellectual climate of the Renaissance and early modern period. The guiding assumption of the time, reflected in Ecclesiastes’ famous words, \textit{“there is nothing new under the sun,”} meant that true wisdom was a process of unveiling, not creating. This mindset influenced the acceptance of Newton’s laws, not as groundbreaking in their novelty, but as the proper articulation of physical truths that had always existed.

The Renaissance was also the age of ballistics and navigation. Motion was not an abstract concept but a pressing practical concern. Cannons fired over battlefields, ships crossed oceans, and trajectories needed to be predicted with increasing precision. In this context, the Newtonian framework naturally emerged as a system suited for understanding and controlling motion in the real world. The very formulation of dynamics, built around rectilinear motion and decomposable forces, was tailored to the concerns of the ballistic age.

\subsection*{The three laws}
 With the development of true mathematical models of dynamics, our need for lurid prose evaporates, and we may concentrate on the models themselves.  Newton's three laws provide the foundation for classical mechanics, but they lack conceptual depth in their first formulation. They assert truths about motion but do not fully explain the underlying principles. They function well enough to describe the world, but they feel more like an initial codification of intuition than a fully realized theoretical framework.  While they scarcely need repeating, Newton's laws three laws may be summarized (together with some typical conceptual issues) as follows.
\subsubsection*{First Law: The law of inertia}
A body at rest remains at rest, and a body in motion remains in uniform motion unless acted upon by an external force.

Mathematically:
\begin{equation}
    \frac{d\mathbf{v}}{dt} = 0 \Rightarrow \mathbf{v} = \text{const}.
\end{equation}

\textbf{Conceptual issues:}
\begin{itemize}
    \item What is ``uniform motion''? Is motion relative? How do we define a true state of rest?
    \item What is an ``external force''? Newton defines force implicitly through its effects.
    \item Why does inertia exist? The law tells us what happens but does not explain why objects have inertia.
\end{itemize}

\subsubsection*{Second Law: The equation of motion}
Newton's second law formalizes the relationship between force, mass, and acceleration:
\begin{equation}
    \mathbf{F} = m \mathbf{a},
\end{equation}
where \( \mathbf{F} \) is the force, \( m \) is mass, and \( \mathbf{a} \) is acceleration.

\textbf{Conceptual issues:}
\begin{itemize}
    \item \textbf{What is force?} Defined circularly: force causes acceleration, and acceleration is caused by force.
    \item \textbf{What is mass?} Newton treats it as intrinsic, but why does it resist acceleration?
    \item \textbf{Why should force be proportional to acceleration?} This is not derived but postulated.
    \item \textbf{Instantaneous action?} Newton assumes forces act immediately, which is later corrected in field theories.
\end{itemize}

\subsubsection*{Third Law: Action and reaction}
For every action, there is an equal and opposite reaction:
\begin{equation}
    \mathbf{F}_{AB} = -\mathbf{F}_{BA}.
\end{equation}
This law ensures conservation of momentum but does not explain the nature of interactions.

\textbf{Conceptual issues:}
\begin{itemize}
    \item What carries the force? Newton assumes forces ``exist'' between objects, but later field theories show they are mediated by interactions.
    \item Does this hold universally? It fails in some cases, such as electrodynamics, where force propagation is not instantaneous.
\end{itemize}

\subsection*{The role of coordinate systems and the Newtonian worldview}
Newtonian mechanics, though often framed in Cartesian coordinates, does not explicitly require them. However, the Cartesian framework, developed by René Descartes in the 17th century, became the natural mathematical setting for Newton’s equations of motion. The ability to describe motion through independent coordinate axes allowed forces to be decomposed into components, making the prediction of motion more manageable. 

The notion of a coordinate system is fundamental to how we describe and interpret physical laws. It is not an intrinsic part of reality but a choice that shapes how we frame and solve problems. Cartesian coordinates impose a rectilinear grid onto space, breaking motion into independent equations along each axis:
\begin{equation}
    F_x = m \frac{d^2 x}{dt^2}, \quad F_y = m \frac{d^2 y}{dt^2}, \quad F_z = m \frac{d^2 z}{dt^2}.
\end{equation}
This formulation, while powerful, remains tied to a specific way of looking at the world. The Cartesian grid assumes an underlying space in which positions are absolute and can be defined relative to a fixed reference frame. Yet, this assumption itself is an artifact of the mathematical tools available in Newton’s time. 

While Cartesian coordinates provided a practical way to decompose forces and motion, they are not always the most natural framework. Many physical systems—especially those involving rotation or constraints—are cumbersome when forced into a rectilinear grid. This limitation points towards a more general framework, one where coordinates are chosen to fit the system rather than the other way around. This leads us toward generalized coordinates and the Lagrangian formulation of mechanics, where motion is encoded in energy principles rather than direct force analysis.

\section*{The principle of least action and Lagrangian mechanics}

Classical mechanics, as formulated by Newton, describes motion through forces and differential equations. However, hidden within Newtonian mechanics are conservation laws that suggest a deeper principle underlying motion. Instead of forces dictating acceleration, we can approach motion as following from a more fundamental principle—the extremization of an action integral. This shift not only provides a more elegant formulation but also reveals a structure that connects classical mechanics with more advanced theories such as quantum mechanics and general relativity.

\subsection*{Historical origins of variational principles}
Variational principles were not first developed for mechanics but emerged in optics. \textbf{Fermat’s principle of least time} (1657) states that light follows the path that requires the shortest travel time between two points. This principle suggested that nature tends to select optimal paths, a striking departure from force-based, differential descriptions of motion.

This approach was later generalized to mechanics, particularly in the \textbf{brachistochrone problem}, posed by Johann Bernoulli in 1696. The challenge was to determine the curve along which a frictionless bead, sliding under gravity, reaches the lowest point in the shortest time. While Newton solved this problem using direct computation, the solutions of Jacob Bernoulli and Euler led to the development of the calculus of variations, which underlies modern Lagrangian mechanics.

\subsection*{The hidden conservation laws of Newtonian mechanics}
Newton’s equations of motion encode deep symmetries. Conservation laws emerge naturally from these equations:

\begin{itemize}
    \item \textbf{Momentum conservation}: If no external force acts, momentum remains constant.
    \item \textbf{Energy conservation}: If the forces are conservative, total energy is preserved.
    \item \textbf{Angular momentum conservation}: If no external torque acts, angular momentum is conserved.
\end{itemize}

These conservation laws are \textbf{not assumptions} but emerge as consequences of Newton’s laws. However, their existence suggests a deeper underlying structure: what if motion itself is dictated not by forces but by an overarching principle?

\subsection*{The principle of least action}
Instead of forces acting on particles, we frame motion as arising from a quantity that is extremized: the \textit{action} \( S \), defined as:
\begin{equation}
    S = \int L(q,\dot{q},t) \, dt,
\end{equation}
where $q(t)$ is the trajectory of particle in coordinate space, and the \textbf{Lagrangian function} is:
\begin{equation}
    L = T - V.
\end{equation}
where $T$ and $V$ are the \textit{kinetic} and \textit{potential} energies of the system. Rather than solving \( \mathbf{F} = m \mathbf{a} \) explicitly, we determine the motion by requiring that \( S \) be \textit{stationary} under small variations of the path.

To find the function \( q(t) \) that extremizes \( S \), consider a small perturbation:
\begin{equation}
    q(t) \to q(t) + \epsilon \eta(t),
\end{equation}
where \( \eta(t) \) is an arbitrary small function vanishing at the endpoints: \( \eta(t_1) = \eta(t_2) = 0 \).

Expanding the action:
\begin{equation}
    S[\epsilon] = \int_{t_1}^{t_2} L(q + \epsilon \eta, \dot{q} + \epsilon \dot{\eta}, t) dt.
\end{equation}
Expanding to first order in \( \epsilon \):
\begin{equation}
    S[\epsilon] = S[0] + \epsilon \int_{t_1}^{t_2} \left( \frac{\partial L}{\partial q} \eta + \frac{\partial L}{\partial \dot{q}} \dot{\eta} \right) dt + O(\epsilon^2).
\end{equation}
Since \( S \) is extremized, the first-order term must vanish:
\begin{equation}
    \int_{t_1}^{t_2} \left( \frac{\partial L}{\partial q} - \frac{d}{dt} \frac{\partial L}{\partial \dot{q}} \right) \eta dt = 0.
\end{equation}
Since \( \eta \) is arbitrary, we obtain the \textbf{Euler-Lagrange equation}:
\begin{equation}
    \frac{d}{dt} \frac{\partial L}{\partial \dot{q}} - \frac{\partial L}{\partial q} = 0.
\end{equation}

To see how Newton’s laws emerge from the Euler-Lagrange equations, consider the standard Lagrangian for a particle in a potential:
\begin{equation}
    L = \frac{1}{2} m \dot{q}^2 - V(q).
\end{equation}
Computing the derivatives:
\begin{equation}
    \frac{\partial L}{\partial q} = -\frac{dV}{dq}, \quad \frac{\partial L}{\partial \dot{q}} = m \dot{q}, \quad \frac{d}{dt} \frac{\partial L}{\partial \dot{q}} = m \ddot{q}.
\end{equation}
Plugging into the Euler-Lagrange equation:
\begin{equation}
    m \ddot{q} = -\frac{dV}{dq},
\end{equation}
which is just Newton’s second law:
\begin{equation}
    m a = F.
\end{equation}
Thus, forces and vectors reappear as the \textit{derivatives} of a single function—the Lagrangian—eliminating the need for force components while preserving the essence of Newtonian mechanics.

This marks a profound conceptual shift. Newtonian mechanics works with \textit{local} forces acting at each instant, while the variational principle, on the other hand, considers the \textit{global} evolution of a system. This approach:
\begin{itemize}
    \item \textbf{Eliminates explicit force vectors}, replacing them with derivatives of a scalar function.
    \item \textbf{Is coordinate-independent}, allowing for more flexibility in describing constrained systems.
    \item \textbf{Connects to deeper principles}, laying the groundwork for Hamiltonian mechanics and quantum physics.
\end{itemize}
This shift transforms mechanics from an equation-driven approach to the one where motion follows from a fundamental principle.

\section*{Hamiltonian dynamics}

We began with Newtonian mechanics, where forces determine acceleration, leading to second-order equations of motion in \textit{configuration space} (positions and velocities). Then we saw Lagrangian mechanics, which reformulated motion in terms of an extremization principle, encoding the equations of motion via the Euler--Lagrange equations. However, Lagrangian mechanics still describes motion as \textit{trajectories in configuration space}. The fundamental variables are:

\begin{itemize}
    \item \( q \) (positions, generalized coordinates),
    \item \( \dot{q} \) (velocities, time derivatives of coordinates).
\end{itemize}

There is a problem here. \textbf{Velocities are not truly independent variables}—they depend on the choice of coordinates and time. To describe mechanics in a way that encodes constraints naturally, provides a structure that generalizes to statistical ensembles, and allows for deeper conservation principles to emerge naturally, we must introduce a new concept: \textbf{canonically conjugate coordinates}.

\subsection*{The legend of the Legendre transform}
All this is achieved by a deceptively simple operation known as the \textit{Legendre transform}. This maps a \textit{velocity-based description} to a \textit{momentum-based description}. It introduces a new set of variables, where we replace \( \dot{q} \) with \( p \) (momentum) as our fundamental descriptor of motion:
\begin{equation}
    p_i = \frac{\partial L}{\partial \dot{q}_i}.
\end{equation}
This is a one-to-one mapping (if \( L \) is convex in \( \dot{q} \)), meaning that for every valid velocity, we can assign a corresponding momentum. The key observation is that \textbf{momentum contains the same information as velocity but is often more fundamental}.

Now, we define a new function:
\begin{equation}
    H(q, p) = \sum_i p_i \dot{q}_i - L.
\end{equation}
This \textbf{Legendre transform} changes the variables of the system from \( (q, \dot{q}) \) to \( (q, p) \). What we are really doing is reframing mechanics in a way that \textbf{symmetrizes} the equations of motion.

\subsection*{Canonically conjugate coordinates: A pair of variables}
The critical realization here is that \( q \) and \( p \) form a natural mathematical pair. The Hamiltonian description elevates momentum to be \textbf{on equal footing} with position:
\begin{equation}
\label{eq:Hamiltonequation}
    \frac{dq_i}{dt} = \frac{\partial H}{\partial p_i}, \quad \frac{dp_i}{dt} = -\frac{\partial H}{\partial q_i}.
\end{equation}
Instead of viewing motion as a trajectory of \textit{positions evolving in time}, we now view it as a \textit{structured evolution in phase space}, where \textbf{position and momentum evolve together}.

Perhaps at first glance, the utility of this idea is non-obvious. What has really changed? In terms of the equations one needs to solve, we have swapped a second order differential equation based on the Lagrangian, for a pair of coupled first order equations. What is remarkable is that such a slight change of perspective provides the most natural route to enlarging our mathematical models to encompass all known physics, both classical and quantum. 

How does this happen? Consider first that the state of our system at any moment in time is completely characterized by its set of canonical coordinates and momenta. That is, our system is now represented by a single \textit{point} in phase space. The dimensionality of phase space is, of course, very different from real space, where three spatial dimensions correspond to a \textit{six-dimensional} phase space (for a single particle). This is just for a system that only requires one set of canonical coordinates! In general, the dimensionality of phase space will grow linearly with the number of particles—independent canonical coordinates—in the system we are modeling. One might argue that a similar perspective can be fashioned from the Lagrangian representation, but the true power of Hamiltonian dynamics does not lie solely in its mapping of the states of physical systems to points in phase space. The truly revelatory upgrade is that phase space provides a geometric interpretation of the dynamics of those states. 

\subsection*{Hamiltonian flow}
What do we mean by this? Put simply, Hamilton's equations define a \textit{flow} over phase space. Each point in the space represents a physical state of the system, which together cover all possible configurations of position and momenta. Now, take \textit{any} point in the space, and the Hamiltonian $H$  associates to that point a velocity via the Hamilton's equations. Taken over the entirety of phase space, $H$ defines a \emph{velocity field} $\Gamma(q,p)$:
\[
\Gamma(q,p) 
\;=\;
\bigl(\dot{q},\,\dot{p}\bigr)
\;=\;
\Bigl(
\tfrac{\partial H}{\partial p},
-\tfrac{\partial H}{\partial q}
\Bigr).
\]

That is, each point in phase space has a velocity determining how a trajectory through that point will evolve. Critically, this means we're no longer tied to thinking about our dynamics in terms of its evolution from a particular initial condition. Instead, the Hamiltonian gives us a flow through space—a description of how each point $(q,p)$ in phase space is evolved in time, determined solely by $H$. We are describing how \textit{all possible system states} evolve, all at the same time. More than that, we can describe how any function $F(q,p)$ of our phase space coordinates will vary along one of these Hamiltonian trajectories. Specifically, if we consider its total time derivative along a Hamiltonian flow, we have
\[
\frac{dF}{dt}
\;=\;
\frac{\partial F}{\partial q}\,\dot{q}
\;+\;
\frac{\partial F}{\partial p}\,\dot{p}.
\]
Substituting the Hamiltonian equations for $\dot{q}$ and $\dot{p}$ we can define this time-dependence in terms of the \textit{Poisson bracket}
\[
\{F,\,H\}
\;=\;
\frac{\partial F}{\partial q}\,\frac{\partial H}{\partial p}
\;-\;
\frac{\partial F}{\partial p}\,\frac{\partial H}{\partial q},
\]
and hence:
\[
\frac{dF}{dt}
\;=\;
\{\,F,\,H\}.
\]
One immediate consequence of this is that we have a general condition which conserved quantities will identically obey - $\{\,F,\,H\} =0$, no matter what the system's initial conditions should be. Moreover, it provides a direct justification for why it should be the total energy—the sum of kinetic $T$ and potential $V$ energies—that  is conserved. If $L=T-V$, then the corresponding Hamiltonian is $H= T+V$. The phase space picture of dynamics as a Hamiltonian flow in phase space tells us that its conservation is not an accident, but in fact an essential consequence of its true identity. $H$ is not simply conserved by dynamics, it is their \textit{generator}. 

The conceptual gains made by the Hamiltonian formulation do not end there however. By representing dynamics as a flow over the space of all possible system states, we are no longer bound to thinking about evolutions according to a particular initial condition. Instead of taking a single system configuration—a single point in phase space—as our initial condition, we can instead consider a \textit{distribution} $\rho(q,p,t)$ over phase space. That is, rather than picking some initial condition $(q_0, p_0)$,  we can consider an  \textit{ensemble} of initial states using $\rho(q,p,0)$, and predict how this distribution evolves with time. In this manner, it is possible to encode our own uncertainty into dynamics. In particular, if we interpret $\rho(q,p,t)$ as a \textit{probability density}, the conservation of probability demands that at all times:
\begin{equation}
    \int {\rm d} q {\rm d} p \ \rho(q,p,t) = 1.
\end{equation} 

This is equivalent to a \textit{continuity} equation over $\rho$:
\[
\frac{\partial\rho}{\partial t}
\;+\;
\nabla \cdot \bigl(\rho\,\mathbf{v}\bigr)
\;=\;
0,
\]
where $\mathbf{v}$ is the velocity field. This is in precise analogy to physical fluid flows—in the absence of a sink or source, all the density entering a region has to be balanced by the amount leaving it. Again, this is just another way of stating that if we are describing our system as a probability distribution over all its possible states, that probability cannot ``leak". It always has to add up to one over the space of all possible states. In phase space, $\mathbf{v} = \Gamma$, so
\[
\frac{\partial\rho}{\partial t}
\;+\;
\nabla_{(q,p)} \cdot \bigl(\rho\,\Gamma\bigr)
\;=\;
0.
\]
For Hamiltonian flows, we have \emph{incompressibility} in phase space:
\[
\nabla_{(q,p)} \cdot \Gamma
\;=\; \frac{\partial^2 H}{\partial x\partial p}- \frac{\partial^2 H}{\partial x\partial p} \;=\; 0.
\]
This property is baked into the structure of Hamilton's equation, and means that when the divergence is expanded, the continuity equation is simply
\[
\frac{\partial\rho}{\partial t}
\;+\;
\Gamma \cdot \nabla_{(q,p)} \rho
\;=\;
0.
\]
But $\Gamma \cdot \nabla_{(q,p)} \rho = \{\rho,\,H\}$. Thus
%\[
%\boxed{
\BoxedEquation{
\frac{\partial\rho}{\partial t}
\;+\;
\{\rho,\,H\}
\;=\;
0,
}
%}
%\]
which is the celebrated \emph{Liouville equation}. It asserts that the density $\rho$ is transported by the Hamiltonian flow without sources or sinks, a statement equivalent to the preservation of phase--space volume.

In some sense, this generalization of dynamics is somewhat counter intuitive—that deepening our \textit{mathematical} understanding of dynamics is necessary in order to describe scenarios where our \textit{physical} knowledge is diminished. It reflects a certain irony—propagating ignorance requires a surer grasp on the underlying formalism than certainty! This is—to be sure—a huge conceptual broadening. We are no longer thinking about the equation of motion for some set of point projectiles, but the dynamics of \textit{all} possible realizations of that system, with the initial condition reflecting our probabilistic guess at the system's initial state.  This change in perspective can sometimes inspire discomfit—are we not giving up something by moving to this statistical picture? The answer—emphatically—is \textit{no}. Deterministic dynamics is nothing but a special case of this more capacious perspective. To see this, it is worth emphasizing how the elements of this \textit{statistical picture} correspond to measurement. In particular, anything we can measure is described as an \textit{observable} $O(q,p)$, whose value will be a function of the phase space coordinates. The most obvious observables are the canonical coordinates themselves, $q$ and $p$. The \textit{expectation} of any observable at some time $t$ is given by:
\begin{equation}
    \expval{O}(t) = \int {\rm d} q {\rm d} p \ O(q,p) \rho(q,p,t). 
\end{equation}
In such a statistical picture, we are able to relax the ontological requirement of certainty. If in an experiment we find that the observable we measure has a distribution of values, we can rest safe in the knowledge that Hamiltonian dynamics can capture this scenario and indeed provide a prediction on how this distribution will evolve. And if it so happens that the dynamics of the observable appear deterministic at the finest precision available to us, the Liouville equation will still describe it. In this instance, however, the deterministic dynamics of observables will correspond to an initial distribution of the form 
\begin{equation}
    \rho(q,p,0)=\delta(q-q^0)\delta(p-p^0),
\end{equation}
where $\expval{O}(0) = O(q^0, p^0)$. It is relatively easy to use the definition of observable expectations to derive dynamical equations for them. In the case of a deterministic condition, it is easy to verify that the observables will follow precisely the appropriate Newtonian dynamics for a system with initial position $q^0$ and momentum $p^0$. 

The point, however, is that this scenario is a limiting case of our new, richer description of dynamics. We are perhaps beating a dead horse, but it bears repeating—this statistical description of dynamics \textit{generalizes} deterministic dynamics. This is, we believe, a vital point to catch in anticipation of the final step into quantum mechanics: to stress that, in their full generality, \textit{all} dynamics—quantum and classical—must be statistical in their nature. The simple reason is that a distribution over dynamical trajectories is necessarily more general than a single instance of one. 

Moreover, upgrading our object of study to a probability distribution gives us the necessary tools to understand a much richer range of dynamical behavior, ranging from chaos to thermodynamics. Indeed, one of the ongoing tasks of classical physics is to properly characterize how the spectacular variety of dynamical complexity we see around us emerges as a particular manifestation of the Hamiltonian dynamics undergirding it.

With this, our historical survey of dynamics has come to its natural conclusion. We have seen that a \textit{statistical} picture of dynamics emerges. The question now is how best to represent such a statistical description. The answer—which we will develop over the remainder of this book—lies in \textit{Hilbert space}.

\chapter{1. Mathematical background}\label{Chapter:1}
\section{The Dirac bra-ket notation}\label{Sec:DiracBraKetFormalism}

We begin with the basic notations, which are also given in more detail in Refs.~\cite{Shankar1994,Sakurai1994}.
%\cite[Sec. 1.2]{Sakurai1994}.
%\fixme{(See pages 23--36 in Ref. \cite{Sakurai1994}.)}
Let us review some fundamental definitions:

\emph{A vector space} $\mathcal{V} = \{ \ket{A}, \ket{B}, \ket{C}, \ldots\}$ is a set of elements called vectors and denoted by the ket $\ket{\quad}$ for which multiplication by a complex number and addition are defined. More formally, if $\ket{A}, \ket{B} \in \mathcal{V}$ then $\alpha \ket{A} + \beta \ket{B} \in \mathcal{V}$, for arbitrary complex numbers $\alpha$ and $\beta$ with the assumptions $0 \cdot \ket{A} = 0 \in \mathcal{V}$ and $\ket{B} + 0 = \ket{B}$. 

\emph{A Hilbert space} is a vector space endowed with an operation called \emph{the scalar product}, which takes as input any two vectors $\ket{A}$, $\ket{B}$ and returns a complex number denoted by $\langle A \ket{B}$ such that the following axioms are obeyed:
\begin{itemize}
	\item $\langle A \ket{B} = \langle B \ket{A}^*$, where ${}^*$ denotes a complex conjugation,
	\item $\bra{C} \Big( \alpha \ket{A} + \beta \ket{B} \Big) = \alpha \langle C \ket{A} + \beta \langle C \ket{B}$,
	\item $\Big( \alpha \bra{A} + \beta \bra{B} \Big)\ket{C} = \alpha \langle A \ket{C} + \beta \langle B \ket{C}$.
\end{itemize}

We also need to introduce the operation of \emph{hermitian conjugation} ($\dagger$) mapping ket into bra vectors $\ket{A}^{\dagger} = \bra{A}$ and vice versa $\bra{A}^{\dagger} = \ket{A}$, such that 
\begin{itemize}
	%\item \fixme{$(x^{\dagger})^{\dagger} = x$,}
	\item $\ket{A}^{\dagger} \ket{B} = \langle A \ket{B}$,
	\item $\Big( \alpha \ket{A} + \beta \ket{B} \Big)^{\dagger} = \alpha^* \bra{A} + \beta^* \bra{B}$,
	\item $\Big( \alpha \ket{A_1}\bra{B_1} + \beta \ket{A_2}\bra{B_2} \Big)^{\dagger} = \alpha^* \ket{B_1}\bra{A_1} + \beta^* \ket{B_2}\bra{A_2}$,
	\item $\bra{C} \hat{A} \ket{B}^* = \bra{B} \hat{A}^{\dagger} \ket{C}$.
\end{itemize}
While the scalar product $\langle A \ket{B}$ is also called the \emph{inner product}, the expression $\ket{A}\bra{B}$ is called the \emph{outer product}, since it takes two vectors and returns a linear operator.

There are two important examples of Hilbert spaces relevant for us. First, the Hilbert space of \emph{square-integrable} (the integral $\int d^n y \, |f(y_1, \ldots, y_n)|^2 $ is finite) complex-valued functions: Since we know how to add such functions and multiply by a complex number without making the just-mentioned integral diverge, the set of functions is a vector space. It becomes a Hilbert space by adding the scalar product $\langle f \ket{g} = \int d^n y \, f^*(y_1, \ldots, y_n) g(y_1, \ldots, y_n)$. This Hilbert space is used to represent states of a closed quantum and classical systems.
 
The second example of the Hilbert space is a set of linear operators acting in a Hilbert space with a \emph{finite trace} $\Tr\left( \hat{A}^{\dagger} \hat{A} \right)$ with the Hilbert-Schmidt scalar product $\langle A \ket{B} = \Tr\left( \hat{A}^{\dagger} \hat{B} \right)$. This construction plays a crucial role in the theory of open quantum systems (see, e.g., Sec.~\ref{Sec:UncertaintyPrincDensMatrix}).

%%%%%%%%%%%%%%%%%%%
%%% SECTION 1.2 %%%
%%%%%%%%%%%%%%%%%%%
\section{The Dirac delta-function: Conventional notations}

The Dirac delta function is widely used, but, unfortunately, is rarely understood in detail. Since we make heavy use of the Dirac bra-ket notation, we recall basic facts about the delta function. The first and foremost is that the Dirac function is neither Dirac's (it was introduced by Oliver Heaviside) \cite{jackson2008examples} nor a function. The delta function is in fact a \emph{functional}: it takes a whole function as an argument   and returns a number according to the rule
\begin{align}\label{FormalDefDiracDeltaFunc}
	\delta(x-x') : f(x) \to f(x'), \qquad f(x) \in C_0^{\infty},
\end{align} 
where $f(x) \in C_0^{\infty}$ denotes the requirement for the function $f(x)$ to be an infinitely differentiable and identically equal to zero along with all its derivatives after some finite value of the argument. The latter property is known as the finite support, implying that there is a fixed value $L$ such that $f^{(n)}(x)=0$ for all $|x| > L$ and $n=0,1,2,\ldots$. This will become a very useful property shortly as
\begin{align}\label{BoundaryCondF}
	f^{(n)}(\pm\infty)=0.
\end{align}

Arguably, the definition (\ref{FormalDefDiracDeltaFunc}) is not visually appealing, so let us slightly abuse the notation and rewrite Eq.~(\ref{FormalDefDiracDeltaFunc}) as
\begin{align}\label{DefDiracDeltaFunc}
	\int dx \, \delta(x-x') f(x) = f(x'). 
\end{align}
Note the integration $\int dx$ in Eq. (\ref{DefDiracDeltaFunc}) no longer means ``the area under the curve,'' it rather denotes the action of $\delta(x-x')$ onto the function $f(x)$. 

The $n^{\rm th}$ derivative of the delta function is defined as
\begin{align}
	\delta^{(n)} (x-x') : f(x) \to (-1)^n f^{(n)}(x'), \qquad f(x) \in C_0^{\infty},
\end{align}
or using the convention of Eq. (\ref{DefDiracDeltaFunc}),
\begin{align}\label{DefDerDiracDeltaFunc0}
	\int dx \, \delta^{(n)}(x-x') f(x) = (-1)^n f^{(n)}(x'). 
\end{align}

The definition (\ref{DefDerDiracDeltaFunc0}) in fact justifies such a recycling of the integration symbol. The definition (\ref{DefDerDiracDeltaFunc0}) can be obtained from (i.e., motivated by) Eq. (\ref{DefDiracDeltaFunc}) via the integration by parts 
\begin{align}
	\int dx \, \delta'(x-x') f(x) = \delta(x-x') f(x)\Big|_{-\infty}^{+\infty} - \int dx \, \delta(x-x') f'(x) = -f'(x'),
\end{align}
where we made use of Eq. (\ref{BoundaryCondF}). Equation (\ref{DefDerDiracDeltaFunc0}) can be obtained by performing the integration by parts $n$ times.

The following properties of the delta function follow from the definitions (for details see, e.g., Refs. \cite{Gelfand1964, Kanwal1998}):
\begin{align}
	& \int dx'' \, \delta^{(n)} (x-x'') \delta^{(m)} (x''-x') = \delta^{(n+m)} (x-x'), \label{DeltFunctProperty1} \\
	& \int \frac{dx}{2\pi} x^n e^{-ipx} = i^n \delta^{(n)} (p), \label{DeltFunctProperty2} \\
	& g(x) \delta (x-x') = g(x') \delta(x-x'), \label{DeltFunctPropert3} \\
	%& \delta^{(n)}(-x) = (-1)^n \delta^{(n)}(x), \label{DeltaFunctionParity} \\
	& g(x) \delta^{(n)} (x-x') = \sum_{r=0}^n (-1)^{n+r} \left( n \atop r \right) g^{(n-r)}(x') \delta^{(r)} (x-x'), 
    \qquad \left( n \atop r \right) = \frac{n!}{r! (n-r)!}, 
		\label{DeltFunctProperty4} \\
	& x^m \delta^{(n)} (x) = \left\{
		\begin{array}{ccc}
			(-1)^m \frac{n!}{(n-m)!} \delta^{(n-m)}(x) & \mbox{ if } & m < n, \\
			(-1)^n n! \delta (x)				& \mbox{ if } & m=n, \\
			0							& \mbox{ if } & m > n.
		\end{array} \right. \label{DeltFunctProperty5}
\end{align}
Here the function $g(x)$ needs to be an infinity differentiable only (i.e., $g(x) \in C^{\infty}$); it does not have to be of finite support  (i.e., it is not necessary that $g(x) \in C_0^{\infty}$, $C_0^{\infty} \subset C^{\infty}$). 

How should we understand the identities (\ref{DeltFunctProperty1})--(\ref{DeltFunctProperty5})? For example, consider Eq.~(\ref{DeltFunctPropert3}). This equality should be viewed as a shortcut of the following expanded expression
\begin{align}
	\int dx \, g(x) \delta (x-x') f(x) = \int dx \, g(x') \delta(x-x') f(x).
\end{align}
The latter is indeed an identity since the r.h.s. equals the l.h.s. under the definition~(\ref{DefDiracDeltaFunc})
\begin{align}
	\int dx \, g(x) \delta (x-x') f(x) = g(x') f(x')= \int dx \, g(x') \delta(x-x') f(x).
\end{align}
A special case of Eq.~(\ref{DeltFunctPropert3}) with $x'=0$ and $g(x)=x$ reads
\begin{align}\label{XTimesDeltaEquaslZero}
	x\delta(x) = 0.
\end{align}

It is crucial to emphasize that one should avoid multiplication of the delta function and its derivatives by a singular function (i.e., the smoothness of $g(x)$ in Eqs. (\ref{DeltFunctPropert3}) and (\ref{DeltFunctProperty4}) is crucial) in order to avoid paradoxes. In particular, consider the \emph{Schwartz counter example} 
\begin{align}
	\frac{1}{x} x \delta(x) = ?,
\end{align}
whose final result depends on where the parenthesis are drawn: on the one hand,   
\begin{align}
	\left( \frac{1}{x} x\right) \delta(x) = \delta(x);
\end{align}
on the other hand, using Eq. (\ref{XTimesDeltaEquaslZero})
\begin{align}
	 \frac{1}{x} \Big( x \delta(x) \Big) = 0. 
\end{align}

The Dirac delta function and its derivatives play an important role due to the following property: If $y(x)$ is a function zero everywhere except a point $x_0$, then
\begin{align}\label{MostGeneralFormofDistribution}
	y(x) = \sum_{n=0}^{\infty} c_n \delta^{(n)} (x-x_0).
\end{align}  
Furthermore, there are sequences of coefficients $c_n$ such that $\sum_{n=0}^{\infty} c_n \delta^{(n)} (x-x_0)$ converges to a smooth function (as you shall show in the homework). In other words, the set $\delta(x), \delta^{(1)} (x), \delta^{(2)} (x), \ldots $ can be used as \emph{a basis} for expanding a function $y(x)$, in a similar way, how monomials $1, x, x^2, \ldots$ are used as basis elements in the Taylor expansion. According to Eq. (\ref{DeltFunctProperty2}), the Fourier transform of $x^n$ equals  $\delta^{(n)} (x)$. Thus, the fact of $\{ \delta^{(n)} (x) \}$ being a basis is simply a consequence of the Fourier transform's inevitability.

%%%%%%%%%%%%%%%%%%%
%%% SECTION 1.3 %%%
%%%%%%%%%%%%%%%%%%%
\section{The Dirac delta-function: Bra-ket notations}

%The Dirac delta function is widely used, but unfortunately, rarely understood. Since we make heavy use of the Dirac bra-ket notation, we recall basic facts about the delta function. The first and foremost is that the Dirac function is neither  Dirac's \cite{jackson2008examples} (it was introduced by Oliver Heaviside) nor a function. 
{Using the stated fact that the delta function is a \emph{functional}, it can also be viewed as a \emph{bra vector}:} it takes a ket vector as an argument  and returns a number. Consider the basis $\ket{x}$, such that the following resolution of identity takes the form: $\hat{1} = \int dx \ket{x}\bra{x}$. The Dirac delta function and its derivatives $\langle \delta_{x'}^{(n)} |$ are such that for a vector $\ket{f}$ 
\BoxedEquation{\label{FormalDefDiracDeltaFunc1}
	\langle \delta_{x'}^{(n)} \ket{f} = (-1)^n f^{(n)}(x'), \qquad f(x) = \langle x \ket{f}, \quad n=0,1,\ldots.
}
Let us pick a vector $\ket{\omega}$ such that $\langle x \ket{\omega} = e^{i\omega x}$, then
\begin{align*}
	\langle \delta_{x'}^{(n)} \ket{\omega} = \int dx \, \langle \delta_{x'}^{(n)} \ket{x}\langle x \ket{\omega} 
		= \int dx \, \langle \delta_{x'}^{(n)} \ket{x} e^{i\omega x} = (-i\omega)^n e^{i\omega x'}.
\end{align*}
This gives us
\begin{align}
		\langle \delta_{x'}^{(n)} \ket{x} = \int \frac{d\omega}{2\pi} (-i\omega)^n e^{-i\omega (x-x')},
	~~\mbox{ or simply }~~ 
	 	\langle \delta_{x'}^{(n)} \ket{x} = \langle \delta_{0}^{(n)} \ket{x-x'}.
\end{align}
The latter identity suggests introducing the shorthand notation
$
	\delta(x-x') = \langle \delta_{x'}^{(0)} \ket{x}.
$
Furthermore,
\begin{align}
	\langle \delta_{x'}^{(n)} \ket{x} = \int \frac{d\omega}{2\pi} (-i\omega)^n e^{-i\omega (x-x')}
		= \frac{\partial^n}{\partial x^n} \int \frac{d\omega}{2\pi} e^{-i\omega (x-x')}
		= \frac{\partial^n}{\partial x^n} \langle \delta_{x'}^{(0)} \ket{x}, \notag
\end{align}
which gives
\BoxedEquation{ 
	\langle \delta_{x'}^{(n)} \ket{x} = \delta^{(n)}(x-x') = \int \frac{d\omega}{2\pi} (-i\omega)^n e^{-i\omega (x-x')}, \label{DeltFunctProperty2b} \\
	\int dx \, \delta^{(n)}(x-x') f(x) = (-1)^n f^{(n)}(x').  \label{DefDerDiracDeltaFunc}
}
Note also that
\begin{align}
	\delta(-x) = \int_{-\infty}^{+\infty} \frac{d\omega}{2\pi} e^{i\omega x}
		= -\int_{+\infty}^{-\infty} \frac{d\omega}{2\pi} e^{-i\omega x}
		=  \int_{-\infty}^{+\infty} \frac{d\omega}{2\pi} e^{-i\omega x} = \delta(x), \notag
\end{align}
thus,
\BoxedEquation{\label{DeltaFunctionParity}
	\delta^{(n)}(-x) = (-1)^n \delta^{(n)}(x). 
}
Now, let us select the ket vector $\ket{f}$ in Eq.~\eqref{FormalDefDiracDeltaFunc1} such that $\langle x \ket{f} = g(x)e^{i\omega x}$. Then, with help of Eq. (\ref{DeltFunctProperty2b}), we obtain
\begin{align}
	\int dx \, \delta^{(n)} (x-x') g(x) e^{i\omega x} = (-1)^n \left[ g(x') e^{i\omega x'} \right]^{(n)} 
		= (-1)^n \sum_{k=0}^n \left( n \atop k \right) g^{(n-k)}(x') (i\omega)^k e^{i\omega x'} \Longrightarrow \notag\\
	\delta^{(n)} (x-x') g(x) = (-1)^n \sum_{k=0}^n \left( n \atop k \right) g^{(n-k)}(x') \int \frac{d\omega}{2\pi}(i\omega)^k e^{-i\omega (x-x')}
    \nonumber
    \\
    = \sum_{k=0}^n(-1)^{n+k} \left( n \atop k \right) g^{(n-k)}(x')\delta^{(k)} (x-x')
    ,
\end{align}
which coincides with the property~\eqref{DeltFunctProperty4}.

%%%%%%%%%%%%%%%%%%%
%%% SECTION 1.4 %%%
%%%%%%%%%%%%%%%%%%%
\section{Elementary properties of self-adjoint operators}\label{Sec:Self_Adj_Op}

A linear operator $\hat{A}$ is self-adjoint if $\hat{A}^{\dagger} = \hat{A}$. Let us demonstrate that \emph{eigenvalues of a self-adjoint operators are real},
\begin{align}
	\hat{A} \ket{A} = A \ket{A}, \quad \bra{A} \hat{A}^{\dagger}  =  \bra{A} \hat{A} = \bra{A} A^* \Longrightarrow \bra{A} \hat{A} \ket{A} = A \langle A \ket{A} = A^* \langle A \ket{A} \Longrightarrow A = A^*
\end{align}
because $\langle A \ket{A}  \neq 0$.

Now, assuming that $\hat{A}$ has a continuous spectrum, we show that \emph{the eigenvectors of a self-adjoint operator $\hat{A}$ form a complete and orthogonal basis}. Given that
\begin{align*}
	\hat{A} \ket{A} = A \ket{A}, \quad \bra{A'} \hat{A} = \bra{A'} A',
\end{align*}
we have
\begin{align}
	\bra{A'} \hat{A} \ket{A} = A' \langle A' \ket{A} = A \langle A' \ket{A} \Longrightarrow
	\left(A' - A\right)  \langle A' \ket{A} = 0.
\end{align}
Whence, $ \langle A' \ket{A} = 0$ if $A' \neq A$ and $ \langle A \ket{A} \neq 0$. In other words, $g(A' - A) = \langle A' \ket{A} $ is a distribution zero everywhere except the origin.
Then according to Eq. (\ref{MostGeneralFormofDistribution}), we have
\begin{align}\label{DeltaFuncExpantionOfBraAKetA}
	\langle A' \ket{A} = \sum_{n=0}^{\infty} c_n \delta^{(n)} (A' - A).
\end{align}
Substituting Eq.~(\ref{DeltFunctProperty5}) into Eq.~(\ref{DeltaFuncExpantionOfBraAKetA}), one shows that $\langle A' \ket{A} = c_0 \delta (A' - A)$; hence, the orthogonality of the basis $\{ \ket{A} \}$ is established. The completeness means that $1 = \int dA \ket{A}\bra{A}$, thus 
\begin{align}
	\langle A' \ket{A''} = \int dA \langle A' \ket{A} \langle A \ket{A''},
\end{align}
Then, by using Eq. (\ref{DeltFunctProperty1}) with $n=m=0$, we obtain
\begin{align}
	c_0 \delta (A'-A'') = c_0^2 \int dA \delta (A' - A) \delta(A-A'') = c_0^2 \delta (A'-A''),
\end{align}
thus, $\langle A' \ket{A} = \delta (A' - A)$.

%%%%%%%%%%%%%%%%%%%
%%% SECTION 1.5 %%%
%%%%%%%%%%%%%%%%%%%
\section{What can be said about operators $\hat{A}$ and $\hat{B}$ If $[\hat{A}, \hat{B}] = i\kappa$?}\label{Sec_Realizations_of_HeisenbergAlg}

Surprisingly, a lot of information about self-adjoint operators $\hat{A}$ and $\hat{B}$ can be deduced from the fact that they obey the canonical commutation relation 
\begin{align}\label{AB_commut_relation}
	[\hat{A}, \hat{B}] = i\kappa, \qquad \kappa > 0.
\end{align}

We will find a convenient representation of the operators $\hat{A}$ and $\hat{B}$ assuming that they have purely continuous spectra and
\begin{align}
	&\hat{A} \ket{A} = A \ket{A}, \qquad \langle A' \ket{A} = \delta(A' - A), \qquad 1 = \int dA \ket{A}\bra{A}, \\
	&\hat{B} \ket{B} = B \ket{B}, \qquad \langle B' \ket{B} = \delta(B' - B), \qquad 1 = \int dB \ket{B}\bra{B}.
\end{align}

\begin{align}
	\bra{A'} [\hat{A}, \hat{B}] \ket{A} = (A' - A) \bra{A'} \hat{B} \ket{A} =  i\kappa \delta(A' - A), 
\end{align}
Using the delta function and its derivative as a basis [Eq. (\ref{MostGeneralFormofDistribution})], 
\begin{align}
	\bra{A'} \hat{B} \ket{A}  = \sum_{n=0}^{\infty} c_n \delta^{(n)}(A' - A),
\end{align}
we obtain (by dropping the term $n=0$)
\begin{align}
	\bra{A'} \hat{B} \ket{A} = -i\kappa \delta'(A' - A),
\end{align}
where  a special case of Eq. (\ref{DeltFunctProperty5}), $x\delta'(x)=-\delta(x)$, was used. In a similar fashion, one obtains 
\begin{align}
	\bra{B'} \hat{A} \ket{B} = i\kappa \delta'(B' - B).
\end{align}
Therefore, we can write
\begin{align}
	B \langle A  \ket{B} = \bra{A}\hat{B} \ket{B} = \int dA' \bra{A}\hat{B} \ket{A'} \langle A' \ket{B}
	= -i\kappa \int dA' \delta'(A - A') \langle A' \ket{B}
    \notag\\
    = i\kappa \int dA' \delta'(A' - A) \langle A' \ket{B}
	\Longrightarrow B \langle A  \ket{B} = -i\kappa \frac{\partial}{\partial A} \langle A  \ket{B},
\end{align}
which results in
\begin{align}
	\langle A  \ket{B} 
    = \mbox{const} \cdot e^{i AB / \kappa}.
\end{align}
Let us determine the constant
\begin{align}
	\delta(A - A') = \langle A \ket{A'} = \int dB \langle A \ket{B} \langle B \ket{A'} 
	= \mbox{const}^2 \int dB e^{i(A - A')B/\kappa} = \mbox{const}^2 \cdot 2\pi \kappa \delta(A - A').
\end{align}
Hence, we have established that
\begin{align}
	[\hat{A}, \hat{B}] = i\kappa \Longrightarrow \langle A  \ket{B} = \frac{1}{\sqrt{2\pi\kappa}} e^{i AB / \kappa}.
\end{align}

However, the converse is also true (i.e., we can draw $\Longleftrightarrow$ instead of $\Longrightarrow$)
\begin{align}
	\bra{A'} [\hat{A}, \hat{B}] \ket{A} = (A' - A) \bra{A'} \hat{B} \ket{A} = (A' - A) \int dB B \langle A' \ket{B} \langle B \ket{A} \notag\\
    = (A' - A) \kappa \int \frac{dB}{2\pi\kappa} \frac{B}{\kappa} e^{i(A' - A)B/\kappa} 
	= (A' - A) i\kappa \delta'(A - A')
    \notag\\
    = -i\kappa (A' - A) \delta'(A' - A)
    = i\kappa \delta(A' - A).
\end{align}
Consider an arbitrary vector $\ket{\psi}$, 
\begin{align}
	\bra{A} \hat{B} \ket{\psi} = \int dA' \bra{A} \hat{B} \ket{A'} \langle A' \ket{\psi} = -i\kappa \int dA'  \delta'(A - A') \langle A' \ket{\psi} 
    \nonumber
	\\
    = i\kappa \int dA'  \delta'(A' - A) \langle A' \ket{\psi} = -i\kappa \frac{\partial}{\partial A} \langle A \ket{\psi}, 
\end{align}
thus, $\psi(A) = \langle A \ket{\psi}$ the $\ket{\psi}$ state in the $A$-representation
\begin{align}
	\bra{A} \hat{B} \ket{\psi} = -i\kappa \frac{\partial}{\partial A} \psi(A)
	\Longleftrightarrow
	``\hat{A} = A, \quad \hat{B} = -i\kappa \frac{\partial}{\partial A}."
\end{align}
Similarly for $\psi(B) = \langle B \ket{\psi}$
\begin{align}
	\bra{B} \hat{A} \ket{\psi} = i\kappa \frac{\partial}{\partial B} \psi(B)
	\Longleftrightarrow
	``\hat{A} = i\kappa \frac{\partial}{\partial B}, \quad \hat{B} = B."
\end{align}

To summarize this section, we write
\begin{eqnbox}
    \vspace{-12pt}
    \begin{eqnarray}
        &&[\hat{A}, \hat{B}] = i\kappa 
        ~\Longleftrightarrow~ \langle A  \ket{B} 
        = \frac{1}{\sqrt{2\pi\kappa}} e^{i AB / \kappa} ~\Longleftrightarrow
        \nonumber
        \\
        &&\bra{A} \hat{B} \ket{\psi} = -i\kappa \frac{\partial}{\partial A} \langle A \ket{\psi}, \quad
	\bra{B} \hat{A} \ket{\psi} = i\kappa \frac{\partial}{\partial B} \langle B \ket{\psi}.\label{EqSummaryCommutatorSection}
    \end{eqnarray}
\end{eqnbox}
\section{Function of operators}

For our further development, we need to understand how to define a function of an operator. We will do it by drawing inspiration from well-known results for ordinary scalar functions. Let us tacitly assume that the function $f(x)$ is ``good'', that is, Taylor expandable. Then tautologically, the power series
\begin{align}
	f(x) = \sum_{n=0}^{\infty} \frac{f^{(n)}(0)}{n!} x^n
\end{align}
can be viewed as the definition of function $f(x)$. This is readily generalizable to the operator case,
\BoxedEquation{\label{FTaylorExpansion}
	f(\hat{A}) \coloneqq \sum_{n=0}^{\infty} \frac{f^{(n)}(0)}{n!} \hat{A}^n.
}
Indeed, Eq.~(\ref{FTaylorExpansion}) qualifies as the definition of an operator-valued function because it involves three known types of operations: product and sum of two operators and multiplication of an operator by a scalar. 

The following rule is an important consequence of Eq. \eqref{FTaylorExpansion} that will be used in numerical simulations 
\BoxedEquation{\label{EqFunctionActingOnEigenVals}
	\hat{A} \ket{A} = A \ket{A} \Longrightarrow f(\hat{A}) \ket{A} = f(A) \ket{A},
	\qquad
	\bra{A} \hat{A} = A \bra{A} \Longrightarrow \bra{A} f(\hat{A})  = f(A) \bra{A}.
}
This is derived as follows
\begin{align}
	f(\hat{A}) \ket{A} = \sum_{n=0}^{\infty} \frac{f^{(n)}(0)}{n!} \underbrace{\hat{A} \cdots \hat{A}\hat{A}}_{\mbox{$n$-times}} \ket{A}  
	= \sum_n \frac{f^{(n)}(0)}{n!} \underbrace{\hat{A} \cdots \hat{A}}_{\mbox{$(n-1)$-times}} A \ket{A} 
    \notag\\
	= \cdots = \sum_{n=0}^{\infty} \frac{f^{(n)}(0)}{n!} A^n \ket{A}	=  f(A) \ket{A}.
\end{align}

Likewise, recall the Cauchy Integral Formula, 
\begin{align}
	f(x) = \frac{1}{2\pi i} \oint_{\Gamma} \frac{f(z)dz}{z-x},
\end{align}
where the contour $\Gamma$ encloses the point $x$. This suggests another definition
\begin{align}\label{FCauchyExpansion}
	f(\hat{A}) \coloneqq \frac{1}{2\pi i} \oint_{\Gamma} dz f(z)(z\hat{1}-\hat{A})^{-1},
\end{align}
which involves the resolvent $(z\hat{1}-\hat{A})^{-1}$. In fact, both definitions (\ref{FTaylorExpansion}) and (\ref{FCauchyExpansion}) are identical. One can show it by performing the Taylor expansion of the resolvent and differentiating the Cauchy integral formula.

The Fourier transforming back and forth  does not change a sufficiently smooth function 
\begin{align}
	f(x) = \frac{1}{2\pi} \int d\omega \left( \int d\xi f(\xi) e^{i\xi\omega} \right) e^{-i\omega x}.
\end{align}
The function of an operator can be also defined as 
\begin{align}\label{FCauchyExpansionFourier}
	f(\hat{A}) \coloneqq \frac{1}{2\pi} \int d\omega d\xi \, f(\xi) e^{i\xi\omega} e^{-i\omega \hat{A}} 
	= \frac{1}{2\pi} \int d\omega d\xi \, f(\xi) e^{i\omega(\xi - \hat{A})} .
\end{align}
Thus, in order to calculate an arbitrary function of an operator you jus need to know how to calculate the exponential. It is great because  the Taylor expansion (\ref{FTaylorExpansion}) for $\exp$ always converges.

Let us establish
\BoxedEquation{\label{Eq_UfU_minus_equals_fUU}
	\hat{U} f( \hat{A} ) \hat{U}^{-1} = f\left( \hat{U}\hat{A}\hat{U}^{-1} \right).
}
From Eq. (\ref{FTaylorExpansion}), we get
\begin{align}
	\hat{U} f( \hat{A} ) \hat{U}^{-1} 
	&=\hat{U} \left( \sum_{n=0}^{\infty} \frac{f^{(n)}(0)}{n!} \hat{A}^n \right)  \hat{U}^{-1}
	= \sum_{n=0}^{\infty} \frac{f^{(n)}(0)}{n!} \hat{U} \hat{A}^n \hat{U}^{-1}
	= \sum_{n=0}^{\infty} \frac{f^{(n)}(0)}{n!} \hat{U} \underbrace{\hat{A}\hat{A} \cdots \hat{A}}_{\mbox{$n$-times}} \hat{U}^{-1} \notag\\
	&=\sum_{n=0}^{\infty} \frac{f^{(n)}(0)}{n!} \hat{U} \hat{A} \hat{U}^{-1}  \hat{U} \hat{A} \hat{U}^{-1} \cdots \hat{U} \hat{A} \hat{U}^{-1} 
	=\sum_{n=0}^{\infty} \frac{f^{(n)}(0)}{n!} \left( \hat{U} \hat{A} \hat{U}^{-1} \right)^n
	= f\left( \hat{U}\hat{A}\hat{U}^{-1} \right).
\end{align}

Equation~(\ref{Eq_UfU_minus_equals_fUU}) suggests a simple algorithm for calculating $f(\hat{A})$ if $\hat{A}$ is a diagonalizable matrix, i.e., if there is a representation $\hat{A} = \hat{U} \hat{\Lambda} \hat{U}^{-1}$, where $\hat{\Lambda} = \diag(\lambda_1, \lambda_2, \ldots)$ is a diagonal matrix with the eigenvalues $\lambda_1, \lambda_2, \ldots$ occupying the main diagonal. Since a direct computation of matrix products establishes that $[\diag(\lambda_1, \lambda_2, ...)]^n = \diag(\lambda_1^n, \lambda_2^n, ...)$, we finally get from Eq.~(\ref{Eq_UfU_minus_equals_fUU}):
\begin{align}
	f( \hat{A} ) = \hat{U} \diag\left( f( \lambda_1), f( \lambda_2), \ldots \right) \hat{U}^{-1}. 
\end{align}

Let us study the numerical methods used to calculate the function of a matrix. For that we shall analyze the source code of \texttt{Scipy}\footnote{\url{https://github.com/scipy/scipy/blob/9da60e120ea0c967d77f93bbd2bb6b4fb5ba2944/scipy/linalg/matfuncs.py}}. 
In particular, calculation of the matrix exponential is of prime importance for us, see excellent Ref.~\cite{Moler2003} reviewing nineteen ways of calculating the matrix exponent. We will consider the matrix exponential via the eigenvalue decomposition (\texttt{expm2} in \texttt{Scipy}), Taylor expansion (\texttt{expm3} in \texttt{Scipy}), and Pad\'{e} approximation (\texttt{expm} in \texttt{Scipy}).

%%% SECTION %%% 
\section{Pad\'{e} approximation for efficient matrix exponent calculation}

Assume a fractional approximation
\begin{align}\label{EqGeneralPade}
	e^{\hat{A}} \approx  \left( \sum_{k=0}^{K_b} b_k \hat{A}^k \right) \left( \sum_{k=0}^{K_c} c_k \hat{A}^k \right)^{-1},
\end{align}
which is valid for all matrices $\hat{A}$ fulfilling condition for the norm
\begin{align}\label{EqPadeApplicabilityCondition}
	\| \hat{A} \| \leq N_{\max}.
\end{align}
This is known as the Pad\'{e} approximation for the matrix exponential (see, e.g. Ref.~\cite{Moler2003}) and it is implemented in the number of modern numerical packages (e.g., \texttt{expm} in \texttt{Scipy}).

Here, let us discuss the algorithm, which allows us to utilize the expression~(\ref{EqGeneralPade}), \emph{even if the condition (\ref{EqPadeApplicabilityCondition}) is violated}:
\begin{enumerate}
	\item Scale the matrix if $\| \hat{A} \| > N_{\max}$. Calculate $K$ as the integer ceiling of $\log_2 \frac{\| \hat{A} \|}{N_{\max}}$, then $\hat{B} \coloneqq \hat{A} / 2^K$. Note that 
	$$
		\| \hat{B} \| = \| \hat{A} / 2^K \| = \| \hat{A} \| / 2^K,
        \mbox{ thus }\quad
			\| \hat{B} \| \leq \| \hat{A} \| / 2^{\log_2 (\| \hat{A} \| / N_{\max})} = N_{\max}.
	$$
	Hence, the scaled matrix $\hat{B}$ obeys the condition (\ref{EqPadeApplicabilityCondition}).
	\item Use the approximation (\ref{EqGeneralPade}) to get $e^{\hat{B}}$.
	\item Finally, $e^{\hat{A}} = \left( e^{\hat{A} / 2^K} \right)^{ 2^K} = \left( e^{\hat{B}} \right)^{ 2^K} = \hat{\Pi}_K$, where $\hat{\Pi}_n$ must be defined recursively as $\hat{\Pi}_0 =  e^{\hat{B}}$, $\hat{\Pi}_{n+1} = \left( \hat{\Pi}_n \right)^2$, $n=1,2,\ldots$. Note that only $K$ matrix multiplications are required to recover $e^{\hat{A}}$ from $e^{\hat{B}}$.
\end{enumerate}
Summarizing this, the Pad\'{e} approximation (\ref{EqGeneralPade}) requires in total $K + \max(K_b, K_c)$ matrix multiplications. In other words, it takes $O\left(\log\| \hat{A} \| \right)$ of matrix multiplication as $\| \hat{A} \| \to \infty$. 

Let us compare this approach with the straightforward  Taylor expansion,
\begin{align}\label{EqTaylorExp}
	e^{\hat{A}} = \sum_{k=0}^{\infty} \frac{\hat{A}^k}{ k! }.
\end{align}
For numerical calculations, the summation should be terminated when the term is smaller or equal to some tolerance value $\varepsilon$ determining the accuracy of the final result. Thus, we find $K$ [indicating the number of matrix multiplications necessary for computations via Eq. (\ref{EqTaylorExp})] such that 
\begin{align}
	\varepsilon \sim \| \hat{A}^K \| / K! \leq  \| \hat{A} \|^K / K!,  \notag
\end{align}
in order to get a rough estimate for $K$
\begin{align}
	\varepsilon \sim \| \hat{A} \|^K / K! \Longrightarrow \| \hat{A} \| \sim (\varepsilon K!)^{1/K} 
	\sim K / \exp(1), \quad K \to \infty.
\end{align}
Therefore, the Taylor expansion requires $O\left(\| \hat{A} \| \right)$ matrix multiplications as $\| \hat{A} \| \to \infty$, \emph{exponentially} more than the Pad\'{e} approximation.

\paragraph{The Baker-Campbell-Hausdorff formula.} 
Since we started working with the matrix exponential in this section, for completeness let us also recall the Baker-Campbell-Hausdorff formula, which will be helpful below,
\BoxedEquation{\label{EqBCH}
	& \exp(\hat{A}) \exp(\hat{B}) = \exp(\hat{C}_1 + \hat{C}_2 + \hat{C}_3 + \cdots ), \notag\\
	& \hat{C}_1 = \hat{A} + \hat{B}, \quad \hat{C}_2 = \frac{1}{2} \left[\hat{A}, \hat{B}\right], \quad
	\hat{C}_3 = \frac{1}{12} \left\{ \left[ \hat{A}, [\hat{A}, \hat{B}] \right] + \left[ \hat{B}, [\hat{B}, \hat{A}] \right] \right\}, \quad \ldots
}

%%% SECTION %%% 
\section{Noncommutative analysis: The Weyl calculus}\label{Sec_Weyl_calculus}

Noncommutative analysis \cite{Weyl1950, Feynman1951, Nelson1970, Taylor1973, Maslov1976, Karasev1979, Nazaikinskii1992, Nazaikinskii1996, Madore2000, Andersson2004, Jefferies2004, Rosas2004, Nielsen2005, Muller2007} is a broad and active field of mathematics with a number of important applications. This branch of analysis aims at identifying functions of noncommutative variables and specifying operations with such objects. There are many ways of introducing functions of operators; however, the choice of a particular definition is a matter of convenience \cite{Nazaikinskii1992}. 

To make the paper self-consistent, we shall review basic results from the Weyl calculus, which is a popular version of noncommuting analysis. Equation (\ref{Weyl_derivative_of_function}) plays a crucial role for us. Even though we prove this result within the Weyl calculus, it is valid in more general settings (see, e.g., Ref. \cite{Maslov1976} and page 63 of Ref. \cite{Nazaikinskii1996}).

The starting point is the well known fact that Fourier transforming back and forth  does not change a sufficiently smooth function of $n$-arguments,
\begin{align}\label{Fourier_identity}
	f(\lambda_1, \ldots, \lambda_n) = \frac{1}{(2\pi)^n} \int \prod_{l=1}^n d\xi_l d\eta_l 
		\exp\left[ i \sum_{q=1}^n \eta_q(\lambda_q - \xi_q) \right] f(\xi_1, \ldots, \xi_n).
\end{align}
Following this observation, we define the function of noncommuting operators within the Weyl calculus as 
\BoxedEquation{\label{Weyl_main_definition}
	f(\hat{A}_1, \ldots, \hat{A}_n) \coloneqq \frac{1}{(2\pi)^n} \int \prod_{l=1}^n d\xi_l d\eta_l
		\exp\left[ i \sum_{q=1}^n \eta_q(\hat{A}_q - \xi_q) \right] f(\xi_1, \ldots, \xi_n),
}
where the exponential of an operator is specified by the Taylor expansion (\ref{EqTaylorExp}). 
From the relation
\begin{align}
	f^{\dagger} (\hat{A}_1, \ldots, \hat{A}_n) &= \frac{1}{(2\pi)^n} \int_{-\infty}^{\infty} \prod_{l=1}^n d\xi_l d\eta_l \exp\left[ -i \sum_{q=1}^n \eta_q(\hat{A}_q^{\dagger} - \xi_q) \right] f(\xi_1, \ldots, \xi_n) \notag\\
	&= \frac{1}{(2\pi)^n} \int_{\infty}^{-\infty} \prod_{l=1}^n d\xi_l d(-\eta_l) \exp\left[ i \sum_{q=1}^n \eta_q(\hat{A}_q^{\dagger} - \xi_q) \right] f(\xi_1, \ldots, \xi_n) \notag\\
	&= \frac{1}{(2\pi)^n} \int_{-\infty}^{\infty} \prod_{l=1}^n d\xi_l d\eta_l \exp\left[ i \sum_{q=1}^n \eta_q(\hat{A}_q^{\dagger} - \xi_q) \right] f(\xi_1, \ldots, \xi_n) \notag
\end{align}
we see that the identity
\BoxedEquation{
	f^{\dagger} (\hat{A}_1, \ldots, \hat{A}_n) = f(\hat{A}_1^{\dagger}, \ldots, \hat{A}_n^{\dagger})
}
implies that the function of self-adjoint operators (\ref{Weyl_main_definition}) is itself a self-adjoint operator. Moreover, one may demonstrate that  

\begin{align}
	&f'_{\hat{A}_k} (\hat{A}_1, \ldots, \hat{A}_n)  \coloneqq \lim_{\epsilon\to 0} \frac 1{\epsilon} 
		\left[ f(\hat{A}_1, \ldots, \hat{A}_k + \epsilon, \ldots, \hat{A}_n)  - f(\hat{A}_1, \ldots, \hat{A}_k, \ldots, \hat{A}_n)\right]
        \notag\\
		&=\lim_{\epsilon\to 0} \frac{1}{\epsilon (2\pi)^n} \int \prod_{l=1}^n d\xi_l d\eta_l \,  
			\left( \exp\left[ i \sum_{q=1}^n \eta_q(\hat{A}_q + \epsilon \delta_{q,k} - \xi_q) \right] 
			- \exp\left[ i \sum_{q=1}^n \eta_q(\hat{A}_q - \xi_q) \right] \right) 
        \notag\\
		&\times f(\xi_1, \ldots, \xi_n) 
        = \frac{1}{(2\pi)^n} \int \prod_{l=1}^n d\xi_l d\eta_l  \lim_{\epsilon\to 0} \frac{e^{i\eta_k \epsilon} - 1}{\epsilon} \exp\left[ i \sum_{q=1}^n \eta_q(\hat{A}_q - \xi_q) \right] f(\xi_1, \ldots, \xi_n) \notag\\
		&= \frac{1}{(2\pi)^n} \int \prod_{l=1}^n d\xi_l d\eta_l \, i\eta_k 
			\exp\left[ i \sum_{q=1}^n \eta_q(\hat{A}_q - \xi_q) \right] f(\xi_1, \ldots, \xi_n) \label{EqPreDiracTheorem}\\
		&= \frac{1}{(2\pi)^n} \int \prod_{l=1}^n d\xi_l d\eta_l
			f(\xi_1, \ldots, \xi_n) \left( -\frac{\partial}{\partial \xi_k} \right) \exp\left[ i \sum_{q=1}^n \eta_q(\hat{A}_q - \xi_q) \right]  \notag
\end{align}

Therefore, we have
\BoxedEquation{\label{Weyl_derivative_of_function}
	f'_{\hat{A}_k} (\hat{A}_1, \ldots, \hat{A}_n) &= \frac{1}{(2\pi)^n} \int \prod_{l=1}^n d\xi_l d\eta_l
			\exp\left[ i \sum_{q=1}^n \eta_q(\hat{A}_q - \xi_q) \right] f'_{\xi_k} (\xi_1, \ldots, \xi_n).
}
Equation (\ref{Weyl_main_definition}) defines a one-to-one mapping between a function $f(\xi_1, \ldots, \xi_n)$ and a linear operator $f(\hat{A}_1, \ldots, \hat{A}_n)$. By the same token, Eq.~(\ref{Weyl_derivative_of_function}) establishes a one-to-one mapping between the derivative of a function and the derivative of a linear operator.

Let us find another representation of the derivative from Eq.~(\ref{EqPreDiracTheorem}). To this end, we start from the following operator identity
\BoxedEquation{
	[\hat{A}_1 \hat{A}_2, \hat{B} ] = [\hat{A}_1, \hat{B}] \hat{A}_2 + \hat{A}_1 [\hat{A}_2, \hat{B}],
}
which comes from the relation
\begin{align}\notag
	\hat{A}_1 \hat{A}_2 \hat{B} - \hat{B} \hat{A}_1 \hat{A}_2 
	= \hat{A}_1 \hat{B} \hat{A}_2 - \hat{B} \hat{A}_1 \hat{A}_2
	+ \hat{A}_1 \hat{A}_2 \hat{B} - \hat{A}_1 \hat{B} \hat{A}_2. 
\end{align}
Its generalization can be obtained via the mathematical induction.
Since
\begin{align}
	[ \hat{A}_1 \cdots \hat{A}_{n+1}, \hat{B} ]  &= [ (\hat{A}_1 \cdots \hat{A}_n ) \hat{A}_{n+1}, \hat{B} ] 
	= [ \hat{A}_1 \cdots \hat{A}_n, \hat{B} ] \hat{A}_{n+1} + \hat{A}_1 \cdots \hat{A}_n [\hat{A}_{n+1}, \hat{B}] \notag\\
	&= \sum_{k=1}^n \hat{A}_1 \cdots \hat{A}_{k-1} \commut{ \hat{A}_k, \hat{B} } \hat{A}_{k+1} \cdots \hat{A}_n \hat{A}_{n+1} + \hat{A}_1 \cdots \hat{A}_n [\hat{A}_{n+1}, \hat{B}] \notag\\
	&= \sum_{k=1}^{n+1} \hat{A}_1 \cdots \hat{A}_{k-1} \commut{ \hat{A}_k, \hat{B} } \hat{A}_{k+1} \cdots \hat{A}_{n+1},
\end{align}
we obtain
\BoxedEquation{\label{EqLeibnizCommutatorRule}
	[ \hat{A}_1 \cdots \hat{A}_n, \hat{B} ]  &= 
	\sum_{k=1}^n \hat{A}_1 \cdots \hat{A}_{k-1} \commut{ \hat{A}_k, \hat{B} } \hat{A}_{k+1} \cdots \hat{A}_n.
}

Introducing a simplifying attention $\hat{A} \coloneqq i \sum_{q=1}^n \eta_q(\hat{A}_q - \xi_q)$. Assume there exist operators $\hat{B}_k$ such that $[\hat{A}_l, \hat{B}_k] = \delta_{l,k}$, then $[\hat{A}, \hat{B}_k] = i\eta_k$.
\begin{align}\notag
	[e^{\hat{A}}, \hat{B}_k] = \sum_{n=0}^{\infty} \frac{1}{n!} [ \underbrace{\hat{A} \cdots \hat{A}}_{\mbox{$n$-times}}, \hat{B}_k] 
	= \sum_{n=1}^{\infty} \frac{1}{n!}  i\eta_k n\underbrace{\hat{A} \cdots \hat{A}}_{\mbox{$(n-1)$-times}}
	= i\eta_k \sum_{n=1}^{\infty} \frac{\hat{A}^{n-1}}{(n-1)!} 
	=  i\eta_k e^{\hat{A}}.
\end{align}
\begin{align}\notag
	f'_{\hat{A}_k} (\hat{A}_1, \ldots, \hat{A}_n) = \frac{1}{(2\pi)^n} \int \prod_{l=1}^n d\xi_l d\eta_l \, i\eta_k 
		e^{\hat{A}} f(\xi_1, \ldots, \xi_n)
	\notag\\
    = \frac{1}{(2\pi)^n} \int \prod_{l=1}^n d\xi_l d\eta_l \, [e^{\hat{A}}, \hat{B}_k] f(\xi_1, \ldots, \xi_n).
\end{align}
Hence, we conclude that
\BoxedEquation{\label{EqDiracTheorem}
	f'_{\hat{A}_k} (\hat{A}_1, \ldots, \hat{A}_n) = [ f(\hat{A}_1, \ldots, \hat{A}_n), \hat{B}_k].
}

Now, let us take $\hat{B}$ such that $[\hat{A}_l, \hat{B}] = c_l$ is a constant. Then $\hat{B} = \sum_k c_k \hat{B}_k$ and from Eq.~(\ref{EqDiracTheorem}), we obtain
\begin{align}
	[ f(\hat{A}_1, \ldots, \hat{A}_n), \hat{B}] = \sum_k c_k [ f(\hat{A}_1, \ldots, \hat{A}_n), \hat{B}_k]. 
\end{align}
Therefore, we arrive at
\BoxedEquation{\label{EqWelCommutatorTheorem} 
	[\hat{A}_k, \hat{B}] = \mbox{const} \quad \Longrightarrow \quad
	[ f(\hat{A}_1, \ldots, \hat{A}_n), \hat{B}] = \sum_k [\hat{A}_k, \hat{B}] f'_{\hat{A}_k} (\hat{A}_1, \ldots, \hat{A}_n),
}
which is of fundamental importance.

In the context of Maslov calculus \cite{Maslov1976, Nazaikinskii1992, Nazaikinskii1996, Transtrum2005}, Eq. (\ref{EqWelCommutatorTheorem}) has been extended to a more general case of the  commutators of the type $\commut{ f(\hat{A}_1,\ldots, \hat{A}_n), g(\hat{B}_1,\ldots, \hat{B}_n) }$.

We finally note that it follows from Eqs. (\ref{Weyl_main_definition}) and (\ref{Eq_UfU_minus_equals_fUU}) that
\BoxedEquation{\label{Eq_U_Weyl_U_inv}
	\hat{U} f\left(\hat{A}_1, \ldots, \hat{A}_n\right) \hat{U}^{-1}
	= f\left(\hat{U} \hat{A}_1\hat{U}^{-1}, \ldots, \hat{U} \hat{A}_n\hat{U}^{-1} \right).
}

\newpage\thispagestyle{empty}\blankpage

\chapter{2. Kinematic description}\label{Chapter:2}

Generalizing Schwinger's motto, ``quantum mechanics: symbolism of atomic measurements''~\cite{Schwinger2003}, we assert that any physical model is a symbolic representation of the experimental evidence supporting it (see Fig.~\ref{Fig_House_ODM}). 
\begin{figure*}
	\begin{center}
		\includegraphics[width=0.5\hsize]{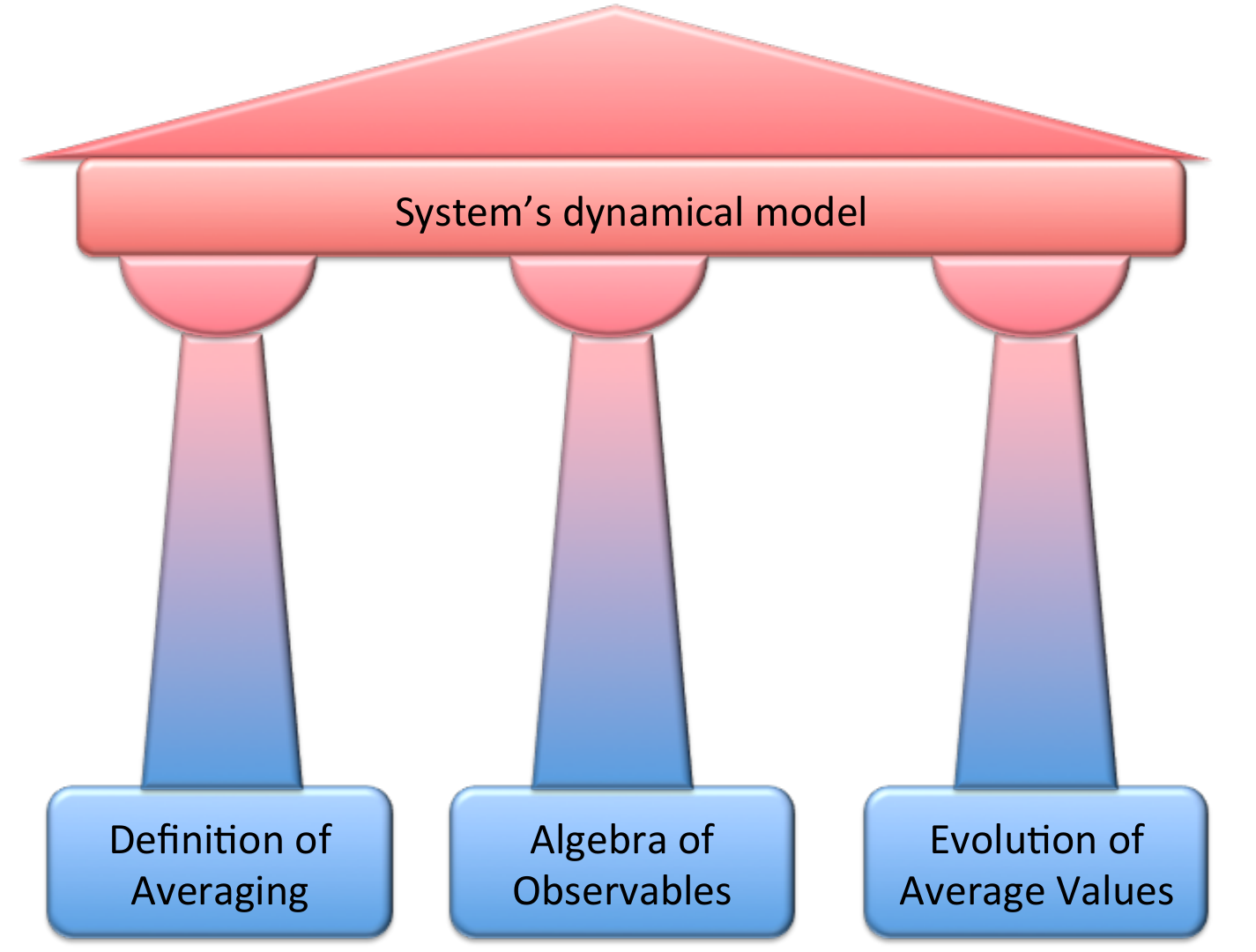}
		\caption{Observation-driven theory generating framework. Key idea: Any physical model (theory) is symbolic representation of experimental data.}\label{Fig_House_ODM}
	\end{center}
\end{figure*}
The mathematical symbolism for this purpose needs to be considered. A formalism specialized in describing a specific class of behavior (e.g., classical mechanics expressed in terms of phase-space trajectories) can be effective, but it may be unsuitable for connecting different classes of phenomena (e.g.,  unifying quantum and classical mechanics). In this case, a general and versatile formalism is preferred. Building a formalism around Hilbert space is a suitable candidate for this role. Hilbert space is well understood, rich in mathematical structure, and convenient for practical computations. 	

Consider the following postulates:
\begin{enumerate}
 \item The states of a system are represented by normalized vectors $| \Psi \rangle$ of a complex Hilbert space, and the observables are given by self-adjoint operators (see Sec.~\ref{Sec:Self_Adj_Op}) acting on this space;
\item The expectation value of a measurable $\hat{A}$ at time $t$ is $\overline{A}(t) =  \langle \Psi (t) | \hat{A} | \Psi(t) \rangle$; 
\item The probability that a measurement of an observable $\hat{A}$ at time $t$ yields $A$ is $\left|\langle A | \Psi(t) \rangle \right|^2$, where $\hat{A} | A \rangle = A | A \rangle$;
\item The state space of a composite system is the tensor product of the subsystems' state spaces. 
\end{enumerate}

 Having accepted these postulates, the rest—state spaces, observables, and the equations of motion—can be deduced directly from observable data. Importantly, these axioms are just the well-known quantum mechanical postulates with the adjective ``quantum'' {\it removed, since $| \Psi \rangle$ is a general state encompassing classical and quantum behavior.} We will demonstrate below that these postulates are sufficient to capture all the features of both quantum {\it and} classical mechanics, as well as the associated hybrid mechanics.

%%% SECTION %%% 
\section{Classical kinematics} \label{Sec_Class_Kinematics}

Without loss of generality, we consider only one-dimensional systems throughout.
Assume a system under investigation, hidden in a black box, that can be subjected only to the 
measurement of the (one-dimensional) coordinate and momentum (see Fig. \ref{Fig_Measurement2Commutators}). 
%%% 
\begin{figure*}
	\begin{center}
		\includegraphics[width=\textwidth]{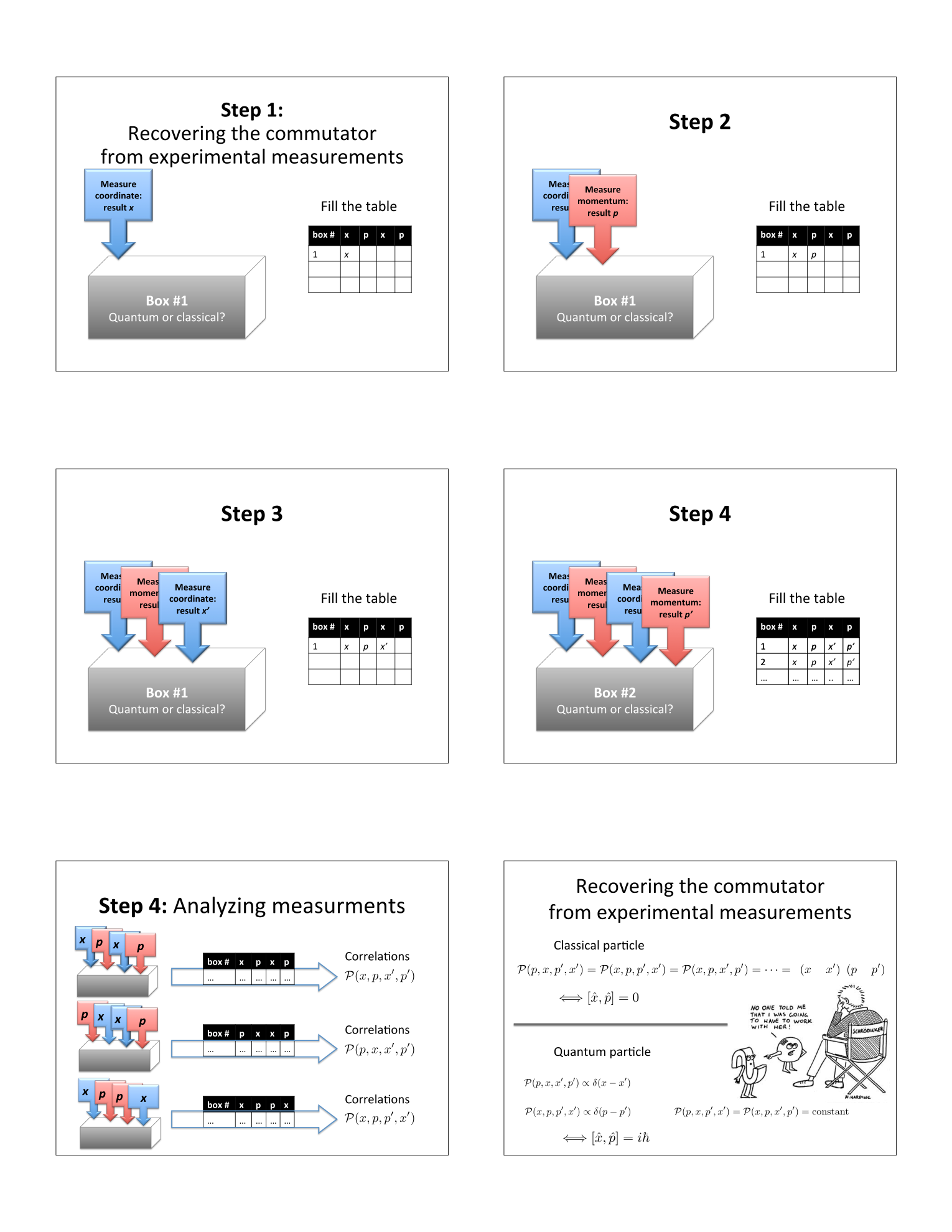}
		\caption{A thought experiment for recovering the commutator from experimental measurements}\label{Fig_Measurement2Commutators}
	\end{center}
\end{figure*}
%%% 
Let $\mathcal{P}(p, x, p', x')$ denote the probability that the outcome of a sequence of measurements: momentum-coordinate-momentum-coordinate are $p$, $x$, $p'$, $x'$, respectively. By permuting the order of the measurements, one can define the probabilities $\mathcal{P}(x, p, p', x')$, $\mathcal{P}(x, p, x', p')$, etc. All these measurements are assumed to be performed instantaneously and without any delay between measurements\footnote{
	Obviously, there is a finite time delay between any two measurements in a real experiment. The adjective ``instantaneously'' implies that such a delay should be beyond the accuracy of the experiments. However, an experimentalist must assure that the correct order of measurements is preserved.
}.
Having performed an ensemble of the introduced measurements for a free classical particle, we obtain
\BoxedEquation{\label{Combined_ClassicalMeasurments_Results}
	\mathcal{P}(p, x, p', x') = \mathcal{P}(x, p, p', x') = \mathcal{P}(x, p, x', p') 
	 = \cdots = \delta (x-x') \delta(p-p').
}
Equation (\ref{Combined_ClassicalMeasurments_Results}) implies that the order of the measurements does not affect the outcome of the measurements. This property is taken as an operational definition of a classical particle. 

%Following the ideology put forth in Sec.~\ref{Sec_Ideology}, 
Next, we shall somehow represent the condition~(\ref{Combined_ClassicalMeasurments_Results}) in the Hilbert space. Let $\hat{x}$ and $\hat{p}$ be self-adjoint operators that represent the observables of coordinate and momentum.  In such a language, Eq.~(\ref{Combined_ClassicalMeasurments_Results}) can be derived from the axiom that the operators $\hat{x}$ and $\hat{p}$ have the same eigenfunctions,
\begin{align}\label{Definition_of_classicalXP}
	\hat{x}  | x \, p \rangle = x | x \, p \rangle, \qquad \hat{p}  | x \, p \rangle = p | x \, p \rangle,
\end{align}
or from the equivalent axiom that the classical momentum and coordinate commute, 
\BoxedEquation{\label{Ch2_ClassicalCommutator}
	[\hat{p}, \hat{x}]_{\textrm{class.}} = 0.
}
The equality $\langle x \, p | x' \, p' \rangle =  \delta (x-x') \delta(p-p')$ follows from the fact that $\hat{x}$ and $\hat{p}$ are self-adjoint operators. Using Born's rule, we obtain 
\begin{align}
	\mathcal{P}(p, x, p', x') = \left| \langle x \, p | x' \, p' \rangle  \right|^2 =  \delta (x-x') \delta(p-p'),
\end{align}
where the order of momenta and coordinates is irrelevant. Therefore, we indeed reached Eq.~(\ref{Combined_ClassicalMeasurments_Results}).

Note that Eq.~(\ref{Ch2_ClassicalCommutator}) is a condensate of all classical kinematics, and it is crucial for our subsequent discussion.

%%% SECTION %%% 
\section{Quantum kinematics}\label{Sec_Quantum_Kinematics}

Following the consideration in Sec.~\ref{Sec_Class_Kinematics}, we suppose that the black box contains \emph{a quantum particle}. Then, we observe that
\BoxedEquation{\label{Combined_QuantumlMeasurments_Results}
	\mathcal{P}(p,x,x',p') \propto \delta(x-x'), \quad
	\mathcal{P}(x,p,p',x') \propto \delta(p-p'), 
    \notag\\
	\mathcal{P}(p,x,p',x') = \mathcal{P}(x,p,x',p') = \mbox{const}.
}
Contrary to the classical case [see Eq.~(\ref{Combined_ClassicalMeasurments_Results})], the order of measurements does matter. It is a hallmark of quantum mechanics, and it can be taken as an operational definition of the quantum case. 

Collecting the outcomes of the sequence of two measurements rather than four, we conclude that
\begin{align}\label{P_TwoMeasured_Quantum}
	\mathcal{P}(p,p') = \delta(p-p'), \quad \mathcal{P}(x,x') = \delta(x-x'), 
	\quad \mathcal{P}(x,p) = \mbox{const}.
\end{align} 
To embed these relations into the Hilbert space formulation of quantum mechanics, we introduce self-adjoint operators of the coordinate $\hat{\bs{x}}$ and momentum $\hat{\bs{p}}$ with distinct eigenfunctions, 
\begin{align}
	\hat{\bs{x}} | \bs{x} \rangle = \bs{x} | \bs{x} \rangle, \quad \hat{\bs{p}} | \bs{p} \rangle = \bs{p} | \bs{p} \rangle, \quad 
    \int d\bs{x} | \bs{x} \rangle \langle \bs{x}| 
    = \int d\bs{p} | \bs{p} \rangle \langle \bs{p}| = 1.
\end{align} 
The first two equalities in Eq.~(\ref{P_TwoMeasured_Quantum}) follow from the self-adjointness of $\hat{\bs{x}}$ and $\hat{\bs{p}}$. The third condition from Eq.~\eqref{P_TwoMeasured_Quantum} is equivalent to
\begin{align}
	\left| \langle \bs{x} | \bs{p} \rangle \right|^2 = \mbox{const}. \label{quantum_xp_const}
\end{align}

In what follows, we use refined arguments originally put forth in Ref. \cite{Bondar2011a}. From Eq. (\ref{quantum_xp_const}), one readily concludes that
\begin{equation}\label{x_p_product_general}
    \langle \bs{x} | \bs{p} \rangle = C \exp[i f(\bs{x},\bs{p})],
\end{equation}
where $C$ is a real constant and $f(\bs{x},\bs{p})$ is an analytic real-valued function. The integral equation for the unknown function $f(\bs{x},\bs{p})$ can be obtained as 
\begin{equation}\label{IntegralEq_for_f_general}
	\delta (\bs{x}-\bs{x}') = \!\int\! d\bs{p} \, \langle \bs{x} | \bs{p} \rangle \langle \bs{p} | \bs{x}' \rangle = C^2 \!\!\int\! d\bs{p} \, e^{if(\bs{x},\bs{p}) - if(\bs{x}',\bs{p})},
\end{equation}
In the case of $\bs{x}\neq \bs{x}'$ we have
\begin{equation}
\label{IntegralEq_for_f_special}
	\int\! d\bs{p} \, e^{if(\bs{x},\bs{p}) - if(\bs{x}',\bs{p})} = 0 \qquad (\bs{x}\neq \bs{x}').
\end{equation}
Because $f(\bs{x},\bs{p})$ is analytic, we can expand it 
\begin{equation}\label{ExpansionF}	
	f(\bs{x},\bs{p}) = f_1(\bs{x}) + g_1(\bs{x})\bs{p} + g_2(\bs{x})\bs{p}^2 + g_3(\bs{x})\bs{p}^3 + \ldots
\end{equation}
If only the first two terms are kept, we have
\begin{align}
	e^{i f_1(\bs{x}) - i f_1(\bs{x}')} \int\! d\bs{p} \, e^{ i[g_1(\bs{x}) - g_1(\bs{x}')]\bs{p}} 
	= 2\pi \delta( g_1(\bs{x}) - g_1(\bs{x}')) e^{i f_1(\bs{x}) - i f_1(\bs{x}')},
\end{align}
which satisfies Eq.~(\ref{IntegralEq_for_f_special}). If we truncate the expansion (\ref{ExpansionF}) after the $n$th term (for arbitrary $n \geqslant 2$), we obtain the integral
\begin{equation}\nonumber
	\int\! d\bs{p} \, \exp\left( i\sum_{k=1}^n g_k \bs{p}^k \right).
\end{equation}
This type of integral is known as a diffraction integral (see, e.g., Ref. \cite{Gilmore1981}) and does not satisfy Eq.~(\ref{IntegralEq_for_f_special}) in general. Thus we assume 
$
	f(\bs{x},\bs{p}) = f_1(\bs{x}) + f_2(p) + g_1(\bs{x})\bs{p}
$.
Another integral equation for the function $f(\bs{x},\bs{p})$ reads
\begin{equation}
	\delta(\bs{p}-\bs{p}') = \!\int d\bs{x} \, \langle \bs{p}' | \bs{x} \rangle \langle \bs{x} | \bs{p} \rangle = C^2 \!\!\int d\bs{x} \, e^{if(\bs{x},\bs{p}) - if(\bs{x},\bs{p}')}.
\end{equation}
Instead of Eq.~(\ref{ExpansionF}), we expand $f(\bs{x},\bs{p})$ in the power series of $\bs{x}$ and use the previous argument to conclude that $f(\bs{x},\bs{p}) = c\bs{p}\bs{x} + f_1(\bs{x}) + f_2(\bs{p})$, where $c$ is a real constant. If we replace the left-hand side of Eq.~(\ref{IntegralEq_for_f_general}) by the Fourier representation of the delta function (\ref{DeltFunctProperty2}), we find
\begin{equation}\label{EqualityForCc}
	\int\!\frac{d\bs{p}}{2\pi\hbar} \, e^{i\bs{p}(\bs{x}-\bs{x}')/\hbar} = C^2\!\! \int\! d\bs{p}\, e^{ic\bs{p}(\bs{x}-\bs{x}')}.
\end{equation}
Therefore $C=1/\sqrt{2\pi\hbar}$ and $c=1/\hbar$, so that Eq.~(\ref{x_p_product_general}) reduces to
\begin{align}\label{bra_x_ket_p}
	\langle \bs{x} | \bs{p} \rangle = \frac{1}{\sqrt{2\pi\hbar}} e^{i\bs{x}\bs{p}/\hbar}.
\end{align}
The constant $\hbar$  was introduced in Eq.~(\ref{EqualityForCc}) for dimensional purposes. 

The canonical commutation relation for the position and momentum operators readily follows from Eqs. (\ref{bra_x_ket_p}) and (\ref{EqSummaryCommutatorSection}) from Sec.~\ref{Sec_Realizations_of_HeisenbergAlg},
\BoxedEquation{\label{XP_CommutationalRelation}
	 [ \hat{\bs{x}}, \hat{\bs{p}} ] = i\hbar.
}

%%% SECTION %%% 
\section{Derivation of the Heisenberg uncertainty principle}\label{Sec:WaveFuncHeisUncertaintyPrinc}

Let us establish the fact that an arbitrary vector $\ket{u}$ can be represented as a linear combination of two orthogonal vectors $\ket{z}$ and $\ket{v}$ (see Fig. \ref{Fig_UVZ_representation}),
\begin{align}
	& \ket{u} = \ket{v} \frac{\langle v \ket{u}}{ \langle v \ket{v}} + \ket{z} \Longrightarrow \langle v \ket{u} = \langle v \ket{u} + \langle v \ket{z} \Longrightarrow \langle v \ket{z} = 0, 
    \quad \langle u \ket{u} = \frac{| \langle u \ket{v} |^2 }{ \langle v \ket{v}} + \langle z \ket{z}, \nonumber
\end{align}
\begin{figure*}
	\begin{center}
		\includegraphics[width=0.4\hsize]{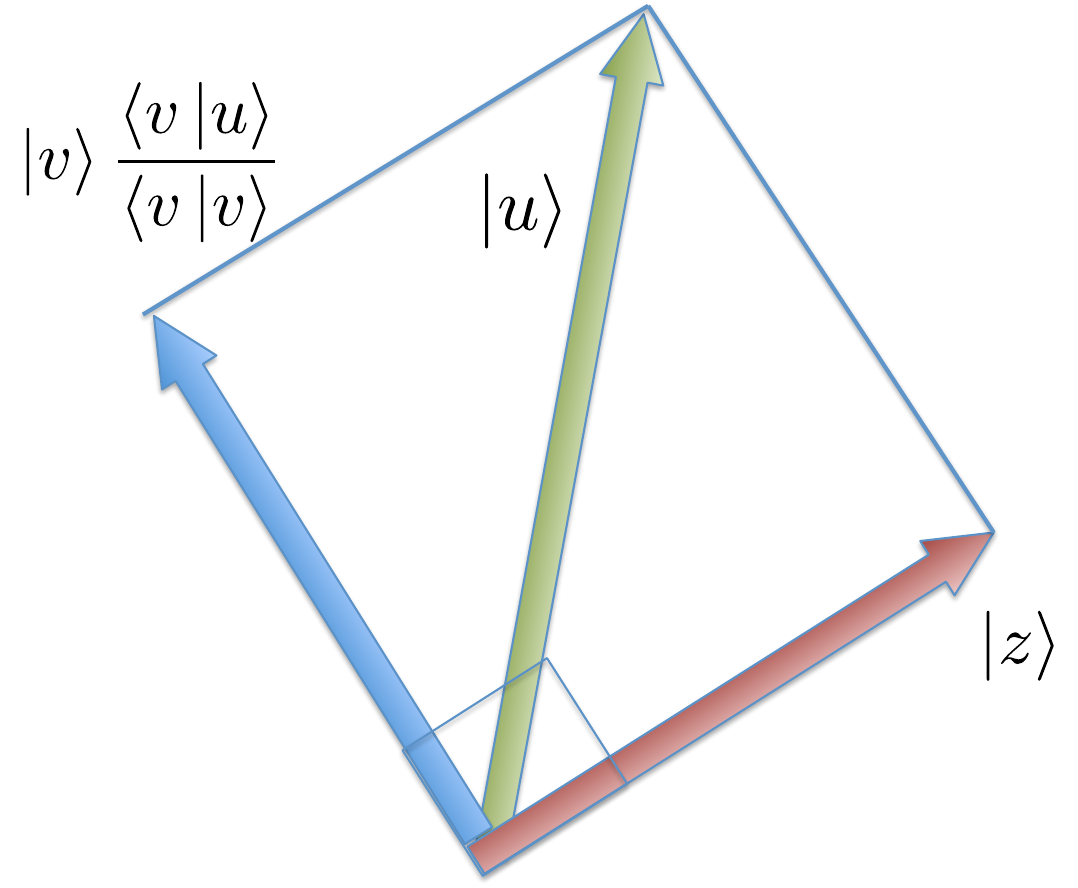}
		\caption{The representation of a vector $\ket{u}$ as a linear combination of two orthogonal vectors $\ket{z}$ and $\ket{v}$.}\label{Fig_UVZ_representation}
	\end{center}
\end{figure*}
Therefore,
\BoxedEquation{
	\langle u \ket{u} \langle v \ket{v} \geqslant | \langle u \ket{v} |^2. \label{CauchySchwarzIneq}
}
Inequality (\ref{CauchySchwarzIneq}) is known as \emph{the Cauchy--Schwarz inequality}.

Let $\ket{\psi}$ denote a normalized wave function. The standard deviation for the coordinate and momentum are defined as
\begin{align}
	& \sigma_{\bs{x}} = \sqrt{ \bra{\psi} \hat{\bs{x}}^2 \ket{\psi} - \bar{\bs{x}}^2 }, \quad 
	\sigma_{\bs{p}} = \sqrt{ \bra{\psi} \hat{\bs{p}}^2 \ket{\psi} - \bar{\bs{p}}^2 }, \nonumber
\end{align}
where the mean values are
\begin{align}
    \bar{\bs{x}} = \bra{\psi} \hat{\bs{x}} \ket{\psi}, 
    \quad \bar{\bs{p}} = \bra{\psi} \hat{\bs{p}} \ket{\psi}. 
    \nonumber
\end{align}
We can also write
\begin{align}
	& \sigma_{\bs{x}}^2 = \bra{\psi} (\hat{\bs{x}}^2 - \bar{\bs{x}}^2) \ket{\psi} = \bra{\psi} (\hat{\bs{x}} - \bar{\bs{x}})^2 \ket{\psi} 
		= \langle f \ket{f}, \quad \sigma_{\bs{p}}^2  = \langle g \ket{g}, 
    \nonumber
\end{align}
where
\begin{align}\label{eq:fg}
	\ket{f} \coloneqq \ket{  (\hat{\bs{x}} - \bar{\bs{x}}) \psi}, \qquad \ket{g} \coloneqq \ket{  (\hat{\bs{p}} - \bar{\bs{p}}) \psi}.
\end{align}
We arrive at the short notation
\BoxedEquation{
	\sigma_{\bs{x}}^2 = \langle f \ket{f}, \qquad \sigma_{\bs{p}}^2  = \langle g \ket{g}.
}

According to the Cauchy-Schwarz inequality~(\ref{CauchySchwarzIneq}), we obtain
\begin{align}
	\sigma_{\bs{x}}^2 \sigma_{\bs{p}}^2 = \langle f \ket{f} \langle g \ket{g} \geqslant | \langle f \ket{g} |^2.
\end{align}
For any complex number $z$
\begin{align}\label{AuxiliaryComplexIneq} 
	|z|^2  = \left(\Re z\right)^2 + \left(\Im z\right)^2 \geqslant \left(\Im z\right)^2 = \left( \frac{z- z^*}{2i} \right)^2,
\end{align}
meaning that
\begin{align}
	 | \langle f \ket{g} |^2 \geqslant \left( \frac{\langle f \ket{g} - \langle g \ket{f}}{2i} \right)^2.
\end{align} 
Recalling the definition~\eqref{eq:fg},
\begin{align}
	\langle f \ket{g} = \bra{\psi}(\hat{\bs{x}} - \bar{\bs{x}})(\hat{\bs{p}} - \bar{\bs{p}}) \ket{\psi} = \bra{\psi} \hat{\bs{x}}\hat{\bs{p}} \ket{\psi} - 
		\bra{\psi} \hat{\bs{x}} \ket{\psi} \bra{\psi} \hat{\bs{p}} \ket{\psi},
\end{align}
and analogously
\begin{align}
		\langle g \ket{f} =   \bra{\psi} \hat{\bs{p}}\hat{\bs{x}} \ket{\psi} - \bra{\psi} \hat{\bs{x}} \ket{\psi} \bra{\psi} \hat{\bs{p}} \ket{\psi} ,
\end{align}
we obtain
\begin{align}
 |\langle f \ket{g} |^2 \geqslant \left( \bra{\psi} \commut{  \hat{\bs{p}}, \hat{\bs{x}} } \ket{\psi}/ (2i) \right)^2.
\end{align} 

Finally, we write
 \BoxedEquation{
 	\sigma_{\bs{x}} \sigma_{\bs{p}} \geqslant |\bra{\psi} \commut{  \hat{\bs{p}}, \hat{\bs{x}} } \ket{\psi}| / 2 = \hbar / 2.
 }

\chapter{3. Formal analysis of evolution}\label{Chapter:3}
\section{Stone's theorem}\label{Sec:Stones_Th}

Let us define the unitary evolution operator $\hat{U}(t_2, t_1)$ with the following properties:
\begin{align}
	\mbox{Group product} &: \quad \hat{U}(t_3, t_2) \hat{U}(t_2, t_1) = \hat{U}(t_3, t_1), \label{EqUPCausality}  \\
	\mbox{The principle of reversibility} &: \quad \hat{U}(t_2, t_1)^{\dagger} = \hat{U}(t_2, t_1)^{-1} = \hat{U}(t_1, t_2),  \label{EqUBackwardDyn} \\
	\mbox{Group identity (``nothing happening'')} &: \quad \hat{U}(t,t) = 1. \label{EqUNoEvol} 
\end{align}
If $t_3 \geq t_2 \geq t_1$, then Eq. (\ref{EqUPCausality}) implies the principle of causality.

Assuming the differentiability with respect to the arguments, one gets from Eq. (\ref{EqUNoEvol})
\begin{align}
	\hat{U}(t + \delta t, t) = 1 + \hat{g}(t) \delta t + O\left( \delta t^2 \right),  \qquad
	\hat{g}(t) = \left. \frac{\partial}{\partial t'} \hat{U}(t', t) \right|_{t'=t}.
\end{align}
Now we verify Eq. (\ref{EqUPCausality})
\begin{align}
	&\hat{U}(t + \delta t, t') \hat{U}(t', t)
	= \left[ \hat{U}(t, t') + \frac{\partial}{\partial t} \hat{U}(t, t') \delta t + O\left( \delta t^2 \right) \right] \hat{U}(t', t) \notag\\
	&\quad
    = 1 + \left[ \frac{\partial}{\partial t} \hat{U}(t, t') \right] \hat{U}(t', t) \delta t + O\left( \delta t^2 \right)
		= \hat{U}(t + \delta t, t) =1 + \hat{g}(t) \delta t + O\left( \delta t^2 \right) \notag\\
	  &\quad\Longrightarrow
     \left[ \frac{\partial}{\partial t} \hat{U}(t, t') \right] \hat{U}(t', t) = \hat{g}(t) \Longrightarrow
	\frac{\partial}{\partial t} \hat{U}(t, t') = \hat{g}(t) \hat{U}(t', t)^{-1} = \hat{g}(t) \hat{U}(t, t').
\end{align}
Since $1 = \hat{U}(t, t')^{\dagger} \hat{U}(t, t')$, we have
\begin{align}
	0 &= \frac{\partial}{\partial t} \left[ \hat{U}(t, t')^{\dagger} \hat{U}(t, t') \right]
		= \frac{\partial}{\partial t} \left[ \hat{U}(t, t')^{\dagger} \right] \hat{U}(t, t')  
		+ \hat{U}(t, t')^{\dagger} \frac{\partial}{\partial t} \hat{U}(t, t') \notag\\
	   &= \hat{U}(t, t')^{\dagger} \hat{g}(t)^{\dagger} \hat{U}(t, t') + \hat{U}(t, t')^{\dagger} \hat{g}(t) \hat{U}(t, t')
	   = \hat{U}(t, t')^{\dagger} \left[ \hat{g}(t)^{\dagger} + \hat{g}(t) \right] \hat{U}(t, t') \notag\\
	   &\Longrightarrow \hat{g}(t)^{\dagger}  = - \hat{g}(t). 
\end{align}

Note that an arbitrary operator $\hat{A}$ can be represented as a sum of a self-adjoint and anti self-adjoint operators,
\BoxedEquation{
	\hat{A} = \underbrace{\frac{\hat{A} + \hat{A}^{\dagger}}{2}}_{\mbox{self-adjoint}} 
			+ \underbrace{\frac{\hat{A} - \hat{A}^{\dagger}}{2}}_{\mbox{anti self-adjoint}}
		= \underbrace{\frac{\hat{A} + \hat{A}^{\dagger}}{2}}_{\mbox{self-adjoint}}
		+ \, i \underbrace{\frac{\hat{A} - \hat{A}^{\dagger}}{2i}}_{\mbox{self-adjoint}}.
} 
Therefore, $\hat{g}(t) = -i\hat{G}(t)$, where $G(t)$ is self-adjoint. This finally leads us to the Stone's theorem~\cite{Reed1980, Araujo2008} (see also Ref.~\cite[Sec. X.12]{Reed1975}) 
\BoxedEquation{\label{EqStoneForPropagator}
	\mbox{Given Eqs.~(\ref{EqUPCausality})--(\ref{EqUNoEvol})} \quad \Longleftrightarrow \quad
	i \frac{\partial}{\partial t} \hat{U}(t, t') = \hat{G}(t) \hat{U}(t, t'), \quad
	\hat{G}(t)^{\dagger} = \hat{G}(t).
}
However, nothing more can be said about the form of $\hat{G}(t)$.
Assuming that $\ket{\psi(t)} = \hat{U}(t, 0) \ket{\psi(0)}$, then Eq. (\ref{EqStoneForPropagator}) yields
\BoxedEquation{\label{EqStoneTheorem}
	\ket{\psi(t)} \mbox{ evolves unitary} \quad \Longleftrightarrow \quad
	 i\frac{\partial}{\partial t} \ket{\psi(t)} = \hat{G}(t) \ket{\psi(t)}, \quad
	 \hat{G}(t)^{\dagger} = \hat{G}(t).
}

Stone's theorem guaranties not only the existence of unique solutions of Schr\"{o}dinger and Liouville equations, but also the conservation of the wave function norms.

%%% SECTION %%% 
\section{Time-ordered exponent}\label{Sec:TimeOrderedExponent}

For an arbitrary linear operator $\hat{\mathcal{G}}(t)$, consider the equation
\begin{align}\label{EqMotionForArbitraryG}
	i \frac{\partial}{\partial t} \hat{\mathcal{U}}(t, t') = \hat{\mathcal{G}}(t) \hat{\mathcal{U}}(t, t'). \qquad 
	\mbox{($\hat{\mathcal{G}}$ need not to be self-adjoint.)}
\end{align}
Let us show that the following series, known as the \emph{time-ordered exponent}, provides a formal solution of Eq. (\ref{EqMotionForArbitraryG}):
\BoxedEquation{\label{EqTExpGeneric}
	\hat{\mathcal{U}}(t, t') =& \hat{\mathcal{T}} \exp\left[ -i \int_{t'}^{t} \hat{\mathcal{G}}(\tau) d\tau \right] \notag\\
		=& 1 - i \int_{t'}^{t} d\tau_1 \hat{\mathcal{G}}(\tau_1) 
		+ (-i)^2 \int_{t'}^{t} d\tau_1 \int_{t'}^{\tau_1} d\tau_2 \hat{\mathcal{G}}(\tau_1) \hat{\mathcal{G}}(\tau_2) \notag\\
		& + (-i)^3 \int_{t'}^{t} d\tau_1 \int_{t'}^{\tau_1} d\tau_2 \int_{t'}^{\tau_2} d\tau_3 \hat{\mathcal{G}}(\tau_1) \hat{\mathcal{G}}(\tau_2) \hat{\mathcal{G}}(\tau_3)
		+ \cdots
}
Equality (\ref{EqTExpGeneric}) takes place for any generator of motion $G(t)$ that does not have to be  self-adjoint. By a direct substitution of Eq. (\ref{EqTExpGeneric}) into Eq. (\ref{EqStoneForPropagator}), we confirm that
\begin{align}
	i \frac{\partial}{\partial t} \hat{\mathcal{T}} \exp\left[ -i \int_{t'}^{t} \hat{\mathcal{G}}(\tau) d\tau \right]  =&
	 	\hat{\mathcal{G}}(t) 
		+ (-i)^1 \int_{t'}^{t} d\tau_2 \hat{\mathcal{G}}(t) \hat{\mathcal{G}}(\tau_2) 
		+ (-i)^2 \int_{t'}^{t} d\tau_2 \int_{t'}^{\tau_2} d\tau_3 \hat{\mathcal{G}}(t) \hat{\mathcal{G}}(\tau_2) \hat{\mathcal{G}}(\tau_3)
		+ \cdots \notag\\
	=& \hat{\mathcal{G}}(t) \left[ 1 -i \int_{t'}^{t} d\tau_1 \hat{\mathcal{G}}(\tau_1) 
		+ (-i)^2 \int_{t'}^{t} d\tau_1 \int_{t'}^{\tau_1} d\tau_2 \hat{\mathcal{G}}(\tau_1) \hat{\mathcal{G}}(\tau_2)
		+ \cdots \right] \notag\\
	=& \hat{\mathcal{G}}(t) \hat{\mathcal{T}} \exp\left[ -i \int_{t'}^{t} \hat{\mathcal{G}}(\tau) d\tau \right].
\end{align}
It is important to remember that $\hat{\mathcal{T}}$ is not a mere decoration. If we drop $\hat{\mathcal{T}}$ in Eq. (\ref{EqTExpGeneric}), we obtain
\begin{align}\label{EqExpWithoutT}
	\exp\left[ -i \int_{t'}^{t} \hat{\mathcal{G}}(\tau) d\tau \right] =& 1 -i \int_{t'}^{t} \hat{\mathcal{G}}(\tau) d\tau 
		+ \frac{1}{2!} \left[ -i \int_{t'}^{t} \hat{\mathcal{G}}(\tau) d\tau \right]^2 + \frac{1}{3!} \left[ -i \int_{t'}^{t} \hat{\mathcal{G}}(\tau) d\tau \right]^3 + \cdots \notag\\
		=&  1 -i \int_{t'}^{t} \hat{\mathcal{G}}(\tau) d\tau
		+ \frac{(-i)^2}{2!}  \int_{t'}^{t} d\tau_1 \int_{t'}^{t} d\tau_2 \hat{\mathcal{G}}(\tau_1) \hat{\mathcal{G}}(\tau_2) \notag\\
		& + \frac{(-i)^3}{3!} \int_{t'}^{t} d\tau_1 \int_{t'}^{t} d\tau_2 \int_{t'}^{t} d\tau_3 \hat{\mathcal{G}}(\tau_1) \hat{\mathcal{G}}(\tau_2) \hat{\mathcal{G}}(\tau_3) + \cdots.
\end{align}
We will return to the comparison between Eqs. (\ref{EqTExpGeneric}) and (\ref{EqExpWithoutT}) in Sec.~\ref{Sec:SplitOpSchrodinger} [in particular, see Eq. (\ref{EqTimeOrderingDropping})].

If the generator of motion is time independent, both the series (\ref{EqTExpGeneric}) and (\ref{EqExpWithoutT}) coincide 
\BoxedEquation{\label{EqTimeIndependentTExp}
	\hat{\mathcal{T}} \exp\left[ -i (t-t') \hat{\mathcal{G}}\right] = \exp\left[ -i (t-t') \hat{\mathcal{G}}\right]. 
}
To show this, we make use of the following integral identity\footnote{\url{https://dlmf.nist.gov/1.4.E31}}
\begin{align}\label{EqSequentialIntegrIdenity}
	\int_{t'}^{t} d\tau_1 \int_{t'}^{\tau_1} d\tau_2 \cdots \int_{t'}^{\tau_{n-1}} d\tau_n \, f(\tau_n) = \frac{1}{(n-1)!} \int_{t'}^{t} d\tau (t-\tau)^{n-1}  f(\tau) .
\end{align}
Eqs. (\ref{EqTExpGeneric}) and (\ref{EqSequentialIntegrIdenity}) yield
\begin{align}
	&\hat{\mathcal{T}} \exp\left[ -i (t-t') \hat{\mathcal{G}}\right]
     \notag\\
	&=
     1 -i\hat{\mathcal{G}}  (t-t')   + (-i \hat{\mathcal{G}})^2  \int_{t'}^{t} d\tau_1 \int_{t'}^{\tau_1} d\tau_2 \, 1 
		 + (-i \hat{\mathcal{G}})^3 \int_{t'}^{t} d\tau_1 \int_{t'}^{\tau_1} d\tau_2 \int_{t'}^{\tau_2} d\tau_3 \, 1 + \cdots \notag\\
	&= 1 -i\hat{\mathcal{G}}  (t-t')  + (-i \hat{\mathcal{G}})^2 \frac{1}{1!} \int_{t'}^{t} d\tau (t-\tau) + (-i \hat{\mathcal{G}})^3 \frac{1}{2!} \int_{t'}^{t} d\tau (t-\tau)^2 + \cdots \notag\\
	&= 1 -i\hat{\mathcal{G}}  (t-t')  + (-i \hat{\mathcal{G}})^2 \frac{(t-t')^2}{1! \cdot 2} + (-i \hat{\mathcal{G}})^3 \frac{(t-t')^3 }{2! \cdot 3} + \cdots = \exp\left[ -i (t-t') \hat{\mathcal{G}}\right].
\end{align}

%%% SECTION %%% 
\section{Time ordering operator $\hat{\mathcal{T}}$: Representing  $\hat{\mathcal{T}} \exp$ as the product of  $\hat{\mathcal{T}}$ and $\exp$}

Let us begin by establishing the differentiation rule
\begin{align}\label{EqLeibnizTimeDiffAB}
	\frac{d}{dt} \left[ \hat{A}(t)\hat{B}(t) \right] = \hat{A}(t) \frac{d \hat{B}(t)}{dt} + \frac{d \hat{A}(t)}{dt}\hat{B}(t),
\end{align}
which is derived as 
\begin{align}
	\frac{d}{dt} \left[ \hat{A}(t)\hat{B}(t) \right] &= \lim_{\delta t \to 0} \frac{1}{\delta t} \left[ \hat{A}(t + \delta t)\hat{B}(t + \delta t) -  \hat{A}(t)\hat{B}(t)  \right] \notag\\
		&=  \lim_{\delta t \to 0} \frac{1}{\delta t} \left[ \left(  \hat{A}(t) + \delta t \frac{d \hat{A}(t)}{dt} + O\left( \delta t^2 \right) \right)\left(  \hat{B}(t) + \delta t \frac{d \hat{B}(t)}{dt} + O\left( \delta t^2 \right) \right) -  \hat{A}(t)\hat{B}(t) \right] \notag\\
		&= \lim_{\delta t \to 0} \frac{1}{\delta t} \left[ \hat{A}(t)\hat{B}(t) + \delta t \hat{A}(t) \frac{d \hat{B}(t)}{dt} + O\left( \delta t^2 \right) + \delta t \frac{d \hat{A}(t)}{dt} \hat{B}(t) + O\left( \delta t^2 \right) -  \hat{A}(t)\hat{B}(t) \right]. \notag
\end{align}
Using the mathematical induction [in the similar fashion as it has been used to prove Eq.~\eqref{EqLeibnizCommutatorRule}], Eq. \eqref{EqLeibnizTimeDiffAB} can be generalized to
\BoxedEquation{\label{EqLeibnizMultipleTimeDiff}
	\frac{d}{dt} \left[ \hat{A}_1(t) \cdots \hat{A}_n(t) \right] = \sum_{k=1}^n \hat{A}_1(t) \cdots \hat{A}_{k-1}(t) \frac{d \hat{A}_k(t)}{dt}  \hat{A}_{k+1}(t) \cdots  \hat{A}_n(t).
}

We introduce \emph{the time ordering operator} $\hat{\mathcal{T}}$,
\BoxedEquation{\label{EqDefTimeOrderingOp}
	\hat{\mathcal{T}} \left[ \hat{\mathcal{G}}(t_1) \hat{\mathcal{G}}(t_2) \right]	=
	\left\{ 
		\hat{\mathcal{G}}(t_1) \hat{\mathcal{G}}(t_2) \mbox{ if } t_1 \geq t_2,
		\atop
		\hat{\mathcal{G}}(t_2) \hat{\mathcal{G}}(t_1) \mbox{ otherwise,} 
	\right.
}
assuming to obey the property of linearity 
\begin{align}\label{EqLinearityTimeOrdering}
	\hat{\mathcal{T}} \left[ \alpha \hat{\mathcal{G}}(t_1) \hat{\mathcal{G}}(t_2) + \beta \hat{\mathcal{G}}(t_3) \hat{\mathcal{G}}(t_4) \right]
	= \alpha \hat{\mathcal{T}} \left[ \hat{\mathcal{G}}(t_1) \hat{\mathcal{G}}(t_2)  \right]
 	+ \beta \hat{\mathcal{T}} \left[ \hat{\mathcal{G}}(t_3) \hat{\mathcal{G}}(t_4) \right].
\end{align}
The definition \eqref{EqDefTimeOrderingOp} is extendable to multiple product of the operators
\begin{align}\label{EqDefTimeOrderingOpMultiple}
	\hat{\mathcal{T}} \left[  \hat{\mathcal{G}}(t_1)\hat{\mathcal{G}}(t_2) \cdots \hat{\mathcal{G}}(t_n) \right]
	= \hat{\mathcal{G}}(t_{k_1}) \hat{\mathcal{G}}(t_{k_2}) \cdots \hat{\mathcal{G}}(t_{k_n}), 
	\qquad t_{k_1} \geq  t_{k_2} \geq \ldots \geq t_{k_n},
\end{align}
where the sequence $(k_1, k_2, \ldots, k_n)$ is a permutation of $(1, 2, \ldots, n)$ such that $t_{k_1} \geq  t_{k_2} \geq \ldots \geq t_{k_n}$. Note that the order of operator multiplication is irrelevant under the sign of the time ordering operator. As we will see in Eq. \eqref{EqMeaingTInTexp} below, the introduced $\hat{\mathcal{T}}$ in fact coincides with the same symbol used in Eq. \eqref{EqTExpGeneric}.
 
We are going to prove the identity
\begin{align}\label{EqTPowerIdentity}
	\hat{\mathcal{T}} \left\{ \left[ \int_{t'}^t \hat{\mathcal{G}}(\tau) d\tau \right]^N \right\}
	= N! \, \int_{t'}^{t} d\tau_1 \int_{t'}^{\tau_1} d\tau_2 \cdots \int_{t'}^{\tau_{N-1}} d\tau_N \,\hat{\mathcal{G}}(\tau_1) \hat{\mathcal{G}}(\tau_2) \cdots \hat{\mathcal{G}}(\tau_N). 
\end{align}
Note that, expect for $N!$, the r.h.s. of Eq. \eqref{EqTPowerIdentity} coincides with the $N$-th term in the definition of the time ordered exponent \eqref{EqTExpGeneric}. Consider the operator 
\begin{align}
	\hat{g}_N(t) = \hat{\mathcal{T}} \left\{ \left[ \int_{t'}^t \hat{\mathcal{G}}(\tau) d\tau \right]^N \right\}.
\end{align}
Due to the linearity \eqref{EqLinearityTimeOrdering},  $\hat{\mathcal{T}}$ commutes with the time-derivative; hence,
\begin{align}
	\frac{d \hat{g}_N(t)}{dt} &= 
	\hat{\mathcal{T}} \left\{ \frac{d}{dt} \left[ \int_{t'}^t \hat{\mathcal{G}}(\tau) d\tau \right]^N \right\} \notag\\
	&= \hat{\mathcal{T}} \left\{  
		\sum_{k=1}^N 
		\left( \int_{t'}^t \hat{\mathcal{G}}(\tau) d\tau \right)^{k-1}
		\left( \frac{d}{dt} \int_{t'}^t \hat{\mathcal{G}}(\tau) d\tau \right)
		\left( \int_{t'}^t \hat{\mathcal{G}}(\tau) d\tau \right)^{N-k}
	\right\}  
    \qquad \mbox{[using Eq. \eqref{EqLeibnizMultipleTimeDiff}]} \notag\\
	&= \hat{\mathcal{T}} \left\{  
		\sum_{k=1}^N 
		\left( \int_{t'}^t \hat{\mathcal{G}}(\tau) d\tau \right)^{k-1}
		\hat{\mathcal{G}}(t)
		\left( \int_{t'}^t \hat{\mathcal{G}}(\tau) d\tau \right)^{N-k}
	\right\}
    \notag\\
	&= \hat{\mathcal{G}}(t) \hat{\mathcal{T}} \left\{  
		\sum_{k=1}^N 
		\left( \int_{t'}^t \hat{\mathcal{G}}(\tau) d\tau \right)^{N-1}
	\right\} \mbox{ [using Eq. \eqref{EqDefTimeOrderingOpMultiple}]} \notag\\
	&= N \hat{\mathcal{G}}(t) \hat{g}_{N-1}(t) \qquad \Longrightarrow \notag\\
\hat{g}_N(t) &= \int_{t'}^t d\tau_1 \frac{d \hat{g}_N(\tau_1)}{d\tau_1}
		= N \int_{t'}^t d\tau_1 \hat{\mathcal{G}}(\tau_1) \hat{g}_{N-1}(\tau_1)
		= N \int_{t'}^t d\tau_1 \hat{\mathcal{G}}(\tau_1)  \int_{t'}^{\tau_1} d\tau_2 \frac{d \hat{g}_{N-1}(\tau_2)}{d \tau_2} \notag\\
		&= N (N-1) \int_{t'}^t d\tau_1 \hat{\mathcal{G}}(\tau_1)  \int_{t'}^{\tau_1} d\tau_2 \hat{\mathcal{G}}(\tau_2) \hat{g}_{N-2}(\tau_2) 
		= \ldots \Longrightarrow \mbox{Eq. \eqref{EqTPowerIdentity}}. \notag
\end{align}
From Eqs. \eqref{EqTExpGeneric}, \eqref{EqExpWithoutT}, and \eqref{EqTPowerIdentity}, it immediately follows that
\BoxedEquation{\label{EqMeaingTInTexp}
	\hat{\mathcal{T}} \exp\left[ -i \int_{t'}^{t} \hat{\mathcal{G}}(\tau) d\tau \right]
	= \hat{\mathcal{T}} \left\{ 
		\exp\left[ -i \int_{t'}^{t} \hat{\mathcal{G}}(\tau) d\tau \right]
	\right\}.
}
As a simple consequence, Eq. \eqref{EqTimeIndependentTExp} readily follows from the just derived representation.

Finally, we note that the time ordering operator $\hat{\mathcal{T}}$ plays an important role in \emph{Wick's theorem} heavily used in many-body and quantum field theory (see, e.g., Ref. \cite{mattuck1992guide}).

%%% SECTION %%% 
\section{The battle of exponents: $\hat{\mathcal{T}} \exp$ vs $\exp$}

Motivated by applications in numerical simulations, we want understand how well (or badly) the ordinary exponent (\ref{EqExpWithoutT}) approximates the time-ordered exponential (\ref{EqTExpGeneric}) in the case of time dependent generator of motion $\hat{\mathcal{G}}(t)$. Throughout this section $\delta t$ is assumed to be small. 

Consider the Taylor expansion of the function 
\begin{align}
	F(T) &= \int_{t}^{T} f(\tau) d\tau, \notag\\
	F(t + \delta t) &= F'(t) \delta t + \frac{1}{2} F''(t) \delta t^2 + O\left(\delta t^3\right)
		= f(t) \delta t +  \frac{1}{2} f'(t) \delta t^2 + O\left(\delta t^3\right) 
        = f(t) \delta t
        \notag\\
		& +  \frac{1}{2} \left[ \frac{f(t + \delta t) - f(t)}{\delta t} + O\left(\delta t\right) \right] \delta t^2 + O\left(\delta t^3\right)
		= \frac{\delta t}{2} \left[ f(t+\delta t) + f(t) \right] + O\left(\delta t^3\right). 
\end{align}
where we have used the forward difference approximation (\ref{EqForwardDiffApprox}). Hence, we have obtained the trapezoidal integration rule,
\begin{align}
	\int_{t}^{t+\delta t} f(\tau) d\tau &= \frac{\delta t}{2} \left[ f(t+\delta t) + f(t) \right] + O\left( \delta t^3 \right) \notag\\
		&= \frac{\delta t}{2} \left[ f\left(t + \frac{\delta t}{2}  + \frac{\delta t}{2} \right) 
			+ f\left(t + \frac{\delta t}{2}  - \frac{\delta t}{2} \right) \right] + O\left( \delta t^3 \right) \notag\\
		&= \frac{\delta t}{2} \left[ f\left(t + \frac{\delta t}{2} \right) + \frac{\delta t}{2} f'\left(t + \frac{\delta t}{2} \right)
			+ f\left(t + \frac{\delta t}{2}\right) - \frac{\delta t}{2} f'\left(t + \frac{\delta t}{2} \right) + O\left( \delta t^2 \right) \right] 
			 \notag\\
		&+ O\left( \delta t^3 \right)=  f\left(t + \frac{\delta t}{2}\right) \delta t + O\left( \delta t^3 \right). \notag
\end{align}
Summarizing this, we have
\BoxedEquation{\label{EqTrapezoidalIntegralRule}
	\int_{t}^{t+\delta t} f(\tau) d\tau &= f\left(t + \frac{\delta t}{2}\right) \delta t + O\left( \delta t^3 \right)
		= \frac{\delta t}{2} \left[ f(t+\delta t) + f(t) \right] + O\left( \delta t^3 \right) \notag\\
		&= f\left(t \right) \delta t + O\left( \delta t^2 \right)
		= f\left(t + \delta t\right) \delta t + O\left( \delta t^2 \right).
}

By employing Eq.~(\ref{EqTrapezoidalIntegralRule}), let us now evaluate the following difference:
\begin{align}
	& \exp\left[- i \int_{t}^{t + \delta t} \hat{\mathcal{G}}(\tau) d\tau \right] - \hat{\mathcal{T}} \exp\left[- i \int_{t}^{t + \delta t} \hat{\mathcal{G}}(\tau) d\tau \right]= 1 - i \int_{t}^{t + \delta t} dt' \hat{\mathcal{G}}(t')
    \notag\\
	&   -\frac{1}{2} \left[ \int_{t}^{t + \delta t} dt' \hat{\mathcal{G}}(t') \right]^2 + O\left( \delta t^3 \right)  -1 + i \int_{t}^{t + \delta t} dt' \hat{\mathcal{G}}(t')  +  \int_{t}^{t + \delta t} dt' \int_{t}^{t'} dt'' \hat{\mathcal{G}}(t')\hat{\mathcal{G}}(t'')  \notag\\
	& = -\frac{1}{2} \left[ \int_{t}^{t + \delta t} dt' \hat{\mathcal{G}}(t') \right]^2 +  \int_{t}^{t + \delta t} dt' \int_{t}^{t'} dt'' \hat{\mathcal{G}}(t')\hat{\mathcal{G}}(t'')  + O\left( \delta t^3 \right) \notag\\
	&= -\frac{1}{2} \left[ \hat{\mathcal{G}}\left(t + \frac{\delta t}{2}\right) \delta t + O\left( \delta t^3 \right) \right]^2 + \delta t  \hat{\mathcal{G}}\left(t + \frac{\delta t}{2}\right) \int_{t}^{t + \delta t / 2} dt'' \hat{\mathcal{G}}(t'')  + O\left( \delta t^3 \right) \notag\\
	&=-\frac{1}{2} \hat{\mathcal{G}}\left(t + \frac{\delta t}{2}\right)^2 \delta t^2 + \frac{\delta t^2}{2}  \left[ \hat{\mathcal{G}}\left(t + \frac{\delta t}{2}\right)^2 + O\left( \delta t^2 \right) \right]  + O\left( \delta t^3 \right). \notag
\end{align}
This gives us
\BoxedEquation{\label{EqTimeOrderingDropping}
	\hat{\mathcal{T}} \exp\left[- i \int_{t}^{t + \delta t} \hat{\mathcal{G}}(\tau) d\tau \right]
	&=  \exp\left[-i \int_{t}^{t + \delta t} \hat{\mathcal{G}}(\tau) d\tau \right] + O\left( \delta t^3 \right) \notag\\
	&=  \exp\left[-i \delta t \hat{\mathcal{G}}\left(t + \frac{\delta t}{2}\right) \right] + O\left( \delta t^3 \right).
}
The resulting accuracy can be further improved [see Eq. \eqref{EqBandraukThirdOrderTimeOrderedExp}]. To accomplish this task, we need to investigate other representations of the time-ordered exponent.

%%% SECTION %%% 
\section{The time-ordered exponent as an exponential product and the Trotter formula}

Consider the partitioning of the time interval $(t, t')$: $t_k = t' + k \delta t$, $k=0,\ldots,N$, $\delta t = (t - t')/N$. Thus, from the semi group (i.e., not allowing to go back in time) property \eqref{EqUPCausality}, we obtain
\begin{align}
	\hat{\mathcal{U}}(t, t') &= \hat{\mathcal{U}}(t_N, t_{N-1})\hat{\mathcal{U}}(t_{N-1}, t_{N-2}) \cdots \hat{\mathcal{U}}(t_1, t_0)\notag
        \\
			& =  \hat{\mathcal{T}} \exp\left[- i \int_{t_{N-1}}^{t_{N-1} + \delta t} \hat{\mathcal{G}}(\tau) d\tau \right] \cdots
			\hat{\mathcal{T}} \exp\left[- i \int_{t_{0}}^{t_{0} + \delta t} \hat{\mathcal{G}}(\tau) d\tau \right] \notag\\
			&= \left[ e^{-i \delta t \hat{\mathcal{G}}\left(t_{N-1} + \frac{\delta t}{2} \right)} + O\left( \delta t^3 \right) \right] \cdots \left[ e^{-i \delta t \hat{\mathcal{G}}\left(t_{0} + \frac{\delta t}{2} \right)} + O\left( \delta t^3 \right) \right]
			 \qquad \mbox{[using Eq. \eqref{EqTimeOrderingDropping}]}  \notag\\
			& = \left[ e^{-i \delta t \hat{\mathcal{G}}\left(t_{N} \right)} + O\left( \delta t^2 \right) \right] 
			\cdots \left[ e^{-i \delta t \hat{\mathcal{G}}\left(t_{1} \right)} + O\left( \delta t^2 \right) \right] \qquad \mbox{[using Eq. \eqref{EqTrapezoidalIntegralRule}]}  \notag\\
			&= \left[ e^{-i \delta t \hat{\mathcal{G}}\left(t_{N-1} \right)} + O\left( \delta t^2 \right) \right] 
			\cdots \left[ e^{-i \delta t \hat{\mathcal{G}}\left(t_{0} \right)} + O\left( \delta t^2 \right) \right]. \qquad \mbox{[using Eq. \eqref{EqTrapezoidalIntegralRule}]} 
\end{align}
Since $\delta t \to 0$ as $N \to \infty$, we finally get
\BoxedEquation{\label{EqTimeExpAsProdExp}
	\hat{\mathcal{U}}(t, t') &= \lim_{N \to \infty} \left( e^{-i \delta t \hat{\mathcal{G}}\left(t_{N-1} + \frac{\delta t}{2} \right)} \cdots e^{-i \delta t \hat{\mathcal{G}}\left(t_{0} + \frac{\delta t}{2} \right)}
 \right) \notag\\
 	&= \lim_{N \to \infty} \left( e^{-i \delta t \hat{\mathcal{G}}\left(t_{N-1} \right)} \cdots e^{-i \delta t \hat{\mathcal{G}}\left(t_{0} \right)} \right) 
	= \lim_{N \to \infty}  \left( e^{-i \delta t \hat{\mathcal{G}}\left(t_{N} \right)} \cdots e^{-i \delta t \hat{\mathcal{G}}\left(t_{1} \right)} \right).
}
This identity implies that not only the time-ordered exponent can be sliced into infinitely many ordinary operatorial exponents, but also that these slices must be \emph{ordered in time} from the earlier time on the right to the latest on the left. Thus, the name \emph{time-ordered} exponent. 

An analogue of Eq.~\eqref{EqTimeExpAsProdExp} exists for arbitrary exponents, and it is known as the Trotter product formula, playing an important role in many applications. It reads as
\BoxedEquation{\label{TrotterProductEq}
	e^{\hat{A} + \hat{B}} = \lim_{N\to\infty}\left( e^{\hat{A}/N} e^{\hat{B}/N} \right)^N
		= \lim_{N\to\infty}\left( e^{\hat{B}/N} e^{\hat{A}/N} \right)^N.
}
Equation~(\ref{TrotterProductEq}) follows from simple relations by taking the limit $\delta \to 0$,
\begin{align}\label{SecondOrderTrotterApproxEq}
	 & 1 + (\hat{A} + \hat{B}) \delta + O\left( \delta^2 \right) = \left[ 1 + \hat{A}\delta + O\left( \delta^2 \right) \right]\left[ 1 + \hat{B}\delta + O\left( \delta^2 \right) \right] 
	 \Longrightarrow e^{\hat{A}\delta + \hat{B}\delta} = e^{\hat{A}\delta}e^{\hat{B}\delta} + O\left(\delta^2\right), \notag\\
	 &1 + (\hat{A} + \hat{B}) \delta + O\left( \delta^2 \right) = \left[ 1 + \hat{B}\delta + O\left( \delta^2 \right) \right]\left[ 1 + \hat{A}\delta + O\left( \delta^2 \right) \right] 
	 \Longrightarrow e^{\hat{A}\delta + \hat{B}\delta} = e^{\hat{B}\delta}e^{\hat{A}\delta} + O\left(\delta^2\right).
\end{align}

%%% SECTION %%% 
\section{The Suzuki representation}

We begin by rewriting the Taylor expansion as the exponentiation of the derivative
\begin{align}
	e^{a\frac{\partial}{\partial t}} f(t) 
	= \sum_{n=0}^{\infty} \frac{a^n}{n!} \frac{\partial^n}{\partial t^n} f(t) 
	= \underbrace{\sum_{n=0}^{\infty} \frac{([t+a] - t)^n}{n!} \frac{\partial^n}{\partial t^n} f(t)}_{\mbox{Taylor series for $f(t + a)$}\atop \mbox{about a point $t$}},
    \notag
\end{align}
or simply
\BoxedEquation{\label{DisplacementOp}
	e^{a\frac{\partial}{\partial t}} f(t) =  f(t + a).
}
We see that $e^{a\frac{\partial}{\partial t}}$ can be called the displacement operator. 

Let $\overleftarrow{\frac{\partial}{\partial t}}$ denote the forward time derivative. The arrow indicates the direction of action; in this case, the derivative acts only to the left. Note that directional derivatives play a prominent role in Sec. \ref{Sec:PhaseSpaceQM}. Therefore,
\begin{align}\label{EqDefTimeForwardDerivative}
	\hat{F}(t) e^{\Delta t \overleftarrow{\frac{\partial}{\partial t}}} \hat{G}(t) 
	=  \left[ e^{\Delta t {\frac{\partial}{\partial t}}} \hat{F}(t) \right] \hat{G}(t)
	= \hat{F}(t + \Delta t) \hat{G}(t),
\end{align}
for any $\hat{G}$ and $\hat{F}$. Now we evaluate using Eq. \eqref{TrotterProductEq}
\begin{align}
	e^{\Delta t \left(-i\hat{\mathcal{G}}(t) + \overleftarrow{\frac{\partial}{\partial t}} \right)}
	&= \lim_{N \to \infty} \left( 
		e^{-i \delta t \hat{\mathcal{G}}(t)} 
		e^{\delta t \overleftarrow{\frac{\partial}{\partial t}}} 
	\right)^N \qquad
	[\delta t= \Delta t / N] \notag\\
	&= \lim_{N \to \infty}  \left(\prod_{k=1}^{N-1} e^{-i \delta t \hat{\mathcal{G}}(t)} 
		e^{\delta t \overleftarrow{\frac{\partial}{\partial t}}} 
	\right) e^{-i \delta t \hat{\mathcal{G}}(t)} 
		e^{\delta t \overleftarrow{\frac{\partial}{\partial t}}} 
	= \lim_{N \to \infty}  \left(\prod_{k=1}^{N-1} e^{-i \delta t \hat{\mathcal{G}}(t + \delta t)} 
		e^{\delta t \overleftarrow{\frac{\partial}{\partial t}}} 
	\right) e^{-i \delta t \hat{\mathcal{G}}(t + \delta t)} \notag\\
	&= \lim_{N \to \infty}  \left(\prod_{k=1}^{N-2} e^{-i \delta t \hat{\mathcal{G}}(t + \delta t)} 
		e^{\delta t \overleftarrow{\frac{\partial}{\partial t}}} 
	\right) e^{-i \delta t \hat{\mathcal{G}}(t + \delta t)} 
		e^{\delta t \overleftarrow{\frac{\partial}{\partial t}}}
		e^{-i \delta t \hat{\mathcal{G}}(t + \delta t)} \notag\\
	&= \lim_{N \to \infty}  \left(\prod_{k=1}^{N-2} e^{-i \delta t \hat{\mathcal{G}}(t + 2\delta t)} 
		e^{\delta t \overleftarrow{\frac{\partial}{\partial t}}} 
	\right) e^{-i \delta t \hat{\mathcal{G}}(t + 2\delta t)} 
		e^{-i \delta t \hat{\mathcal{G}}(t + \delta t)} \notag\\
	& = \cdots = \lim_{N \to \infty} \left( 
		e^{-i \delta t \hat{\mathcal{G}}(t + N\delta t)} 
		\cdots
		e^{-i \delta t \hat{\mathcal{G}}(t + 2\delta t)} 
		e^{-i \delta t \hat{\mathcal{G}}(t + 1\delta t)} 
	\right) \notag\\
	& = \hat{\mathcal{U}}(t + \Delta t, t).
	\quad \mbox{[using Eq. \eqref{EqTimeExpAsProdExp}]}
\end{align}
In summary, we obtained the Suzuki representation of the time-ordered exponent \cite{suzuki1993general},
\BoxedEquation{\label{EqSuzukiTimeOrderedExp}
	\hat{\mathcal{T}} \exp\left[- i \int_{t}^{t + \Delta t} \hat{\mathcal{G}}(\tau) d\tau \right] = \exp\left[\Delta t \left(-i\hat{\mathcal{G}}(t) + \overleftarrow{\frac{\partial}{\partial t}} \right)\right],
}
playing an important role to get high accuracy numerical methods (see, e.g., Ref. \cite{chin2002gradient}). Equation \eqref{EqSuzukiTimeOrderedExp} can be interpreted as a reduction of the time dependent  (r.h.s.) to time-independent case (l.h.s.) with the modified generator of motion. Alternatively, we can say that due to the non-commutativity of $\overleftarrow{\frac{\partial}{\partial t}}$  and $\hat{\mathcal{G}}$, insertion of $\overleftarrow{\frac{\partial}{\partial t}}$ scrambles an ordinary exponent into a weird time-ordered exponential. A conceptually similar transition is done when the numerical methods for classical dynamics are developed in Sec. \ref{Sec:SymplecticPropagator}.

%%% SECTION %%% moved: The Baker-Campbell-Hausdorff formula

%%% SECTION %%% 
\section{The battle of exponents: $\hat{\mathcal{T}} \exp$ vs $\exp$ II}

Here, let us provide an improved version of Eq. \eqref{EqTimeOrderingDropping} with a higher-order  error estimate. We begin with the approximation \cite{bandrauk1992higher, bandrauk1993exponential}
\begin{align}\label{EqBandraukForthOrder}
	e^{\delta t (\hat{A} + \hat{B})} = e^{\delta t \frac{s}2 \hat{A}} e^{\delta t s \hat{B}} 
	e^{\frac{1-s}2 \delta t \hat{A}} e^{\delta t (1-2s) \hat{B}} e^{\frac{1-s}2 \delta t \hat{A}} 
	e^{\delta t s \hat{B}} e^{\delta t \frac{s}2 \hat{A}} 
	+ O\left( \delta t^4 \right),
\end{align}
where 
\begin{align}\label{EqValSConstBandrauk}
	s = 2^{1/3} / 3 + 2^{2/3} / 6 + 2/ 3 \approx 1.35
\end{align}
is the real solution of the algebraic equation $2s^3 + (1-2s)^3 = 0$.  Relation \eqref{EqBandraukForthOrder} is established by recursively merging the exponents on the l.h.s. of Eq. \eqref{EqBandraukForthOrder} into one by recursively applying Eq. \eqref{EqBCH}. This derivation is cumbersome and is best performed using symbolic computation packages capable of handling non-commuting quantities. Using Eqs.~\eqref{EqDefTimeForwardDerivative}, \eqref{EqSuzukiTimeOrderedExp}, and \eqref{EqBandraukForthOrder}, we obtain
\begin{align}
	\hat{\mathcal{U}}(t + \delta t, t) =& 
		e^{-i\delta t \frac{s}2 \hat{\mathcal{G}}(t)} 
		e^{\delta t s \overleftarrow{\frac{\partial}{\partial t}}} 
		e^{-i\frac{1-s}2 \delta t \hat{\mathcal{G}}(t)} 
		e^{\delta t (1-2s) \overleftarrow{\frac{\partial}{\partial t}}} 
		e^{-i\frac{1-s}2 \delta t \hat{\mathcal{G}}(t)} 
		e^{\delta t s \overleftarrow{\frac{\partial}{\partial t}}} 
		e^{-i\delta t \frac{s}2 \hat{\mathcal{G}}(t)} 
	+ O\left( \delta t^4 \right) \notag\\
		=& e^{-i\delta t \frac{s}2 \hat{\mathcal{G}}(t + s\delta t)} 
		e^{\delta t s \overleftarrow{\frac{\partial}{\partial t}}} 
		e^{-i\frac{1-s}2 \delta t \hat{\mathcal{G}}(t + s\delta t)} 
		e^{\delta t (1-2s) \overleftarrow{\frac{\partial}{\partial t}}} 
		e^{-i\frac{1-s}2 \delta t \hat{\mathcal{G}}(t + s\delta t)} 
		e^{-i\delta t \frac{s}2 \hat{\mathcal{G}}(t)} 
	+ O\left( \delta t^4 \right) \notag\\
		=& e^{-i\delta t \frac{s}2 \hat{\mathcal{G}}(t + (1-s)\delta t)} 
		e^{\delta t s \overleftarrow{\frac{\partial}{\partial t}}} 
		e^{-i\frac{1-s}2 \delta t \hat{\mathcal{G}}(t + (1-s)\delta t)} 
		e^{-i\frac{1-s}2 \delta t \hat{\mathcal{G}}(t + s\delta t)} 
		e^{-i\delta t \frac{s}2 \hat{\mathcal{G}}(t)} 
	+ O\left( \delta t^4 \right). \notag
\end{align}
\BoxedEquation{\label{EqBandraukThirdOrderTimeOrderedExp}
	\hat{\mathcal{T}} \exp\left[- i \int_{t}^{t + \delta t} \hat{\mathcal{G}}(\tau) d\tau \right] 
	=& e^{-i\delta t \frac{s}2 \hat{\mathcal{G}}(t + \delta t)} 
		e^{-i\frac{1-s}2 \delta t \hat{\mathcal{G}}(t + (1-s)\delta t)} 
		e^{-i\frac{1-s}2 \delta t \hat{\mathcal{G}}(t + s\delta t)} 
		e^{-i\delta t \frac{s}2 \hat{\mathcal{G}}(t)} 
	+ O\left( \delta t^4 \right). 
}
This equation is a better approximation for the time-ordered exponent than Eq. \eqref{EqTimeOrderingDropping}. If you account for the value of $s$ [Eq. \eqref{EqValSConstBandrauk}], two central exponents in Eq. \eqref{EqBandraukThirdOrderTimeOrderedExp} have positive coefficients. This may not be acceptable for some applications (e.g., the Monte-Carlo simulations in classical mechanics). In this case, other available approximations should be employed. For further discussions see, e.g., Refs. \cite{chin2002gradient, blanes2016concise}.

%%% SECTION %%% 
\section{The integral equations for the time-ordered exponent}

We are going to prove the identities
\BoxedEquation{
	\hat{\mathcal{U}}(t, t') &= \hat{\mathcal{T}} \exp\left[ -i \int_{t'}^{t} \left[ \hat{\mathcal{G}}_0(\tau) + \hat{\mathcal{V}}(\tau) \right]  d\tau \right], 
	\quad 
	\hat{\mathcal{U}}_0(t, t') = \hat{\mathcal{T}} \exp\left[ -i \int_{t'}^{t} \hat{\mathcal{G}}_0(\tau) d\tau \right], \label{EqUU0Deffs} \\
	\hat{\mathcal{U}}(t, t') &= \hat{\mathcal{U}}_0(t, t') 
		- i \int_{t'}^{t} d\tau_1 \, \hat{\mathcal{U}}(t, \tau_1)  \hat{\mathcal{V}}(\tau_1) \hat{\mathcal{U}}_0(\tau_1, t') \label{EqDysonEqPostForm} \\
	&= \hat{\mathcal{U}}_0(t, t') 
		- i \int_{t'}^{t} d\tau_1 \, \hat{\mathcal{U}}_0(t, \tau_1)  \hat{\mathcal{V}}(\tau_1) \hat{\mathcal{U}}(\tau_1, t'),  \label{EqDysonEqPriorForm} \\
	& \mbox{for arbitrary $\hat{\mathcal{G}}_0(\tau)$ and  $\hat{\mathcal{V}}(\tau)$.} \notag
}
According to Eqs. (\ref{EqMotionForArbitraryG}), (\ref{EqTExpGeneric}), and (\ref{EqUU0Deffs}), we have
\begin{align}\label{EqUU0EvolutionEq}
	i \frac{\partial}{\partial t} \hat{\mathcal{U}}(t, t') = \left[ \hat{\mathcal{G}}_0(t) + \hat{\mathcal{V}}(t) \right] \hat{\mathcal{U}}(t, t'),
	\qquad
	i \frac{\partial}{\partial t} \hat{\mathcal{U}}_0(t, t') = \hat{\mathcal{G}}_0(t) \hat{\mathcal{U}}_0(t, t').
\end{align}
Recall Leibniz's theorem for differentiation of integrals\footnote{\url{https://dlmf.nist.gov/1.5.E22}}
\begin{align}\label{EqLeibnizDiffInt}
	\frac{d}{dt} \int_{\alpha(t)}^{\beta(t)} f(t, \tau_1) d\tau_1 = f(t, \beta(t)) \beta'(t) - f(t, \alpha(t)) \alpha'(t) +  \int_{\alpha(t)}^{\beta(t)} \frac{\partial}{\partial t} f(t, \tau_1) d\tau_1. 
\end{align}
Differentiating both sides of Eq. (\ref{EqDysonEqPostForm}) with respect to $t$ and using Eq. (\ref{EqUU0EvolutionEq}), we obtain
\begin{align}
	&\left[ \hat{\mathcal{G}}_0(t) + \hat{\mathcal{V}}(t) \right] \hat{\mathcal{U}}(t, t') 
	= \hat{\mathcal{G}}_0(t) \hat{\mathcal{U}}_0(t, t') 
		- i^2 \frac{\partial}{\partial t} \int_{t'}^{t} d\tau_1 \, \hat{\mathcal{U}}(t, \tau_1)  \hat{\mathcal{V}}(\tau_1) \hat{\mathcal{U}}_0(\tau_1, t') \notag\\
	&~~= \hat{\mathcal{G}}_0(t) \hat{\mathcal{U}}_0(t, t') + \hat{\mathcal{U}}(t, t)  \hat{\mathcal{V}}(t) \hat{\mathcal{U}}_0(t, t')
		+ \int_{t'}^{t} d\tau_1 \, \frac{\partial \hat{\mathcal{U}}(t, \tau_1)}{\partial t}  \hat{\mathcal{V}}(\tau_1) \hat{\mathcal{U}}_0(\tau_1, t')
		\quad \mbox{[using Eq. (\ref{EqLeibnizDiffInt})]} \notag\\
	&~~= \hat{\mathcal{G}}_0(t) \hat{\mathcal{U}}_0(t, t') + \hat{\mathcal{V}}(t) \hat{\mathcal{U}}_0(t, t')
		-i\left[ \hat{\mathcal{G}}_0(t) + \hat{\mathcal{V}}(t) \right] 
			\int_{t'}^{t} d\tau_1 \, \hat{\mathcal{U}}(t, \tau_1) \hat{\mathcal{V}}(\tau_1) \hat{\mathcal{U}}_0(\tau_1, t') \notag\\
	&~~= \left[ \hat{\mathcal{G}}_0(t) + \hat{\mathcal{V}}(t) \right] \hat{\mathcal{U}}_0(t, t')
		+ \left[ \hat{\mathcal{G}}_0(t) + \hat{\mathcal{V}}(t) \right] \left[  \hat{\mathcal{U}}(t, t') -  \hat{\mathcal{U}}_0(t, t') \right]
		\qquad \mbox{[using Eq. (\ref{EqDysonEqPostForm})]} \notag\\
	&~~= \left[ \hat{\mathcal{G}}_0(t) + \hat{\mathcal{V}}(t) \right] \hat{\mathcal{U}}(t, t'). \notag
\end{align}
The relation (\ref{EqDysonEqPriorForm}) can be verified in the same way.

Note that the integral relations (\ref{EqDysonEqPostForm}) and  (\ref{EqDysonEqPriorForm}) are known as the \emph{post} and \emph{prior} forms because the `full' propagator $\hat{\mathcal{U}}$ follows the `perturbation' $\hat{\mathcal{V}}$ in the former case, whereas $\hat{\mathcal{V}}$ precedes $\hat{\mathcal{U}}$ in the latter.

%%% SECTION %%% 
\section{Time-dependent perturbation theory}\label{SecGeneralTimeDependentPertTheory}

Let us recursively substitutive Eq. (\ref{EqDysonEqPostForm}) into itself 
\begin{align}
	{\color{magenta} \hat{\mathcal{U}}(t, t')} &= \hat{\mathcal{U}}_0(t, t') 
		- i \int_{t'}^{t} d\tau_1 \, {\color{magenta} \hat{\mathcal{U}}(t, \tau_1)} \hat{\mathcal{V}}(\tau_1) \hat{\mathcal{U}}_0(\tau_1, t') \notag\\
		&= \hat{\mathcal{U}}_0(t, t') 
		- i \int_{t'}^{t} d\tau_1 \left[ \hat{\mathcal{U}}_0(t, \tau_1) 
		- i \int_{\tau_1}^{t} d\tau_2 \, {\color{magenta} \hat{\mathcal{U}}(t, \tau_2)}  \hat{\mathcal{V}}(\tau_2) \hat{\mathcal{U}}_0(\tau_2, \tau_1) \right] 
		\hat{\mathcal{V}}(\tau_1) \hat{\mathcal{U}}_0(\tau_1, t') \notag\\
		&= \hat{\mathcal{U}}_0(t, t') 
		- i \int_{t'}^{t} d\tau_1 \, \hat{\mathcal{U}}_0(t, \tau_1)  \hat{\mathcal{V}}(\tau_1) \hat{\mathcal{U}}_0(\tau_1, t') \notag\\
	    &~~\qquad\qquad + (-i)^2 \int_{t'}^{t} d\tau_1 \int_{\tau_1}^{t} d\tau_2 \, {\color{magenta} \hat{\mathcal{U}}(t, \tau_2)}  \hat{\mathcal{V}}(\tau_2) \hat{\mathcal{U}}_0(\tau_2, \tau_1) \hat{\mathcal{V}}(\tau_1) \hat{\mathcal{U}}_0(\tau_1, t'). \notag
\end{align}
Continuing this recursive substitution \emph{ad infinitum}, we arrive at the perturbation theory series in the post form,
\BoxedEquation{\label{EqTimeDepPertTheoryPost}
	\hat{\mathcal{U}}(t, t') &= \hat{\mathcal{U}}_0(t, t')  
		-i \int_{t'}^{t} d\tau_1 \, \hat{\mathcal{U}}_0(t, \tau_1)  \hat{\mathcal{V}}(\tau_1) \hat{\mathcal{U}}_0(\tau_1, t') \notag\\
	& + (-i)^2 \int_{t'}^{t} d\tau_1 \int_{\tau_1}^{t} d\tau_2 \, \hat{\mathcal{U}}_0 (t, \tau_2)  \hat{\mathcal{V}}(\tau_2) \hat{\mathcal{U}}_0(\tau_2, \tau_1) \hat{\mathcal{V}}(\tau_1) \hat{\mathcal{U}}_0(\tau_1, t') \notag\\
	& + (-i)^3 \int_{t'}^{t} d\tau_1 \int_{\tau_1}^{t} d\tau_2 \int_{\tau_2}^{t} d\tau_3  \, 
	\hat{\mathcal{U}}_0 (t, \tau_3)  \hat{\mathcal{V}}(\tau_3) \hat{\mathcal{U}}_0(\tau_3, \tau_2) \hat{\mathcal{V}}(\tau_2) \hat{\mathcal{U}}_0(\tau_2, \tau_1) \hat{\mathcal{V}}(\tau_1) \hat{\mathcal{U}}_0(\tau_1, t') + \cdots
}
Applying the same recursive procedure onto Eq. (\ref{EqDysonEqPriorForm}), we get the perturbation theory expansion in the prior form,
 \BoxedEquation{\label{EqTimeDepPertTheoryPrior}
	\hat{\mathcal{U}}(t, t') &= \hat{\mathcal{U}}_0(t, t')  
		-i \int_{t'}^{t} d\tau_1 \, \hat{\mathcal{U}}_0(t, \tau_1)  \hat{\mathcal{V}}(\tau_1) \hat{\mathcal{U}}_0(\tau_1, t') \notag\\
	& + (-i)^2 \int_{t'}^{t} d\tau_1 \int_{t'}^{\tau_1} d\tau_2 \, \hat{\mathcal{U}}_0 (t, \tau_1)  \hat{\mathcal{V}}(\tau_1) \hat{\mathcal{U}}_0(\tau_1, \tau_2) \hat{\mathcal{V}}(\tau_2) \hat{\mathcal{U}}_0(\tau_2, t') \notag\\
	& + (-i)^3 \int_{t'}^{t} d\tau_1 \int_{t'}^{\tau_1} d\tau_2 \int_{t'}^{\tau_2} d\tau_3  \, 
	\hat{\mathcal{U}}_0 (t, \tau_1)  \hat{\mathcal{V}}(\tau_1) \hat{\mathcal{U}}_0(\tau_1, \tau_2) \hat{\mathcal{V}}(\tau_2) \hat{\mathcal{U}}_0(\tau_2, \tau_3) \hat{\mathcal{V}}(\tau_3) \hat{\mathcal{U}}_0(\tau_3, t') + \cdots
}
Note that the only difference between the post (\ref{EqTimeDepPertTheoryPost}) and prior (\ref{EqTimeDepPertTheoryPrior}) forms is 
in the limits of integration. Equation (\ref{EqTExpGeneric}) is recovered by substituting $\hat{\mathcal{G}}_0 = 0$ and $\hat{\mathcal{V}}(\tau) = \hat{\mathcal{G}}(\tau)$ into Eq.~(\ref{EqTimeDepPertTheoryPrior}).

While looking cumbersome, Eqs.~(\ref{EqTimeDepPertTheoryPost}) and (\ref{EqTimeDepPertTheoryPrior})  have a clear structure, which is best visible if a simple diagrammatic language is used, e.g., in the case of the prior expansion~(\ref{EqTimeDepPertTheoryPrior})
\BoxedEquation{\label{EqFeynmanTimeDepPertTheoryPrior}
		\hat{\mathcal{U}}(t, t') = 
		% Zero order Feynman diagram
		\begin{tikzpicture}
			\draw [ultra thick, <-] (0,0) node[below]{$t$} -- (1,0) node[below]{$t'$};
		\end{tikzpicture}
		+
		% First order Feynman diagram
		\begin{tikzpicture}
			\draw [ultra thick, <-] (0,0) node[below]{$t$} -- (1,0);
			\draw [ultra thick] (1,-0.2) node[below]{$\tau_1$} -- (1,0.3);
			\draw [ultra thick, <-] (1,0) -- (2,0) node[below]{$t'$};
		\end{tikzpicture}
		+
		% Second order Feynman diagram
		\begin{tikzpicture}
			\draw [ultra thick, <-] (0,0) node[below]{$t$} -- (1,0);
			\draw [ultra thick] (1,-0.2) node[below]{$\tau_1$} -- (1,0.3);
			\draw [ultra thick, <-] (1,0) -- (2,0);
			\draw [ultra thick] (2,-0.2) node[below]{$\tau_2$} -- (2,0.3);
			\draw [ultra thick, <-] (2,0) -- (3,0) node[below]{$t'$};
		\end{tikzpicture}
		+
		% Third order Feynman diagram
		\begin{tikzpicture}
			\draw [ultra thick, <-] (0,0) node[below]{$t$} -- (1,0);
			\draw [ultra thick] (1,-0.2) node[below]{$\tau_1$} -- (1,0.3);
			\draw [ultra thick, <-] (1,0) -- (2,0);
			\draw [ultra thick] (2,-0.2) node[below]{$\tau_2$} -- (2,0.3);
			\draw [ultra thick, <-] (2,0) -- (3,0);
			\draw [ultra thick] (3,-0.2) node[below]{$\tau_3$} -- (3,0.3);
			\draw [ultra thick, <-] (3,0) -- (4,0) node[below]{$t'$};
		\end{tikzpicture}
		+ \cdots, \\
		%%%%%%%%%%%%%%%%%%%%%%%%%%
		\mbox{Rules :} 
		\begin{tikzpicture}
			\draw [ultra thick, <-] (0,0) node[below]{$t_2$} -- (1,0) node[below]{$t_1$};
		\end{tikzpicture}
		\mbox{ equals to }
		\hat{\mathcal{U}}_0(t_2, t_1),
		\quad
		\begin{tikzpicture}
			\draw [ultra thick] (0,-0.2) node[below]{$\tau$} -- (0,0.3);
		\end{tikzpicture}
		\mbox{ equals to }
		(-i) \hat{\mathcal{V}}(\tau), \mbox{ and prior ordering}\notag\\
		\mbox{integration is assumed over all the introduced time variables [as in Eq. (\ref{EqTimeDepPertTheoryPrior})].} \notag
}
Such representation is the simplest form of the \emph{Feynman diagrams} that are ubiquitously used in physics and chemistry.
Feynman diagrams play a similar role in science as emojis do in communication: these are powerful graphical tools for delivering complex concepts.

%%% SECTION %%% 
\section{Supplementary: Representing the truncated perturbation series expansion in the form of a non-homogeneous differential equation}

Consider a partial sum up to the $n$-th term of the series (\ref{EqFeynmanTimeDepPertTheoryPrior})
\begin{align}
	\hat{\mathcal{U}}_n(t, t') = 
		% Zero order Feynman diagram
		\begin{tikzpicture}
			\draw [ultra thick, <-] (0,0) node[below]{$t$} -- (1,0) node[below]{$t'$};
		\end{tikzpicture}
		+
		% First order Feynman diagram
		\begin{tikzpicture}
			\draw [ultra thick, <-] (0,0) node[below]{$t$} -- (1,0);
			\draw [ultra thick] (1,-0.2) node[below]{$\tau_1$} -- (1,0.3);
			\draw [ultra thick, <-] (1,0) -- (2,0) node[below]{$t'$};
		\end{tikzpicture}
		+ \cdots +
		% n-th order Feynman diagram
		\begin{tikzpicture}
			\draw [ultra thick, <-] (0,0) node[below]{$t$} -- (1,0);
			\draw [ultra thick] (1,-0.2) node[below]{$\tau_1$} -- (1,0.3);
			\draw [ultra thick, dotted,  <-] (1,0) -- (2,0);
			\draw [ultra thick] (2,-0.2) node[below]{$\tau_{n}$} -- (2,0.3);
			\draw [ultra thick, <-] (2,0) -- (3,0) node[below]{$t'$};
		\end{tikzpicture}.
\end{align}
Using Eqs. (\ref{EqLeibnizDiffInt}) and (\ref{EqUU0EvolutionEq}), we obtain
\begin{align}
	i\frac{\partial}{\partial t} \left( 
		\begin{tikzpicture}
			\draw [ultra thick, <-] (0,0) node[below]{$t$} -- (1,0);
			\draw [ultra thick] (1,-0.2) node[below]{$\tau_1$} -- (1,0.3);
			\draw [ultra thick, <-] (1,0) -- (2,0);
			\draw [ultra thick] (2,-0.2) node[below]{$\tau_2$} -- (2,0.3);
			\draw [ultra thick, dotted, <-] (2,0) -- (3,0);
			\draw [ultra thick] (3,-0.2) node[below]{$\tau_{n}$} -- (3,0.3);
			\draw [ultra thick, <-] (3,0) -- (4,0) node[below]{$t'$};
		\end{tikzpicture}
	\right) &= i \times
		\begin{tikzpicture}
			\draw [ultra thick] (1,-0.2) node[below]{$t$} -- (1,0.3);
			\draw [ultra thick, <-] (1,0) -- (2,0);
			\draw [ultra thick] (2,-0.2) node[below]{$\tau_2$} -- (2,0.3);
			\draw [ultra thick, dotted, <-] (2,0) -- (3,0);
			\draw [ultra thick] (3,-0.2) node[below]{$\tau_{n}$} -- (3,0.3);
			\draw [ultra thick, <-] (3,0) -- (4,0) node[below]{$t'$};
		\end{tikzpicture}
	+ \hat{\mathcal{G}}_0(t) \times
		\begin{tikzpicture}
			\draw [ultra thick, <-] (0,0) node[below]{$t$} -- (1,0);
			\draw [ultra thick] (1,-0.2) node[below]{$\tau_1$} -- (1,0.3);
			\draw [ultra thick, <-] (1,0) -- (2,0);
			\draw [ultra thick] (2,-0.2) node[below]{$\tau_2$} -- (2,0.3);
			\draw [ultra thick, dotted, <-] (2,0) -- (3,0);
			\draw [ultra thick] (3,-0.2) node[below]{$\tau_{n}$} -- (3,0.3);
			\draw [ultra thick, <-] (3,0) -- (4,0) node[below]{$t'$};
		\end{tikzpicture} 
	\notag\\
	&= 	\hat{\mathcal{V}}(t) \times
		\begin{tikzpicture}
			\draw [ultra thick, <-] (1,0) node[below]{$t$} -- (2,0);
			\draw [ultra thick] (2,-0.2) node[below]{$\tau_2$} -- (2,0.3);
			\draw [ultra thick, dotted, <-] (2,0) -- (3,0);
			\draw [ultra thick] (3,-0.2) node[below]{$\tau_{n}$} -- (3,0.3);
			\draw [ultra thick, <-] (3,0) -- (4,0) node[below]{$t'$};
		\end{tikzpicture}
		+ \hat{\mathcal{G}}_0(t) \times
		\begin{tikzpicture}
			\draw [ultra thick, <-] (0,0) node[below]{$t$} -- (1,0);
			\draw [ultra thick] (1,-0.2) node[below]{$\tau_1$} -- (1,0.3);
			\draw [ultra thick, <-] (1,0) -- (2,0);
			\draw [ultra thick] (2,-0.2) node[below]{$\tau_2$} -- (2,0.3);
			\draw [ultra thick, dotted, <-] (2,0) -- (3,0);
			\draw [ultra thick] (3,-0.2) node[below]{$\tau_{n}$} -- (3,0.3);
			\draw [ultra thick, <-] (3,0) -- (4,0) node[below]{$t'$};
		\end{tikzpicture} 
\notag
\end{align}
and, therefore,
\begin{align}
	i \frac{\partial}{\partial t} \hat{\mathcal{U}}_n(t, t') =&
		% Zero order Feynman diagram
		\hat{\mathcal{G}}_0(t) \times
		\begin{tikzpicture}
			\draw [ultra thick, <-] (0,0) node[below]{$t$} -- (1,0) node[below]{$t'$};
		\end{tikzpicture}
		% First order Feynman diagram
		+ \hat{\mathcal{V}}(t) \times
		\begin{tikzpicture}
			\draw [ultra thick, <-] (0,0) node[below]{$t$} -- (1,0) node[below]{$t'$};
		\end{tikzpicture}
		+ \hat{\mathcal{G}}_0(t) \times
		\begin{tikzpicture}
			\draw [ultra thick, <-] (0,0) node[below]{$t$} -- (1,0);
			\draw [ultra thick] (1,-0.2) node[below]{$\tau_1$} -- (1,0.3);
			\draw [ultra thick, <-] (1,0) -- (2,0) node[below]{$t'$};
		\end{tikzpicture}
		+ \cdots \notag\\
		% n-th order Feynman diagram
	& + \hat{\mathcal{V}}(t) \times
		\begin{tikzpicture}
			\draw [ultra thick, <-] (0,0) node[below]{$t$} -- (1,0);
			\draw [ultra thick] (1,-0.2) node[below]{$\tau_2$} -- (1,0.3);
			\draw [ultra thick, dotted,  <-] (1,0) -- (2,0);
			\draw [ultra thick] (2,-0.2) node[below]{$\tau_{n}$} -- (2,0.3);
			\draw [ultra thick, <-] (2,0) -- (3,0) node[below]{$t'$};
		\end{tikzpicture}
		+ \hat{\mathcal{G}}_0(t)  \times
		\begin{tikzpicture}
			\draw [ultra thick, <-] (0,0) node[below]{$t$} -- (1,0);
			\draw [ultra thick] (1,-0.2) node[below]{$\tau_1$} -- (1,0.3);
			\draw [ultra thick, dotted,  <-] (1,0) -- (2,0);
			\draw [ultra thick] (2,-0.2) node[below]{$\tau_{n}$} -- (2,0.3);
			\draw [ultra thick, <-] (2,0) -- (3,0) node[below]{$t'$};
		\end{tikzpicture} \notag\\
	=& \hat{\mathcal{G}}_0(t) \times \left(
		\begin{tikzpicture}
			\draw [ultra thick, <-] (0,0) node[below]{$t$} -- (1,0) node[below]{$t'$};
		\end{tikzpicture}
		+
		\begin{tikzpicture}
			\draw [ultra thick, <-] (0,0) node[below]{$t$} -- (1,0);
			\draw [ultra thick] (1,-0.2) node[below]{$\tau_1$} -- (1,0.3);
			\draw [ultra thick, <-] (1,0) -- (2,0) node[below]{$t'$};
		\end{tikzpicture}
		+ \cdots +
		\begin{tikzpicture}
			\draw [ultra thick, <-] (0,0) node[below]{$t$} -- (1,0);
			\draw [ultra thick] (1,-0.2) node[below]{$\tau_1$} -- (1,0.3);
			\draw [ultra thick, dotted,  <-] (1,0) -- (2,0);
			\draw [ultra thick] (2,-0.2) node[below]{$\tau_{n}$} -- (2,0.3);
			\draw [ultra thick, <-] (2,0) -- (3,0) node[below]{$t'$};
		\end{tikzpicture}
	\right) \notag\\ 
	&+ \hat{\mathcal{V}}(t) \times \left(
		\begin{tikzpicture}
			\draw [ultra thick, <-] (0,0) node[below]{$t$} -- (1,0) node[below]{$t'$};
		\end{tikzpicture}
		+ \underbrace{
		\begin{tikzpicture}
			\draw [ultra thick, <-] (0,0) node[below]{$t$} -- (1,0);
			\draw [ultra thick] (1,-0.2) node[below]{$\tau_2$} -- (1,0.3);
			\draw [ultra thick, <-] (1,0) -- (2,0) node[below]{$t'$};
		\end{tikzpicture}
		}_{\mbox{comes from the second order}}
		+ \cdots + 
		\begin{tikzpicture}
			\draw [ultra thick, <-] (0,0) node[below]{$t$} -- (1,0);
			\draw [ultra thick] (1,-0.2) node[below]{$\tau_2$} -- (1,0.3);
			\draw [ultra thick, dotted,  <-] (1,0) -- (2,0);
			\draw [ultra thick] (2,-0.2) node[below]{$\tau_{n}$} -- (2,0.3);
			\draw [ultra thick, <-] (2,0) -- (3,0) node[below]{$t'$};
		\end{tikzpicture}
	\right) \notag\\
	=& \hat{\mathcal{G}}_0(t) \hat{\mathcal{U}}_n(t, t') + 
	\hat{\mathcal{V}}(t) \left(
		\begin{tikzpicture}
			\draw [ultra thick, <-] (0,0) node[below]{$t$} -- (1,0) node[below]{$t'$};
		\end{tikzpicture}
		+ \underbrace{
		\begin{tikzpicture}
			\draw [ultra thick, <-] (0,0) node[below]{$t$} -- (1,0);
			\draw [ultra thick] (1,-0.2) node[below]{$\tau_2$} -- (1,0.3);
			\draw [ultra thick, <-] (1,0) -- (2,0) node[below]{$t'$};
		\end{tikzpicture}
		+ \cdots + 
		\begin{tikzpicture}
			\draw [ultra thick, <-] (0,0) node[below]{$t$} -- (1,0);
			\draw [ultra thick] (1,-0.2) node[below]{$\tau_1$} -- (1,0.3);
			\draw [ultra thick, dotted,  <-] (1,0) -- (2,0);
			\draw [ultra thick] (2,-0.2) node[below]{$\tau_{n-1}$} -- (2,0.3);
			\draw [ultra thick, <-] (2,0) -- (3,0) node[below]{$t'$};
		\end{tikzpicture}
		}_{\mbox{after renaming variables: $\tau_2\to\tau_1, \ldots,  \tau_{n}\to\tau_{n-1}$}}
	\right). \notag
\end{align}
Finally, we have arrived at the inhomogeneous equation of motion for the partial sum of the perturbation theory:
\BoxedEquation{\label{EqForPartialUNPertrubTheor}
	i \frac{\partial}{\partial t} \hat{\mathcal{U}}_n(t, t') = \hat{\mathcal{G}}_0(t) \hat{\mathcal{U}}_n(t, t') + \hat{\mathcal{V}}(t) \left(
		\begin{tikzpicture}
			\draw [ultra thick, <-] (0,0) node[below]{$t$} -- (1,0) node[below]{$t'$};
		\end{tikzpicture}
		+ \cdots +
		\begin{tikzpicture}
			\draw [ultra thick, <-] (0,0) node[below]{$t$} -- (1,0);
			\draw [ultra thick] (1,-0.2) node[below]{$\tau_1$} -- (1,0.3);
			\draw [ultra thick, dotted,  <-] (1,0) -- (2,0);
			\draw [ultra thick] (2,-0.2) node[below]{$\tau_{n-1}$} -- (2,0.3);
			\draw [ultra thick, <-] (2,0) -- (3,0) node[below]{$t'$};
		\end{tikzpicture}
	\right).
}
Equation (\ref{EqForPartialUNPertrubTheor}) allows us to see the condition of applicability of the infinite series (\ref{EqTimeDepPertTheoryPrior}) or (\ref{EqFeynmanTimeDepPertTheoryPrior}),
\begin{align}
	i \frac{\partial}{\partial t} \hat{\mathcal{U}}_n(t, t') =& \hat{\mathcal{G}}_0(t) \hat{\mathcal{U}}_n(t, t') + \hat{\mathcal{V}}(t) \left( 
		\hat{\mathcal{U}}_n(t, t')
		- 
		% n-th order Feynman diagram
		\begin{tikzpicture}
			\draw [ultra thick, <-] (0,0) node[below]{$t$} -- (1,0);
			\draw [ultra thick] (1,-0.2) node[below]{$\tau_1$} -- (1,0.3);
			\draw [ultra thick, dotted,  <-] (1,0) -- (2,0);
			\draw [ultra thick] (2,-0.2) node[below]{$\tau_{n}$} -- (2,0.3);
			\draw [ultra thick, <-] (2,0) -- (3,0) node[below]{$t'$};
		\end{tikzpicture}
	\right) \notag\\
	= & \left[ \hat{\mathcal{G}}_0(t) + \hat{\mathcal{V}}(t) \right]  \hat{\mathcal{U}}_n(t, t') - i \times  
		\begin{tikzpicture}
			\draw [ultra thick] (0,-0.2) node[below]{$t$} -- (0,0.3);
			\draw [ultra thick, <-] (0,0) -- (1,0);
			\draw [ultra thick] (1,-0.2) node[below]{$\tau_1$} -- (1,0.3);
			\draw [ultra thick, dotted,  <-] (1,0) -- (2,0);
			\draw [ultra thick] (2,-0.2) node[below]{$\tau_{n}$} -- (2,0.3);
			\draw [ultra thick, <-] (2,0) -- (3,0) node[below]{$t'$};
		\end{tikzpicture}.
	\notag
\end{align}
Comparing this with the equation of motion (\ref{EqUU0EvolutionEq}) for $\hat{\mathcal{U}}$, we find the necessary condition for convergence of the perturbation theory series:
\begin{align}
	\hat{\mathcal{U}}(t, t') = \lim_{n\to \infty} \hat{\mathcal{U}}_n(t, t') \quad \Longrightarrow \quad 
	\lim_{n\to \infty} 
		\begin{tikzpicture}
			\draw [ultra thick] (0,-0.2) node[below]{$t$} -- (0,0.3);
			\draw [ultra thick, <-] (0,0) -- (1,0);
			\draw [ultra thick] (1,-0.2) node[below]{$\tau_1$} -- (1,0.3);
			\draw [ultra thick, dotted,  <-] (1,0) -- (2,0);
			\draw [ultra thick] (2,-0.2) node[below]{$\tau_{n}$} -- (2,0.3);
			\draw [ultra thick, <-] (2,0) -- (3,0) node[below]{$t'$};
		\end{tikzpicture}
	= 0,
\end{align}
which may be satisfied if the perturbation $\hat{\mathcal{V}}$ is small.

\chapter{4. Prerequisites for quantum dynamics analysis}\label{Chapter:4}
\section{Processing dynamically measured data}\label{Sec:DynamicThoughtExp}

\begin{figure*}
	\begin{center}
		\includegraphics[width=0.8\hsize]{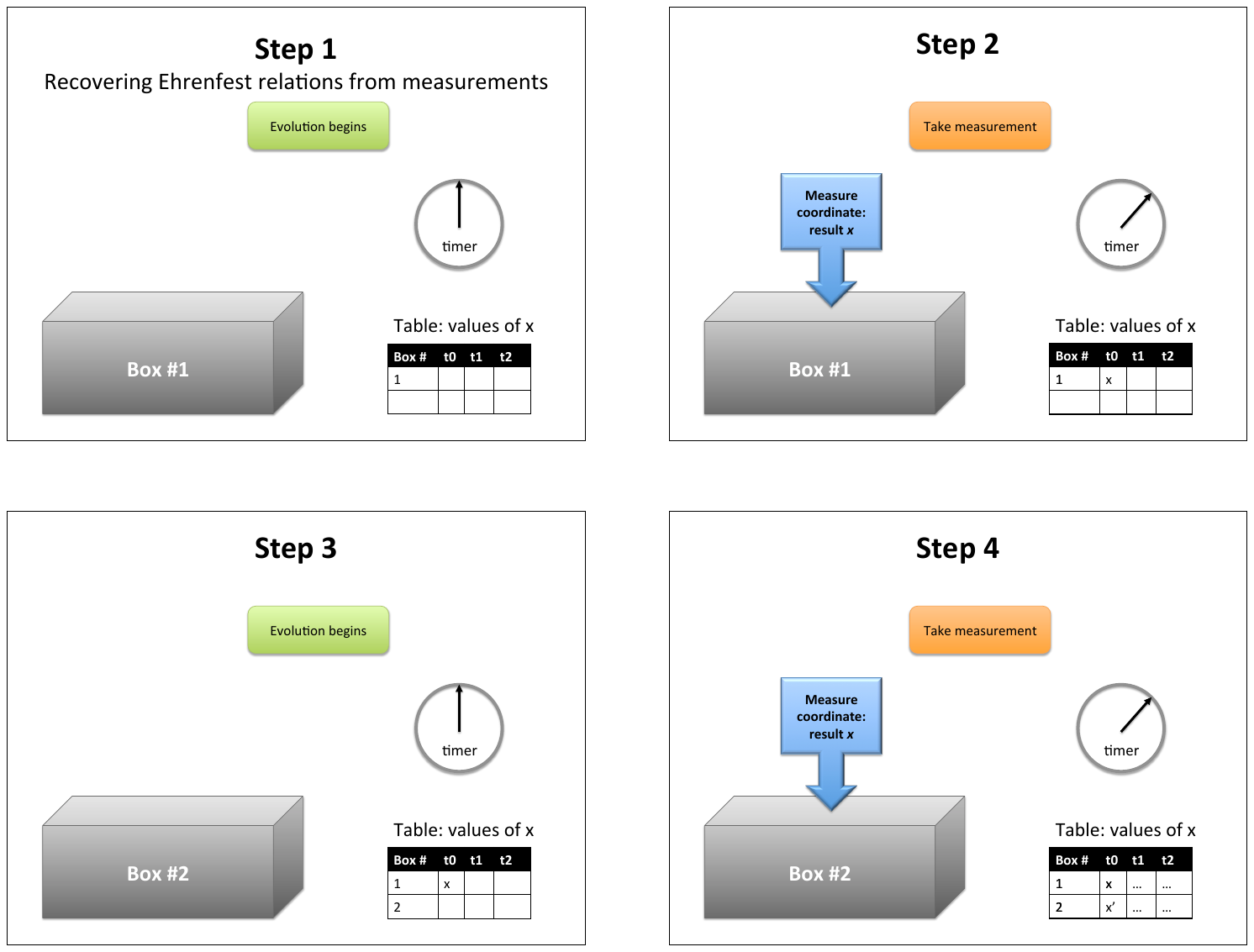}
		\caption{Thought experiments to recover the Ehrenfest relations (\ref{AveragedEquationsForXandP}) from experiment.}\label{Fig_Measurement2Ehrenfest}
	\end{center}
\end{figure*}

Assume we have multiple copies of either a quantum or classical system  (without loss of generality,, we consider single-particle one-dimensional systems throughout). Suppose we can precisely measure different copies of the particle's coordinate $x$ and momentum $p$ at times $\{ t_k \}_{k=1}^K$.  Upon performing ideal measurements of the coordinate or momentum on the $n$-th copy, we experimentally obtain $\{ x_n(t_k) \}$
 and $\{ p_n(t_k) \}$, $n=1,\ldots,N$, requiring a total of $2KN$ observations. Time interpolation of these data points returns the functions $x_n(t)$ and $p_n(t)$. We may then calculate the  statistical moments 
$ \overline{ [x(t)]^l } = \frac{1}{N} \sum_{n=1}^N [x_n (t)]^l$ and
$\overline{ [p(t)]^l } = \frac{1}{N} \sum_{n=1}^N [p_n (t)]^l$ for $l=1,2,3,\ldots$  We make the {\it ansatz}, resembling a Taylor series with coefficients $a_l$, $b_l$, $c_{k,l}$, $d_l$, $e_l$, and $f_{k,l}$, that the first derivative of $\overline{x(t)} = \overline{ [x(t)]^1 }$ and $\overline{p(t)} = \overline{ [p(t)]^1 }$ satisfy
\begin{align}
	\frac{d}{dt} \overline{x(t)} &= \sum_l \left( a_l \overline{[x(t)]^l} + b_l \overline{[p(t)]^l} \right)
		+ \sum_{k,l \neq 0} c_{k,l} \overline{[x(t)]^l [p(t)]^k}, \notag\\
	\frac{d}{dt} \overline{p(t)} &= \sum_l \left( d_l \overline{[x(t)]^l} + e_l \overline{[p(t)]^l} \right)
		+ \sum_{k,l \neq 0} f_{k,l} \overline{[x(t)]^l [p(t)]^k}. \notag
\end{align}
For non-dissipative quantum and classical systems, these relations reduce to
\begin{align}\label{AveragedEquationsForXandP}
	m\frac{d}{dt} \overline{x(t)} = \overline{p(t)}, \qquad \frac{d}{dt} \overline{p(t)} = \overline{-U'(x)}(t),
\end{align}
where $\overline{-U'(x)}(t) = \sum_l d_l  \overline{[x(t)]^l}$.

\section{Quantum dynamics: Inference of Schr\"{o}dinger equation}\label{Sec:ODMSchrodinger} 

The evolution of expectation values of the quantum coordinate and momentum is governed by the Ehrenfest theorems (\ref{AveragedEquationsForXandP}). 
\begin{align}\label{EqEhrenfestInHilbertSpace}
	m \frac{d}{dt} \langle \hat{x} \rangle = \langle \hat{p} \rangle, &\qquad
	\frac{d}{dt} \langle \hat{p} \rangle = \langle -U'(\hat{x}) \rangle, \notag\\
	m \frac{d}{dt} \bra{\Psi(t)} \hat{x} \ket{\Psi(t)} = \bra{\Psi(t)} \hat{p} \ket{\Psi(t)}, &\qquad
	\frac{d}{dt} \bra{\Psi(t)} \hat{p} \ket{\Psi(t)} = \bra{\Psi(t)} -U'(\hat{x}) \ket{\Psi(t)}. 
\end{align}
According to Stone's theorem (\ref{EqStoneTheorem}), there exits an \emph{unknown} generator of motion $\hat{H}$ such that
\begin{align}\label{AbstractSchrodingerEq}
	i\hbar \ket{d \Psi(t)/dt} = \hat{H} \ket{\Psi(t)}, \qquad
	-i\hbar \bra{d \Psi(t)/dt} =  \bra{\Psi(t)} \hat{H}.
\end{align}
Thus, Eq. (\ref{EqEhrenfestInHilbertSpace}) can be rewritten as
\begin{align}
	m \frac{d}{dt} \bra{\Psi(t)} \hat{x} \ket{\Psi(t)} 
	&= m\bra{d\Psi(t)/dt} \hat{x} \ket{\Psi(t)} + m\bra{\Psi(t)} \hat{x} \ket{d\Psi(t) / dt} \notag\\
	&= \frac{im}{\hbar} \bra{\Psi(t)} \hat{H} \hat{x} \ket{\Psi(t)} - \frac{im}{\hbar} \bra{\Psi(t)} \hat{x} \hat{H} \ket{\Psi(t)} \notag\\
	&= \frac{im}{\hbar} \bra{\Psi(t)} [\hat{H}, \hat{x}] \ket{\Psi(t)} = \bra{\Psi(t)} \hat{p} \ket{\Psi(t)}.
\end{align}
The averaging can dropped since the equality should be valid for any time $t$ as well as any initial condition
\begin{align}\label{Commut_Eq_for_Schrod}
	im \commut{\hat{H}, \hat{x}} = \hbar \hat{p}, \qquad i\commut{\hat{H}, \hat{p}} = -\hbar U'(\hat{x}).
\end{align}
Now we are going to make use of the canonical commutation relations (\ref{XP_CommutationalRelation}), $\commut{ \hat{x}, \hat{p} } = i\hbar$, to transform the commutator equations (\ref{Commut_Eq_for_Schrod}) with respect to the unknown $\hat{H}$ into partially differential equations. Assuming $\hat{H} = H(\hat{x}, \hat{p})$, we make use of Eq. (\ref{EqWelCommutatorTheorem}) 
\begin{align}
	im \frac{\partial}{\partial \hat{x}} H(\hat{x}, \hat{p}) [\hat{x}, \hat{x}] + im \frac{\partial}{\partial \hat{p}} H(\hat{x}, \hat{p}) [\hat{p}, \hat{x}] = \hbar \hat{p}, &\quad
	i \frac{\partial}{\partial \hat{x}} H(\hat{x}, \hat{p}) [\hat{x}, \hat{p}] + i \frac{\partial}{\partial \hat{p}} H(\hat{x}, \hat{p}) [\hat{p}, \hat{p}] = -\hbar U'(\hat{x}) \notag\\
    \Longrightarrow
	m \frac{\partial}{\partial \hat{p}} H(\hat{x}, \hat{p}) = \hat{p}, &\qquad
	\frac{\partial}{\partial \hat{x}} H(\hat{x}, \hat{p}) = U'(\hat{x}).
\end{align}
Furthermore, using Eq. (\ref{Weyl_derivative_of_function}), we obtain differentials with respect to the ordinary functions
\begin{align}
	m \frac{\partial}{\partial p} H(x, p) = p, \qquad \frac{\partial}{\partial x} H(x, p) = U'(x).
\end{align}
which can be readily solved 
\begin{align}
	m \frac{\partial}{\partial p} H(x, p) = p \Longrightarrow H(x, p) = \frac{p^2}{2m} + F(x) 
	\Longrightarrow \frac{\partial}{\partial x} H(x, p) = F'(x) = U'(x).
\end{align}
Whence, the familiar quantum Hamiltonian readily follows
\BoxedEquation{\label{Quantum_Hamiltonian}
	\hat{H} = \frac{\hat{p}^2}{2m} + U(\hat{x}).
} 
Since the Schr\"{o}dinger equation was derived from the Ehrenfest theorems (\ref{AveragedEquationsForXandP}) assuming the canonical commutation relation (\ref{XP_CommutationalRelation}), the presentation suggests that the Ehrenfest theorems are more fundamental than the Schr\"{o}dinger equation.

\section{Solving numerically time-independent Schr\"{o}dinger equation}\label{Sec:StationarySchNumeSol} 

In this section, we shall review a different method for finding eigenvalues (i.e., energy spectrum) and eigenfunctions of the Hamiltonian $\hat{H}$ in the coordinate representation [see Eq. (\ref{EqSummaryCommutatorSection})] utilizing atomic units where $\hbar=m=1$
\begin{align}\label{ExactHamiltonianCoordinateRep}
	\langle x | \hat{H} | \psi \rangle =\left[ -\frac{1}{2}\frac{\partial^2}{\partial x^2} + U(x) \right] \langle x | \psi \rangle = E \langle x | \psi \rangle.
\end{align}

\subsection{Forward and backward difference methods}

We shall represent wave function as a vector $\langle x_k | \psi \rangle$ on equally spaced coordinate grid $x_k = k \Delta x$, where $k=0,1,2,\ldots,N-1$ is an integer. According to the forward difference approximation,  the derivative of a function $f(x)$ is approximated by
\begin{align}
	f'(x) \approx [ f(x+\Delta x) - f(x) ] / \Delta x = \left[ f(x) + f'(x) \Delta x + O\left(\Delta x^2 \right) - f(x) \right] / \Delta x
	= f'(x) + O\left( \Delta x \right),
    \label{EqForwardDiffApprox} 
\end{align}
which gives us access to the second derivative,
\begin{align}
	f''(x) \approx [ f'(x+\Delta x) - f'(x) ] / \Delta x  = [ f(x+2\Delta x) - 2f(x + \Delta x) + f(x)] / \Delta x^2 + O\left( \Delta x \right). 
\end{align}
The discretized Eq. (\ref{ExactHamiltonianCoordinateRep}) reads
\begin{align}\label{ForwardDiffHamiltonianDiscreet}
	-\frac{\langle x_{k+2} | \psi \rangle - 2\langle x_{k+1} | \psi \rangle + \langle x_k | \psi \rangle}{2\Delta x^2} + U(x_k)  \langle x_{k} | \psi \rangle = E  \langle x_k | \psi \rangle,
	\qquad k=0,1,2,\ldots,N-1.
\end{align}
In other words, Eq.~(\ref{ExactHamiltonianCoordinateRep}) turns into the following matrix eigenvalue-eigenfunction problem\footnote{See the implementation \url{https://github.com/dibondar/QuantumClassicalDynamics/blob/master/demo_finite_difference.ipynb}}:
\begin{align}\label{ForwardDiffHamiltonianCoordinateRep}
	\left[
		\frac{-1}{2 \Delta x^2} 
		\begin{pmatrix}
			1	& -2	& 1 	&	&	\\% & \\
				& 1	& -2	& 1	&	\\% & \\
				&	& \ddots& \ddots & \ddots \\% & \\
				&	&	& 1	& -2	\\% & 1 \\	
				&	&	&	& 1	\\% & -2 \\
				%&	&	&	&	& 1 \\
		\end{pmatrix}
		+
		\begin{pmatrix}
			U(x_0)	&	&	&	&	& \\
				& U(x_1) &	&	&	& \\
				& 	& 	& \ddots &	& \\
				&	& 	& 	& U(x_{N-2})  & \\
				&	&	& 	& 	& U(x_{N-1})
		\end{pmatrix}
	\right]
	\begin{pmatrix}
	 \langle x_0 | \psi \rangle \\
	 \langle x_1 | \psi \rangle \\
	 \vdots \\
	 \langle x_{N-2} | \psi \rangle \\
	 \langle x_{N-1} | \psi \rangle
	\end{pmatrix}
 \notag\\
	= E
	\begin{pmatrix}
	 \langle x_0 | \psi \rangle \\
	 \langle x_1 | \psi \rangle \\
	 \vdots \\
	 \langle x_{N-2} | \psi \rangle \\
	 \langle x_{N-1} | \psi \rangle
	\end{pmatrix}
\end{align} 
Note that in order to transform Eq.~(\ref{ForwardDiffHamiltonianDiscreet}) into Eq.~(\ref{ForwardDiffHamiltonianCoordinateRep}), the following additional boundary conditions have to be used
\begin{align}
	\langle x_N | \psi \rangle = \langle x_{N+1} | \psi \rangle = 0.
\end{align}
This condition was chosen from the fact that the exact eigenfunction $\langle x | \psi \rangle$ for a bound state must be a square integrable, therefore $\langle \pm\infty | \psi \rangle = 0$. Different boundary conditions must be imposed if we want to find the wavefunction for continuous eigenstate (i.e., scattering states) or to find stationary states for solid state systems.
%(see Sec. \ref{Sec:SolidState}).

Even though Eq. (\ref{ExactHamiltonianCoordinateRep}) contains the self-adjoint Hamiltonian, the discretization with the forward difference approximation (\ref{ForwardDiffHamiltonianCoordinateRep}) is non-Hermitian, thus eigenenergies are not guaranteed to be real. This shows that not all numerical methods preserve physical structure!

The same conclusion holds when we use the backward difference method\footnote{See the implementation \url{https://github.com/dibondar/QuantumClassicalDynamics/blob/master/demo_finite_difference.ipynb}}:
\begin{align}
	f'(x) &\approx [ f(x) - f(x-\Delta x) ] / \Delta x = f'(x) + O\left( \Delta x \right) \\
	f''(x) &\approx [ f'(x) - f'(x-\Delta x) ] / \Delta x  = [ f(x) - 2f(x - \Delta x) + f(x- 2\Delta x)] / \Delta x^2 + O\left( \Delta x \right). 
\end{align}
In fact both forward and backward finite difference methods lead to disastrous results! 
Can you see why?\footnote{Answer: In the case of the forward (backward) difference approximation, we get lower (upper) triangular Hamiltonians with the diagonal values of $-1/(2\Delta x^2) + U(x_k)$. Recall that the eigenvalues of an upper or lower triangular matrix lie on the diagonal.}

\subsection{Central difference method}

However, the central difference method
\begin{align}
	f'(x) &\approx [ f(x+\Delta x/2) - f(x-\Delta x/2) ] / \Delta x = f'(x) + O\left( \Delta x^2 \right) \Longrightarrow \label{CentralFinitDiffApprox} \\
	f''(x) &\approx  [ f'(x+\Delta x/2) - f'(x-\Delta x/2) ] / \Delta x  = [ f(x + \Delta x) - 2f(x) + f(x- \Delta x)] / \Delta x^2  + O\left( \Delta x^2 \right). 
\end{align}
leads to the Hermitian finite dimensional approximation of the Hamiltonian:
\begin{align}\label{CentralDiffHamiltonianDiscreet}
	-\frac{\langle x_{k+1} | \psi \rangle - 2\langle x_{k} | \psi \rangle + \langle x_{k-1} | \psi \rangle}{2\Delta x^2} + U(x_k)  \langle x_k | \psi \rangle = E  \langle x_k | \psi \rangle,
	\qquad k=0,1,2,\ldots,N-1,
\end{align}
\begin{align}\label{CentralDiffHamiltonianCoordinateRep}
	\left[
		\frac{-1}{2 \Delta x^2} 
		\begin{pmatrix}
			-2	& 1 &	&	& \\
			1	& -2	& 1 &	& \\
				 & \ddots & \ddots & \ddots & \\
				&	& 1	& -2	& 1 \\
				&	&	& 1	& -2
		\end{pmatrix}
		+
		\begin{pmatrix}
			U(x_0)	&	&	&	&	& \\
				& U(x_1) &	&	&	& \\
				& 	& 	& \ddots &	 & \\
				&	& 	& 	& U(x_N)  & \\
				&	&	& 	& 	& U(x_{N-1})
		\end{pmatrix}
	\right]
	\begin{pmatrix}
	 \langle x_0 | \psi \rangle \\
	 \langle x_1 | \psi \rangle \\
	 \vdots \\
	 \langle x_{N-2} | \psi \rangle \\
	 \langle x_{N-1} | \psi \rangle
	\end{pmatrix}
  \notag\\
	= E
	\begin{pmatrix}
	 \langle x_0 | \psi \rangle \\
	 \langle x_1 | \psi \rangle \\
	 \vdots \\
	 \langle x_{N-2} | \psi \rangle \\
	 \langle x_{N-1} | \psi \rangle
	\end{pmatrix}.
\end{align} 
To rewrite Eq. (\ref{CentralDiffHamiltonianDiscreet}) in the form of Eq. (\ref{CentralDiffHamiltonianCoordinateRep}), the following boundary condition must be imposed 
\begin{align}
	 \langle x_{-1} | \psi \rangle =  \langle x_N | \psi \rangle = 0,
\end{align}
which is a grid representation of $\langle \pm\infty | \psi \rangle = 0$.

Thus, the central finite difference preserves the physical structure\footnote{See the  implementation \url{https://github.com/dibondar/QuantumClassicalDynamics/blob/master/demo_finite_difference.ipynb}} and it leads to reliable numerical results.

\subsection{Mutually unbiased bases}

In linear algebra, two orthonormal bases $| p_1 \rangle, | p_2 \rangle, \ldots, | p_N \rangle$ and $| x_1 \rangle, | x_2 \rangle, \ldots, | x_N \rangle$ are called mutually unbiased bases if $\left| \langle x_k | p_l \rangle \right|^2 = \mbox{const}$. The latter is the relationship we observed in our thought experiment [see Eq. (\ref{quantum_xp_const})]. To generate such bases, one can employ the methods outlined, e.g., in Refs. \cite{Schwinger2003, Torre2003}. We will follow a more numerically convenient route.

We introduce the dedicated notation for the direct ($\mathcal{F}$) and inverse ($\mathcal{F}^{-1}$) Fourier transform
\BoxedEquation{\label{EqContiniousFourierTransform}
	\mathcal{F}_{x \to p}[ f(x) ] = \int \frac{dx}{\sqrt{2\pi\hbar}} e^{-ixp/\hbar} f(x) = \int dx \, \langle p | x \rangle f(x), 
     \notag\\
	\mathcal{F}_{p \to x}^{-1}[ g(p) ] = \int \frac{dp}{\sqrt{2\pi\hbar}} e^{ixp/\hbar} g(p)  = \int dp \, \langle x | p \rangle g(p).
}
We can obtain 
\begin{align}
	\langle x | p \rangle =  \langle x | \left( \int dx' | x' \rangle\langle x' | \right) | p \rangle 
		= \int dx' \delta(x-x') \langle x' | p \rangle = \mathcal{F}_{x' \to p}^{-1}\left[ \delta(x-x') \right]. 
\end{align}

Finally we can represent the Hamiltonian $\hat{H}$ of a very general form
\begin{align}
	\hat{H} &= K(\hat{p}) + U(\hat{x}), \notag\\
	\langle x | \hat{H} | x' \rangle 
		& = \langle x | \left( \int dp | p \rangle\langle p | \right) K(\hat{p}) \left( \int dp' | p' \rangle\langle p' | \right) | x' \rangle + \langle x | U(\hat{x}) | x' \rangle  \notag\\
		& = \int dp dp' \langle x | p \rangle K(p) \delta(p-p') \langle p' | x' \rangle + U(x) \delta(x-x'), 
\end{align}
\BoxedEquation{\label{EqMUBHamiltonianXX}
	\langle x | \hat{H} | x' \rangle = \mathcal{F}_{p \to x}^{-1} \Big[ \mathcal{F}_{p' \to x'}[ K(p) \delta(p-p') ] \Big] + U(x) \delta(x-x').
}
In order to make use of these identities, we need to be able to efficiently calculate the Fourier transform.

\section{Method 1: Calculating the continuous Fourier transform via the Fast Fourier Transform}\label{Sec_FTviaFFT}

The algorithm of Fast Fourier Transform (FFT), considered to be one of the best of the 20th century \cite{cipra_best_2000}, efficiently calculates the following sums
\BoxedEquation{
	y_k = FFT_{l \to k} [x_l] = \sum_{l=0}^{N-1} x_l e^{-2\pi i kl/N}, \quad
	x_k = iFFT_{l \to k}[y_l] = \frac{1}{N} \sum_{l=0}^{N-1} y_l e^{2\pi i kl/N}, 
	\quad k = 0, \ldots, N-1.
}
The FFT and its inverse iFFT perform these calculations with complexity $O(N\log N)$ instead of $O(N^2)$ if the summation is done straightforwardly. The best performance of FFT/iFFT is achieved when $N$ is a power of $2$. \texttt{SciPy} contains a whole suite devoted to these transformations\footnote{\url{http://docs.scipy.org/doc/scipy/reference/tutorial/fftpack.html}}. One might expect that FFT and iFFT are merely straightforward implementations of $\mathcal{F}$ and $\mathcal{F}^{-1}$ (\ref{EqContiniousFourierTransform}), respectively. This intuitive expectation disappears by running this demo\footnote{\url{https://github.com/dibondar/QuantumClassicalDynamics/blob/master/fourier_transform.py}}.

However, Ref. \cite{bailey1994fast} provides methods how to utilize FFT to approximate the continuous Fourier transform $g(p) = \mathcal{F}_{x \to p}[ f(x) ]$.
Let us assume $f(x)$ vanishes outside the interval $[-L, +L)$. Define the following grid representations for the continuous variables $x$ and $p$ 
\BoxedEquation{\label{EqCoordinateMomentumGrids1}
	x_k = \left(k - N / 2 \right) \Delta x, \quad p_k = \left(k - N/2 \right)\pi / L,
	\qquad 
	k = 0, \ldots, N-1, \quad \Delta x = 2L / N,
}
in particular $[x_k] = [-L, -L + \Delta x, -L + 2\Delta x, \ldots, L - \Delta x]$ note that the last point, $+L$, is missing. 
\begin{align}
	g(p_k) &= \int_{-\infty}^{+\infty} f(x) e^{-ixp_k} dx =  \int_{-L}^{+L} f(x) e^{-ixp_k} dx 
		\approx \sum_{l=0}^{N-1} f(x_l) e^{-ix_l p_k} \Delta x \notag\\
		&= \Delta x \sum_{l=0}^{N-1} f(x_l) e^{-i(l - N/2)(k - N/2) \frac{\pi}{L} \frac{2L}{N}} 
		= \Delta x e^{-i\pi N/2} e^{i\pi k} \sum_{l=0}^{N-1} f(x_l) e^{i\pi l} e^{-i2\pi kl/N}. 
\end{align}
\BoxedEquation{\label{EqCFTviaFFT}
	g(p_k) =  (-1)^k FFT_{l \to k} \left[ (-1)^l f(x_l) \right] \Delta x , \quad
	f(x_k) =  (-1)^k iFFT_{l \to k} \left[ (-1)^l g(p_l) \right] / \Delta x, \notag\\
	 \mbox{for $N$ divisible by 4,} \notag\\
	 \mbox{or symbolically } \mathcal{F}_{l \to k} \approx (-)^k FFT_{l \to k} (-)^l, 
	 \qquad  \mathcal{F}_{l \to k}^{-1} \approx (-)^k iFFT_{l \to k} (-)^l, 
}
Using this computational trick, Eq. (\ref{EqMUBHamiltonianXX}) is finally can be written as
\begin{align}
	\langle x_k | \hat{H} | x_{k'} \rangle &= (-)^{k} iFFT_{l \to k} \Big[ (-)^{l} (-)^{k'} FFT_{l' \to k'} \left[ (-)^{l'} K(p_l) \delta_{l, l'} \right] \Big] + U(x_k) \delta_{k, k'} \notag\\
	&= (-)^{k + k'} iFFT_{l \to k} \Big[ FFT_{l' \to k'} \left[ (-)^{l + l'} K(p_l) \delta_{l, l'} \right] \Big] + U(x_k) \delta_{k, k'},
\end{align}
which is utilized for diagonalizing the Hamiltonian in the example\footnote{\url{https://github.com/dibondar/QuantumClassicalDynamics/blob/master/demo_mub_qhamiltnian.ipynb}}.

Note that in this approach [see Eq. (\ref{EqCoordinateMomentumGrids1})] we are free to choose the coordinate grid, the momentum grid is nevertheless fixed. 

\section{Method 2: Calculating the continuous Fourier transform via the Fractional Fourier Transform}\label{Sec_FTviaFRFT}

Employing FFT an $O(N \log N)$ algorithm can be implemented to calculate the \emph{Fractional Fourier Transform} (FRFT) defined by the sum
\BoxedEquation{
	y_k^{(\alpha)} = FRFT_{l \to k}^{(\alpha)} \left[ x_l \right]= \sum_{l=0}^{N-1} x_l e^{-2\pi i k l \alpha}, \qquad k = 0, \ldots, N-1,	
}
for some fixed real value $\alpha$. 

Using FRFT, we can overcome the shortcoming of the previous method that the step size of the $p$ grid must be fixed by the parameter of the grids. Define the following grid representations for the continuous variables $x$ and $p$:
\BoxedEquation{\label{EqCoordinateMomentumGrids2}
	x_k = \left(k - N / 2 \right) \Delta x, \quad p_k = \left(k - N/2 \right)\Delta p,
	\qquad 
	k = 0, \ldots, N-1, \quad \Delta x = 2L / N.
}
Note that both $\Delta x$ and $\Delta p$ can be made arbitrary.
\begin{align}
	g(p_k) &= \int_{-\infty}^{+\infty} f(x) e^{-ixp_k} dx =  \int_{-L}^{+L} f(x) e^{-ixp_k} dx 
		\approx \sum_{l=0}^{N-1} f(x_l) e^{-ix_l p_k} \Delta x \notag\\
		&= \Delta x \sum_{l=0}^{N-1} f(x_l) e^{-i(l - N/2)(k - N/2) \Delta x \Delta p} 
		= \Delta x e^{\pi i (k - N/2)N \delta}  \sum_{l=0}^{N-1} f(x_l) e^{\pi i l N \delta} e^{-2\pi i k l \delta}, 
			\qquad \delta = \frac{\Delta x \Delta p}{2 \pi}.
\end{align}
\BoxedEquation{
	g(p_k) = \Delta x e^{\pi i (k - N/2)N \delta} FRFT_{l \to k}^{(\delta)} \left[ f(x_l) e^{\pi i l N \delta}  \right], 
		\qquad \delta = \frac{\Delta x \Delta p}{2 \pi}, \notag\\
	 \mbox{or symbolically } \mathcal{F}_{l \to k} \approx e^{\pi i (k - N/2)N \delta} FRFT_{l \to k}^{(\delta)} e^{\pi i l N \delta}.
}

\chapter{5. Time-dependent Schr\"{o}\-din\-ger evolution}\label{Chapter:5}
\section{The Trotter product formula and split-operator method}\label{Sec:SplitOpSchrodinger}

The term ``split-operator'' comes from the fact that we want to splits the single operator $e^{\hat{A} + \hat{B}}$ into constitute components $e^{\alpha \hat{A}}$ and $e^{\beta \hat{B}/N}$ that are easier to evaluate. For example, such a splitting is done in the Trotter formula (\ref{TrotterProductEq}). Below we shall make use of the following identity, which is easily proven by the Taylor expanding both sides (always remember $\hat{A}\hat{B} \neq \hat{B}\hat{A}$):
\BoxedEquation{\label{SecondOrderSplitting}
	e^{\hat{A}\delta + \hat{B}\delta} = e^{\hat{B}\delta/2}e^{\hat{A}\delta}e^{\hat{B}\delta/2} + O\left(\delta^3\right).
}

Now, assuming $\hat{H}(t) = K(t, \hat{p}) + U(t, \hat{x})$ and using Eqs. (\ref{EqTrapezoidalIntegralRule}), (\ref{EqTimeOrderingDropping}), and (\ref{SecondOrderSplitting}), we get
\begin{align}\label{SecondOrderTexpPartEq}
	\hat{U}(t + \delta t, t) = \exp\left[- \frac{i}{\hbar} \int_{t}^{t + \delta t} \hat{H}(\tau) d\tau \right] + O\left( \delta t^3 \right)
		= \exp\left[- \frac{i}{\hbar} \hat{H}\left( t + \frac{\delta t}{2} \right) \delta t + O\left( \delta t^3 \right) \right] + O\left( \delta t^3 \right) \notag
        \\
		=  \underbrace{\exp\left[ -i \frac{\delta t}{2\hbar} U\left(t + \frac{\delta t}{2}, \hat{x}\right) \right] 
			\exp\left[ -i \frac{\delta t}{\hbar} K\left(t + \frac{\delta t}{2}, \hat{p}\right) \right]
			\exp\left[ -i \frac{\delta t}{2\hbar} U\left(t + \frac{\delta t}{2}, \hat{x}\right) \right]}_{\mbox{This is unitary, see Eq. \eqref{EqUnitaritySplitOpMethodSchrodinger} below.}} 
		+ O\left( \delta t^3 \right). 
\end{align}
\begin{align}\label{EqUnitaritySplitOpMethodSchrodinger}
	& e^{-i \frac{\delta t}{2\hbar} U\left(t + \frac{\delta t}{2}, \hat{x}\right)} 
	e^{-i \frac{\delta t}{\hbar} K\left(t + \frac{\delta t}{2}, \hat{p}\right)}
	e^{-i \frac{\delta t}{2\hbar} U\left(t + \frac{\delta t}{2}, \hat{x}\right)}
	\left[ 
		e^{-i \frac{\delta t}{2\hbar} U\left(t + \frac{\delta t}{2}, \hat{x}\right)} 
		e^{-i \frac{\delta t}{\hbar} K\left(t + \frac{\delta t}{2}, \hat{p}\right)}
		e^{-i \frac{\delta t}{2\hbar} U\left(t + \frac{\delta t}{2}, \hat{x}\right)}
	\right]^{\dagger} \notag\\
	&~~~= e^{-i \frac{\delta t}{2\hbar} U\left(t + \frac{\delta t}{2}, \hat{x}\right)} 
	e^{-i \frac{\delta t}{\hbar} K\left(t + \frac{\delta t}{2}, \hat{p}\right)}
	\underbrace{
		e^{-i \frac{\delta t}{2\hbar} U\left(t + \frac{\delta t}{2}, \hat{x}\right)}
		e^{i \frac{\delta t}{2\hbar} U\left(t + \frac{\delta t}{2}, \hat{x}\right)}
	 }_{1}
	e^{i \frac{\delta t}{\hbar} K\left(t + \frac{\delta t}{2}, \hat{p}\right)}
	e^{i \frac{\delta t}{2\hbar} U\left(t + \frac{\delta t}{2}, \hat{x}\right) } \notag\\
	&~~~= e^{-i \frac{\delta t}{2\hbar} U\left(t + \frac{\delta t}{2}, \hat{x}\right)} 
	\underbrace{
		e^{-i \frac{\delta t}{\hbar} K\left(t + \frac{\delta t}{2}, \hat{p}\right)}
		e^{i \frac{\delta t}{\hbar} K\left(t + \frac{\delta t}{2}, \hat{p}\right)}
	 }_{1}
	e^{i \frac{\delta t}{2\hbar} U\left(t + \frac{\delta t}{2}, \hat{x}\right)} 
	= \hat{1}.
\end{align}

Using the coordinate representation
\begin{align}
	\langle x &\ket{\psi(t + \delta t)} = \bra{x} \hat{U}(t + \delta t, t) \ket{\psi(t) } \notag\\
		&= \bra{x} e^{-i \frac{\delta t}{2\hbar} U\left(t + \frac{\delta t}{2}, \hat{x}\right)} 
			e^{-i \frac{\delta t}{\hbar} K\left(t + \frac{\delta t}{2}, \hat{p}\right)} \hat{1}
			e^{-i \frac{\delta t}{2\hbar} U\left(t + \frac{\delta t}{2}, \hat{x}\right)} \hat{1} \ket{\psi(t)}  + O\left( \delta t^3 \right) \notag\\
		&= e^{ -i \frac{\delta t}{2\hbar} U\left(t + \frac{\delta t}{2}, x\right) } \bra{x}
			e^{-i \frac{\delta t}{\hbar} K\left(t + \frac{\delta t}{2}, \hat{p}\right)} \int dp |p \rangle\langle p|
			e^{-i \frac{\delta t}{2\hbar} U\left(t + \frac{\delta t}{2}, \hat{x}\right)} 
			\int dx' |x' \rangle\langle x' \ket{\psi(t)} + O\left( \delta t^3 \right), \notag
\end{align}
where the last line is obtained by using Eq. \eqref{EqFunctionActingOnEigenVals}.
Therefore,
\BoxedEquation{\label{SpectralSplitOpMEq}
	\langle x \ket{\psi(t + \delta t)} = e^{ -i \frac{\delta t}{2\hbar} U\left(t + \frac{\delta t}{2}, x\right) } 
			\mathcal{F}_{p \to x}^{-1}\left[
				e^{-i \frac{\delta t}{\hbar} K\left(t + \frac{\delta t}{2}, p\right)}
			\mathcal{F}_{x' \to p} \left[ 
				e^{-i \frac{\delta t}{2\hbar} U\left(t + \frac{\delta t}{2}, x'\right)} \langle x' \ket{\psi(t)}
			\right]\right] + O\left( \delta t^3 \right),
}
where the integrals were rewritten according to Eq.~\eqref{EqContiniousFourierTransform}.

Equation (\ref{SpectralSplitOpMEq}) is the split-operator method valid for a very general class of time-dependent Hamiltonian. The computational complexity of this operator is $O(N\log N)$ and the storage requirement of $O(N)$, where $N$ denotes the length of array storing the wave function. The 1d implementation is provided in
\footnote{\url{https://github.com/dibondar/QuantumClassicalDynamics/blob/master/demo_split_op_schrodinger1D.ipynb}},
whereas, 2d is given in \footnote{\url{https://github.com/dibondar/QuantumClassicalDynamics/blob/master/demo_split_op_schrodinger2D.py}}.

In particular, utilizing the method described in Sec.~\ref{Sec_FTviaFFT} [Eq.~(\ref{EqCFTviaFFT})], we obtain
\begin{align}
	\langle x_{l'} \ket{\psi(t + \delta t)} \approx e^{ -i \frac{\delta t}{2\hbar} U\left(t + \frac{\delta t}{2}, x_{l'}\right) } 
			(-1)^{l'} iFFT_{k \to l'} \left\{
				(-1)^k e^{-i \frac{\delta t}{\hbar} K\left(t + \frac{\delta t}{2}, p_k\right)}
			(-1)^k \right.\qquad\qquad
            \notag\\
            \left.
            \times FFT_{l \to k} \left[ 
				(-1)^l e^{-i \frac{\delta t}{2\hbar} U\left(t + \frac{\delta t}{2}, x_l\right)} \langle x_l \ket{\psi(t)}
			\right]\right\} \notag\\
		 = (-1)^{l'} e^{ -i \frac{\delta t}{2\hbar} U\left(t + \frac{\delta t}{2}, x_{l'}\right) } 
			 iFFT_{k \to l'} \left\{
				e^{-i \frac{\delta t}{\hbar} K\left(t + \frac{\delta t}{2}, p_k\right)}
				FFT_{l \to k} \left[ 
				(-1)^l e^{-i \frac{\delta t}{2\hbar} U\left(t + \frac{\delta t}{2}, x_l\right)} \langle x_l \ket{\psi(t)}
			\right]\right\}.
\end{align}

Note that the split-operator approach discussed in this section leads to the Feynman path integral formulations for both quantum and classical mechanics \cite{Schulman2005, gozzi2015path} after appropriate transformations in Eq. \eqref{SpectralSplitOpMEq}. 

\section{Hamiltonian diagonalization via the time propagation}

In this section, we assume that the Hamiltonian $\hat{H} = K(\hat{p}) + U(\hat{x})$ is time independent. Let $| n \rangle$ denote the eigenstates, $\hat{H} | n \rangle = E_n | n \rangle$, $n=0,1,\ldots$. Then, if we propagate some initial state $ \sum_n c_n | n \rangle$, we get
\begin{align}
	| \psi(t) \rangle = e^{-it\hat{H}/\hbar} \sum_{n=0} c_n | n \rangle = \sum_{n=0} e^{-itE_n /\hbar} c_n | n \rangle. 
    \label{TimeIndependentPropExpansion}
\end{align}
\BoxedEquation{\label{EqEighViaTimeProp}
	\mathcal{F}_{t \to E}^{-1} \Big[ \langle x | \psi(t) \rangle \Big] =  \sum_{n=0} c_n \sqrt{2\pi\hbar} \delta(E-E_n) \langle x | n \rangle,
}
where we made use of the Fourier representation of the delta function \eqref{DeltFunctProperty2}.

For the sake of completeness, let us provide another derivation of Eq.~\eqref{EqEighViaTimeProp}. Starting from Eq. \eqref{TimeIndependentPropExpansion}, one obtains
\begin{align}\label{EqIntuitiveDerEqEighViaTimeProp}
	\int_{-T}^{+T}  \langle x | \psi(t) \rangle e^{iEt/\hbar} dt &= \sum_n c_n \langle x | n \rangle  \int_{-T}^{+T} e^{i(E-E_n)t/\hbar} dt
		= \sum_n c_n \langle x | n \rangle T \frac{e^{i(E-E_n)T/\hbar} - e^{-i(E-E_n)T/\hbar}}{i(E-E_n)T/\hbar} \notag\\
	&= 2 \sum_n c_n \langle x | n \rangle  T \frac{\sin[(E-E_n)T/\hbar]}{(E-E_n)T/\hbar} 
	\underset{T \to +\infty}{\longrightarrow} 2 \sum_n c_n \langle x | n \rangle \times \left\{ \infty \mbox{ if } E = E_n, \atop 0 \mbox{ otherwise.} \right.
\end{align}
Thus, the inverse Fourier transform of $\langle x | \psi(t) \rangle$ with respect to $t$ leads to the sum of spiky functions peaked at the eigenvalues. 

The calculation of the Fourier transform of the autocorrelation function, defined as $\langle \psi(0) | \psi(t) \rangle $, with the help of Eq. \eqref{DeltFunctProperty2} yields
\begin{align}
	\mathcal{F}_{t \to E}^{-1} \Big[ \langle \psi(0) | \psi(t) \rangle \Big] 
	= \sum_{n,k=0} c_k^* c_n \langle k | n \rangle \mathcal{F}_{t \to E}^{-1} \left[ e^{-itE_n/\hbar} \right] 
	= \sum_{n=0} \sqrt{2\pi\hbar} |c_n|^2 \delta(E-E_n). \notag
\end{align}
Therefore, we write
\BoxedEquation{\label{EqFTAutocorrelation}
	\mathcal{F}_{t \to E}^{-1} \Big[ \langle \psi(0) | \psi(t) \rangle \Big] 
	= \sum_{n=0} \sqrt{2\pi\hbar} |c_n|^2 \delta(E-E_n).
}
The derivation \eqref{EqIntuitiveDerEqEighViaTimeProp} can be utilized to arrive at Eq. \eqref{EqFTAutocorrelation} as well. 

Equations \eqref{EqEighViaTimeProp} and \eqref{EqFTAutocorrelation} allow for the diagonalization of the Hamiltonian with a computational complexity of $O(N^2 \log N)$ and a storage requirement of $O(N^2)$, where $N$ denotes the length of the array storing a wave function. See\footnote{\url{https://github.com/dibondar/QuantumClassicalDynamics/blob/master/eigh_via_time_propagation.ipynb}} for the implementation. Recall that linear algebra routines require $O(N^3)$ operations to diagonalize an $N \times N$ matrix.

\section{Imaginary time propagation to obtain eigenstates}

In practical applications, we need to know only the first few states (i.e., the ground and first excited states). These low lying states can be found with $O(N \log N)$ operations [and $O(N)$ storage] in the following fashion:

Let us perform the substitution (known as the Wick rotation) $t \to -i\tau$ in Eq. (\ref{TimeIndependentPropExpansion}),
\begin{align}\label{GroundStateImaginaryTime}
	| \psi(\tau) \rangle &= e^{-\tau \hat{H}/\hbar} \sum_{n=0} c_n | n \rangle = \sum_{n=0} e^{-\tau E_n /\hbar} c_n | n \rangle 
	\notag\\
    &= e^{-\tau E_0 /\hbar} \left[ c_0 | 0 \rangle + \sum_{n=1} e^{-\tau \overbrace{(E_n  - E_0)}^{\mbox{positive}} /\hbar} c_n | n \rangle \right]
	\underset{\tau\to+\infty}{\longrightarrow} \text{const}\cdot  | 0 \rangle,
\end{align} 
since $E_0 < E_1 < E_2 < \cdots$ The wave function $| \psi(\tau) \rangle$ obeys the Schr\"{o}dinger equation written in imaginary time: $-\hbar \frac{d}{d\tau} | \psi(\tau) \rangle = \hat{H}| \psi(\tau) \rangle$. According to Eq. (\ref{GroundStateImaginaryTime}), almost for \emph{any}  initial guess, we will end up with the ground state (up to a constant) if we propagate the Schr\"{o}dinger equation long enough in imaginary time. This method is also useful to evaluate excited states as shown in the following equations:
\begin{align}
	& \lim_{\tau\to+\infty} \Big( 1 - | 0 \rangle\langle 0 | \Big) | \psi(\tau) \rangle 
	= \lim_{\tau\to+\infty} \sum_{n=0} e^{-\tau E_n /\hbar} c_n \Big( | n \rangle  - | 0 \rangle\langle 0 | n \rangle \Big)
	\notag
    \\
    &= \lim_{\tau\to+\infty} \sum_{n=1} e^{-\tau E_n /\hbar} c_n | n \rangle
	= \text{const}\cdot  | 1 \rangle,
	\\
	& \lim_{\tau\to+\infty} \Big( 1 - | 0 \rangle\langle 0 | - | 1 \rangle\langle 1 | \Big) | \psi(\tau) \rangle = \text{const}\cdot  | 2 \rangle.
\end{align} 
The operator $1 - | 0 \rangle\langle 0 |$ projects out the ground state; whereas, $1 - | 0 \rangle\langle 0 | - | 1 \rangle\langle 1 |$ gets rid of the ground and first exited states simultaneously. Thus,
\BoxedEquation{
	\lim_{\tau\to+\infty} \ket{\psi(t)} \Big|_{t \to -i\tau} \propto \ket{0},
	\qquad
	\lim_{\tau\to+\infty} \left. \left( \ket{\psi(t)} - \sum_{k=0}^{n-1} | k \rangle\langle k\ket{\psi(t)} \right) \right|_{t \to -i\tau} \propto \ket{n}.
}

The imaginary time method is implemented in \footnote{\url{https://github.com/dibondar/QuantumClassicalDynamics/blob/master/demo_imag_time_propagation.ipynb}} to find  ground and first excited states.

%%%%%%%%%%%%%%%%%%%%%%%%
% SPECTRAL GAP SECTION %
%%%%%%%%%%%%%%%%%%%%%%%%
\section{Imaginary time propagation to obtain the spectral gap}

Here, we formulate the following theorem.
\begin{theorem}\label{th:sp.gap}
    Let $H$ be a self-adjoint Hamiltonian such that its spectral decomposition reads 
   \begin{align}
       H = \sum_{n=0}^N E_n\Pi_n, \label{eq:h_decomp}
   \end{align} where $\Pi_n$ are orthogonal projectors, i.e., $H \Pi_n = E_n \Pi_n$, $\Pi_n\Pi_m = \delta_{mn}\Pi_n$, and $E_0 < E_1 < E_2 < \cdots < E_N$ are eigenenergies.  Also let \begin{align}\label{eq:imag_time_prop}
       \ket{\psi(\tau)} = \mathcal{N} e^{-\tau H}\ket{\psi(0)} = \mathcal{N}\sum_{n=0}^N e^{-\tau E_n}\Pi_n\ket{\psi(0)},
   \end{align}
   denote a state $\ket{\psi(0)}$ propagated in imaginary time $\tau$ by preserving the norm. Here, $\mathcal{N}$ is the normalization constant ensuring that $\braket{\psi(\tau)}{\psi(\tau)} = 1$ for all $\tau$.
   
   For a self-adjoint observable $O$, if $\bra{\psi(0)}\Pi_0 O\Pi_1\ket{\psi(0)}$ is not a real number, i.e.,
   \begin{align}\label{eq:condition_for_first_spectral_gap}
        \bra{\psi(0)}\Pi_0 O\Pi_1\ket{\psi(0)} \neq \bra{\psi(0)}\Pi_1 O\Pi_0\ket{\psi(0)},    
   \end{align}
   then as $\tau \to \infty$,
   \begin{align}\label{eq:first_spectral_gap}
       \ln\big|\bra{\psi(\tau)} [H,O] \ket{\psi(\tau)} \big|
        &= -\tau \Delta + \mathcal{O}(1), \qquad
        \Delta = E_1-E_0. %\notag 
   \end{align}
Furthermore, if a physical observable $O$ is such that $ \bra{\psi(0)}\Pi_0 O\Pi_1\ket{\psi(0)}$ is real, but  $\bra{\psi(0)}\Pi_0 O\Pi_2\ket{\psi(0)}$ is non real, i.e.,
\begin{align}\label{eq:condition_for_second_spectral_gap}
        \bra{\psi(0)}\Pi_0 O\Pi_1\ket{\psi(0)} = \bra{\psi(0)}\Pi_1 O\Pi_0\ket{\psi(0)} \mbox{ and}\notag\\
        \bra{\psi(0)}\Pi_0 O\Pi_2\ket{\psi(0)} \neq \bra{\psi(0)}\Pi_2 O\Pi_0\ket{\psi(0)},
\end{align}
then as $\tau \to \infty$,
\begin{align}\label{eq:second_spectral_gap}
       \ln\big|\bra{\psi(\tau)} [H,O] \ket{\psi(\tau)} \big|
        = -\tau(E_2-E_0) + \mathcal{O}(1).
\end{align}
\end{theorem}
\begin{proof}
We begin 
\begin{align}\label{eq:starting_point}
    \bra{\psi(0)} & e^{-\tau H}[H,O] e^{-\tau H}\ket{\psi(0)}  \notag\\
    =& \sum_{l,k = 0}^N e^{-\tau (E_l + E_k) }\bra{\psi(0)} \Pi_l (HO - OH) \Pi_k \ket{\psi(0)} \notag\\
    =& \sum_{l,k = 0}^N e^{-\tau (E_l + E_k) } (E_l - E_k) \bra{\psi(0)} \Pi_l O \Pi_k \ket{\psi(0)}.
\end{align}
Under the condition~\eqref{eq:condition_for_first_spectral_gap}, we get
\begin{align}\label{eq:Phi0adAsympt}
     \bra{\psi(0)} & e^{-\tau H}[H,O] e^{-\tau H}\ket{\psi(0)}  \notag\\
    =& e^{-\tau (E_0 + E_1) } \Big[(E_0 - E_1) \bra{\psi(0)} \Pi_0 O \Pi_1 \ket{\psi(0)} 
    %\notag\\
    %& 
    + (E_1 - E_0) \bra{\psi(0)} \Pi_1 O \Pi_0 \ket{\psi(0)} \Big] 
    \notag\\
    &\qquad 
    + \mathcal{O}\left( e^{-\tau (E_0 + E_2) } \right) \notag\\
    =& (E_1 - E_0) e^{-\tau (E_0 + E_1) } \bra{\psi(0)} \Pi_1 O \Pi_0 
    %\notag\\
    %&
    - \Pi_0 O \Pi_1 \ket{\psi(0)}\left[ 1 + \mathcal{O}\left( e^{-\tau (E_2 - E_1) } \right)\right].
\end{align}

Condition~\eqref{eq:condition_for_first_spectral_gap} also implies that the state $\ket{\psi(0)}$ has a population in the ground state, i.e.,
\begin{align}
    & \bra{\psi(0)}\Pi_1 O\Pi_0\ket{\psi(0)} \neq 0 \Longrightarrow
    \Pi_0\ket{\psi(0)} \neq \ket{\varnothing} \Longrightarrow %\notag\\
    %& 
    \bra{\psi(0)} \Pi_0 \ket{\psi(0)} \neq 0,
\end{align}
where $\ket{\varnothing}$ denotes the null vector. With this in mind, we can estimate the normalization constant in Eq.~\eqref{eq:imag_time_prop} as
\begin{align}\label{eq:NAsympt}
    \mathcal{N}^2 &= \left[\bra{\psi(0)} e^{-2\tau H} \ket{\psi(0)} \right]^{-1} \notag\\
        &= \left[e^{-2\tau E_0} \bra{\psi(0)} \Pi_0 \ket{\psi(0)} + \mathcal{O}\left( e^{-2\tau E_1}\right) \right]^{-1} \notag\\
        &= e^{2\tau E_0} \mathcal{O}(1), \qquad \tau \to \infty.
\end{align}
Using Eqs.~\eqref{eq:Phi0adAsympt} and \eqref{eq:NAsympt}, we obtain
\begin{align}
    \ln & \big|\bra{\psi(\tau)} [H,O] \ket{\psi(\tau)} \big| \notag\\
    &= \ln \mathcal{N}^2 + \ln \left| \bra{\psi(0)}  e^{-\tau H}[H,O] e^{-\tau H}\ket{\psi(0)} \right| \notag\\
    & = 2 \tau E_0 + \mathcal{O}(1) -\tau (E_0 + E_1) \notag\\
    &= -\tau (E_1 - E_0) + \mathcal{O}(1),
\end{align}
thereby completing the proof of Eq.~\eqref{eq:first_spectral_gap}.

Under the condition~\eqref{eq:condition_for_first_spectral_gap}, the leading term in expansion~\eqref{eq:starting_point} reads
\begin{align}
    \bra{\psi(0)} & e^{-\tau H}[H,O] e^{-\tau H}\ket{\psi(0)} = \mathcal{O}(1) e^{-\tau (E_0 + E_2) },
\end{align}
whence Eq.~\eqref{eq:second_spectral_gap} follows.
\end{proof}

We note that Theorem~\ref{th:sp.gap} is a modification of the theorem in Ref.~\cite{leamer2023spectral} and that previously, other works have targeted the development of similar, albeit approximate, spectral gap expressions (see, e.g., Refs.~\cite{koh_effects_2016, hlatshwayo2023quantum}). It has also been long recognized that the energies of the low lying excited states can be extracted via the multi-exponential fitting of a correlation function (see, e.g., Refs.~\cite[Sec.VII.C]{ceperley_calculation_1988} and \cite{caffarel_development_1988}). However, such a fit is numerically unstable. Importantly, Eq.~\eqref{eq:first_spectral_gap} does not suffer from this drawback.

In the thermodynamic limit, quantum systems (e.g., the 2d and 3d transverse-field Ising models studied in Ref.~\cite{Lukin2024}) often have a continuum spectrum in addition to discrete energies. Equation~\eqref{eq:first_spectral_gap} is also applicable in such systems; $E_1$ is either the first excited state, if the system has at least two discrete energy levels, or the infimum of the continuous spectrum, if there is only one bound state.
%%%%%%%%%%%%%%%%%%%%%%%%
%%%%%%%%%%%%%%%%%%%%%%%%

\section{Absorbing boundary for time propagation in bounded potentials}\label{Sec_AbsorbibgBoundary}

If you need to propagate a system with a bounded potential (which does not grow to infinity when $x \to\pm\infty$), e.g., a soft-core potential modeling an atom, it is important to ``disable'' the periodicity built into FFT. This is done using a numerical trick known as \emph{the absorbing boundary}, which constitutes in  altering  the potential part of the split-operator propagator (\ref{SpectralSplitOpMEq}) as 
\BoxedEquation{\label{HowAbsBoundWorks}
	e^{-i \delta t U(t, x)/ (2\hbar) } \to e^{-i \delta t U(t, x)/ (2\hbar) } w(x), \qquad 0 \leq w(x) \leq 1,
}
where $w(x)$ similar to a window function. To understand its physical meaning, consider an identity
\begin{align}
	e^{-i \delta t U(t, x)/ (2\hbar) } w(x) = e^{-i \delta t U(t, x)/ (2\hbar) + \log w(x) }; \notag
\end{align}
thus, the replacement (\ref{HowAbsBoundWorks}) is equivalent to
\begin{align}\label{EqAbsBoundaryCoreection}
	U(x, t) \to U(x, t) -i B(x), \qquad B(x) = -\left( 2 \hbar / \delta t \right) \log w(x)  \geq 0.
\end{align}
adding an imaginary potential, which imitates absorption.

If $B(x)$ is vanishingly small everywhere except the boundary, then according to Eq. (\ref{HowAbsBoundWorks}), the wave packet approaching the boundary gets exponentially suppressed/absorbed, hence the method's name. Note that it is crucial to choose $w(x)$ so as not to affect the dynamics occurring in the middle of the coordinate grid (in particular, the Ehrenfest theorems are numerically obeyed, \emph{etc.}); at the same time, the function $w(x)$ should vary smoothly such that it performs as an absorber rather than a reflector. Regarding an example of numerical implementation, see\footnote{\url{https://github.com/dibondar/QuantumClassicalDynamics/blob/master/absorbing_boundary.ipynb}}.

\section{Evolution in systems with an internal degree of freedom}\label{Sec:PauliLikeEqu}

We shall now turn our attention to the quantum system an additional internal degree of freedom such that the Hamiltonian can be written as
\begin{align}\label{Eq:PauliLikeHamiltonian}
	\hat{H}(t) = \sum_{j=0}^3 [ K_j(t, \hat{p}) + U_j(t, \hat{x})] \sigma_j,
\end{align}
where $\sigma_j$ are the celebrated Pauli matrices
\begin{align}
	\sigma_0 = \left( 
		\begin{array}{cc}
			1 & 0 \\
			0 & 1 \\
		\end{array}
	\right),
	\quad
	\sigma_1 = \left( 
		\begin{array}{cc}
			0 & 1 \\
			1 & 0 \\
		\end{array}
	\right),	
	\quad
	\sigma_2 = \left( 
		\begin{array}{cc}
			0 & -i \\
			i & 0 \\
		\end{array}
	\right),
	\quad	
	\sigma_3 = \left( 
		\begin{array}{cc}
			1 & 0 \\
			0 & -1 \\
		\end{array}
	\right).
\end{align}

We calculate the unitary propagator $\hat{U}(t,t')$ for the Hamiltonian \eqref{Eq:PauliLikeHamiltonian}
\begin{align}\label{Eq:SecondOrdPauli}
	\hat{U}(t + \delta t, t) &= \exp\left[- \frac{i}{\hbar} \int_{t}^{t + \delta t} \hat{H}(\tau) d\tau \right] \notag\\
	&=\exp\left[- \frac{i}{\hbar} \hat{H}\left( t + \frac{\delta t}{2} \right) \delta t \right] + O\left( \delta t^3 \right) \qquad\qquad \mbox{[using Eqs. (\ref{EqTrapezoidalIntegralRule}) and \eqref{EqTimeOrderingDropping}]} \notag\\	
	&=  \exp\left[ -i \frac{\delta t}{2\hbar} \sum_{j=0}^3 \sigma_j U_j\left(t + \frac{\delta t}{2}, \hat{x}\right) \right] 
			\exp\left[ -i \frac{\delta t}{\hbar} \sum_{j=0}^3 \sigma_j K_j\left(t + \frac{\delta t}{2}, \hat{p}\right) \right]  \notag\\
	& \quad \times \exp\left[ -i \frac{\delta t}{2\hbar} \sum_{j=0}^3 \sigma_j U_j\left(t + \frac{\delta t}{2}, \hat{x}\right) \right]
		+ O\left( \delta t^3 \right).
	\qquad \mbox{[using Eqs. (\ref{SecondOrderSplitting}) and \eqref{Eq:PauliLikeHamiltonian}]}	
\end{align}
Similarly to Eq.~\eqref{SecondOrderTexpPartEq}, if we ignore the $O\left( \delta t^3 \right)$ term in r.h.s. of Eq.~\eqref{Eq:SecondOrdPauli}, the remaining expression is unitary. Using the identity
\begin{align}
	e^{ia\sum_{j=1}^3 n_j \sigma_j} = \cos(a) + i\sin(a) \sum_{j=1}^3 n_j \sigma_j,
	\quad
	\sum_{j=1}^3 n_j^2 = 1,
\end{align}
we introduce an auxiliary operator  $\mathcal{P}$  acting on a two-component vector $\left(\psi_1, \psi_2 \right)^T$ as 
\begin{align}
	\mathcal{P}(a; c_j)  \left({\psi_1 \atop \psi_2}\right) =&
	e^{ ia \sum_{j=0}^3 c_j \sigma_j }  \left(\psi_1 \atop \psi_2 \right)\notag\\
    =&
		e^{ia c_0} \left( 
			\cos(ab)\psi_1 + i \sin(ab)\left(c_3 \psi_1 + [c_1 -ic_2]\psi_2\right) / b  
		\atop 
			\cos(ab) \psi_2 + i \sin(ab)\left([c_1 + ic_2]\psi_1 - c_3 \psi_2 \right) / b
		\right),
\end{align}
where $b = \sqrt{ c_1^2 + c_2^2 + c_3^2 }$.
Using the coordinate representation in Eq. \eqref{Eq:SecondOrdPauli}, we obtain
\begin{align}
	\langle x \ket{\psi(t + \delta t)} &= \bra{x} \hat{U}(t + \delta t, t) \ket{\psi(t)} 
		=  \mathcal{P}_x \int dp \langle x | p \rangle \mathcal{P}_p  \int dx' \langle p | x' \rangle  \mathcal{P}_{x'} \langle x' \ket{\psi(t)} + O\left( \delta t^3 \right) \notag\\
		&= \mathcal{P}_x  \mathcal{F}_{p \to x'}^{-1}\left[ \mathcal{P}_p  \mathcal{F}_{x' \to p} \left[ \mathcal{P}_{x'} \langle x' \ket{\psi(t)} \right] \right] + O\left( \delta t^3 \right), 
\end{align}
where
\begin{align}
		\mathcal{P}_x = \mathcal{P}\left[  -\frac{\delta t}{2\hbar}; U_j\left(t + \frac{\delta t}{2}, x\right) \right], \quad
		\mathcal{P}_p = \mathcal{P}\left[  -\frac{\delta t}{\hbar}; K_j\left(t + \frac{\delta t}{2}, p\right) \right].
        %\quad
		%\mathcal{P}_{x'} = \mathcal{P}\left[  -\frac{\delta t}{2\hbar}; U_j\left(t + \frac{\delta t}{2}, x'\right) \right].
\end{align}
This propagator has been implemented in \footnote{\url{https://github.com/dibondar/QuantumClassicalDynamics/blob/master/split_op_pauli-like1D.py}}.

\section{Translationally-invariant lattice systems}\label{Sec:SolidState}

Assume that the potential $U(x)=U(x+a)$ is a periodic function with $a$ denoting the \emph{lattice constant}. The system is invariant under the translation $x \to x + a$. 
Recall Eq. \eqref{DisplacementOp}: $e^{a\frac{\partial}{\partial x}} f(x) =  f(x + a)$. Hence, the corresponding translation operator acting on the wavefunction~$\psi(x)$ reads::
\begin{align}
	\hat{T}_a \psi(x) = \psi(x + a), \qquad \hat{T}_a = e^{a\frac{\partial}{\partial x}} = e^{ia\frac{-i\hbar}{\hbar} \frac{\partial}{\partial x}} = e^{ia\hat{p}/\hbar}.
\end{align}
The physical eigenfunction~$\psi(x)$ of this lattice system is assumed to be an eigenvector of the translation operator $\hat{T}_a $,
\begin{align}\label{EigenFuncTa}
	\hat{T}_a  \psi(x) = \lambda_T \psi(x).
\end{align}
What are the physically permissible values of $\lambda$? If $|\lambda_T| > 1$, then the repeated application $\hat{T}_a$ onto the wavefunction $\psi$ increases its magnitude. If $|\lambda_T| <1$, then the action of $\hat{T}_a$ onto $\psi$ decreases its magnitude. In both cases, the translational symmetry is broken; thus, the only permissible form for the eigenvalues  $\lambda_T$ of the translation operators is $|\lambda_T| = 1$. The latter can be re-written as $\lambda_T = e^{ika}$, where $\hbar k$ is called \emph{quasimomentum}.

Since $\lambda_T$ is unchanged under the substitution $k \to k + 2\pi/a$, the physically accessible region of quasimomenta is confined to the interval $(-2\pi\hbar/a, 2\pi\hbar/a)$.
  
Consider the case of a one-dimensional lattice system consisting of a finite number ($N$) of elementary cells. The translation invariance is preserved only if the system has a ring geometry obeying the periodic  boundary condition
\begin{align}
	\psi(x + Na) = \psi(x) 
	\Longrightarrow \lambda_T^N = e^{ikaN} = 1.
\end{align}
The latter equation has the following $N$ distinct solutions: $k = 2\pi n / (aN)$, $n = 0,\ldots, N-1$. Therefore, the number of elementary cells $N$ coincides with the number of physically attainable values of quasimomenta.

Let us represent the eigenfunction $\psi(x)$ of the translation operator (\ref{EigenFuncTa}) as
\begin{align}\label{BlochThrm1}
	\psi(x) = e^{ikx} \phi(x).
\end{align}
Then
\begin{align}
	\hat{T}_a \psi(x) = \psi(x+a) = e^{ika} \psi(x)
	\Longrightarrow
	e^{ikx} e^{ika} \phi(x + a) = e^{ika} e^{ikx} \phi(x)
	\Longrightarrow
	\phi(x + a) = \phi(x).
\end{align}
In other words, the function $\phi(x)$ is periodic.  Equation (\ref{BlochThrm1}) is the celebrated \emph{Bloch theorem}. 

Now, we aim to find a common set of eigenvectors for the translation operator $\hat{T}_a$ and the Hamiltonian $\hat{H} = H(\hat{x}, \hat{p})$. As we will see below, such eigenstates need to be labeled by two quantum numbers: the quasimomentum $k$ and the band index $n$. Assuming the resolution of identity
\begin{align}
	1 = \int_0^a dx | x \rangle\langle x |,
\end{align}
we rewrite the Bloch theorem \eqref{BlochThrm1} as
\begin{align}
	\langle x | \psi_{k,n} \rangle = e^{ikx} \langle x | \phi_{k,n} \rangle
	= \langle x | e^{ik\hat{x}} | \phi_{k,n} \rangle
	\Longrightarrow
	| \psi_{k,n} \rangle = e^{ik\hat{x}} | \phi_{k,n} \rangle,
\end{align}
which is an eigenvector of the translation operator.
%The normalization condition reads
%\begin{align}
%	\langle \psi_{k',n'} | \psi_{k,n} \rangle = \langle \phi_{k',n'} | \phi_{k,n} \rangle = \delta(k-k') \delta_{n,n'}.
%\end{align}
Employing Eq. (\ref{Eq_U_Weyl_U_inv}), let us find $E_n(k)$ such that
\begin{align}\label{Eq_main_method_find_band_structure}
	\hat{H} | \psi_{k,n} \rangle = E_n(k) | \psi_{k,n} \rangle,
\end{align}
which gives us
\begin{align}\label{Eq_main_method_find_band_structure2}
	\hat{H} e^{ik\hat{x}} | \phi_{k,n} \rangle = E_n(k) e^{ik\hat{x}} | \phi_{k,n} \rangle 
\end{align}
\begin{align}\label{Eq_main_method_find_band_structure3}
	E_n(k)  | \phi_{k,n} \rangle = e^{-ik\hat{x}} H(\hat{x}, \hat{p}) e^{ik\hat{x}} | \phi_{k,n} \rangle 
	=   H\left(\hat{x}, e^{-ik\hat{x}} \hat{p} e^{ik\hat{x}} \right)  | \phi_{k,n} \rangle.
\end{align}
We can summarize:
\BoxedEquation{
	\left\{
		\hat{T}_a | \psi_{k,n} \rangle = e^{ika} | \psi_{k,n} \rangle,  \atop
		\hat{H} | \psi_{k,n} \rangle = E_n(k) | \psi_{k,n} \rangle,
	\right.
	\Longleftrightarrow 
	\left\{
		\hat{T}_a  | \phi_{k,n} \rangle = | \phi_{k,n} \rangle, \atop
		H\left(\hat{x}, \hat{p} + \hbar k \right)  | \phi_{k,n} \rangle = E_n(k)  | \phi_{k,n} \rangle.
	\right.
}
Note that in the last step \eqref{Eq_main_method_find_band_structure3}, we have used the Hausdorff expansion
\begin{equation}\label{HausdorffExpansionEq}
    e^{\hat{A}} \hat{B} e^{-\hat{A}} = \hat{B} + [ \hat{A}, \hat{B}] + \frac{1}{2!} [ \hat{A}, [ \hat{A}, \hat{B}]] + \cdots 
		= \sum_{n=0}^{\infty} \frac{1}{n!} \underbrace{[ \hat{A}, [ \hat{A}, [\cdots [ \hat{A}}_{\mbox{$n$-times}}, \hat{B}] \cdots ].
\end{equation}

In conclusion, for a fixed value of the quasimomentum $k$, we determine $| \phi_{k,n} \rangle$ and $E_n(k)$ as eigenfunctions and eigenvalues (labeled by $n$) of the operator $H\left(\hat{x}, \hat{p} + \hbar k \right)$ with the periodic boundary condition.
Equations (\ref{Eq_main_method_find_band_structure})--(\ref{Eq_main_method_find_band_structure3}) are employed in numerical simulations to calculate the band structure\footnote{\url{https://github.com/dibondar/QuantumClassicalDynamics/blob/master/solid_state_band_structure.py}}.

\begin{comment}
The laser field couples to the momentum operator in the velocity gauge, thus we consider the matrix elements of $\hat{p}$ in the basis of the Bloch functions: Since $[\hat{p}, \hat{T}_a]=0$, both the operators share common eigenfunctions. Thus, the momentum operator is diagonal with respect to the quasimomentum $k$ (because $k$ labels the eigenfunctions of $\hat{T}_a$). In other words, the momentum operator does not couple Bloch states with different quasimomenta, and induces intraband transitions. From the computational point of view, a coherent wave packet needs to be broken into the constituent Bloch functions, which are propagated independently. Therefore, we obtain the propagator for a fixed $k$ quasimomentum 
\begin{align}
	\underbrace{\hat{\mathcal{T}} \exp\left[ \frac{-i}{\hbar}\int_{t}^{t + \delta t}  H\left(\hat{x}, \hat{p} + e A(t) /c \right) d\tau \right]}_{\mbox{Bloch boundary condition:} \psi(x + a) = e^{ika}\psi(x) \atop \mbox{(non-compatible with FFT)}}
	&= \hat{\mathcal{T}} \exp\left[ \frac{-i}{\hbar}\int_{t}^{t + \delta t}  e^{ik\hat{x}} H\left(\hat{x}, \hat{p} + \hbar k + e A(t) /c \right) e^{-ik\hat{x}} d\tau \right] \notag\\
	& = e^{ik\hat{x}} \underbrace{
		\hat{\mathcal{T}} \exp\left[ \frac{-i}{\hbar}\int_{t}^{t + \delta t}   H\left(\hat{x}, \hat{p} + \hbar k + e A(t) /c \right)  d\tau \right]
	}_{\mbox{periodic boundary condition:} \psi(x + a) = \psi(x) \atop \mbox{(Compatible with FFT !!!)}}  e^{-ik\hat{x}}.
\end{align}
\end{comment}

\chapter{6. Classical dynamics}\label{Chapter:6}

\section{Inference of classical Liouville equation from Ehrenfest theorems}\label{ODMLiouvilleSec} 

We are back to classical mechanics.

Let $\hat{x}$ and $\hat{p}$ be self-adjoint operators representing the \emph{classical} coordinate and momentum observables. The commutation relation 
\begin{align}\label{Ch2_ClassicalCommutator2}
	\commut{\hat{x}, \hat{p}}_{\textrm{class.}} = 0,
\end{align}
encapsulates two basic experimental facts of classical kinematics: i) the position and momentum can be measured simultaneously with arbitrary accuracy, ii) the observed values do not depend on the order of performing the measurements. In terms of our axioms, the dynamical observations of the classical particle's position and momentum are summarized in 
\begin{align}\label{EhrenfestThs}
	m\frac{d}{dt} \bra{\Psi(t)} \hat{x} \ket{\Psi(t)} & = \bra{\Psi(t)} \hat{p} \ket{\Psi(t)}, \notag\\
	\frac{d}{dt} \bra{\Psi(t)} \hat{p} \ket{\Psi(t)} & = \bra{\Psi(t)} -U'(\hat{x}) \ket{\Psi(t)}.
\end{align}

We now derive the equation of motion for a classical state. The application of the chain rule to Eq. (\ref{EhrenfestThs}) gives
\begin{align}\label{Ch2_ExpandedClassicsErhenfest_Th}
	\bra{d\Psi/dt} \hat{x} \ket{\Psi} + \bra{\Psi} \hat{x} \ket{d\Psi/dt} &= \bra{\Psi} \hat{p}/m \ket{\Psi}, \notag\\
	\bra{d\Psi/dt} \hat{p} \ket{\Psi}  + \bra{\Psi} \hat{p} \ket{d\Psi/dt} & = \bra{\Psi} -U'(\hat{x}) \ket{\Psi},
\end{align}
into which we substitute a consequence of Stone's theorem 
\begin{align}\label{KvNEquation}
	i \ket{d\Psi(t)/dt} = \hat{L} \ket{\Psi(t)},
\end{align}
and obtain
\begin{align}\label{Avaraged_Equations_for_Liouvillian}
	im \bra{\Psi(t)} \commut{\hat{L}, \hat{x} } \ket{\Psi(t)} &= \bra{\Psi(t)} \hat{p} \ket{\Psi(t)}, \notag\\
	i \bra{\Psi(t)} \commut{\hat{L}, \hat{p}} \ket{\Psi(t)} &= - \bra{\Psi(t)} U'(\hat{x}) \ket{\Psi(t)}.
\end{align}
Since Eq. (\ref{Avaraged_Equations_for_Liouvillian}) must be valid for all possible initial states, the averaging can be dropped, and we have the  system of commutator equations for the motion generator $\hat{L}$,
\begin{align}\label{Ch2_Equations_for_Liouvillian}
	im  \commut{\hat{L}, \hat{x}}  =  \hat{p} , \qquad i \commut{\hat{L}, \hat{p}} = -U'(\hat{x}).
\end{align}
Since $\hat{p}$ and $\hat{x}$ commute, the solution $\hat{L}$ cannot be found by simply assuming $\hat{L} = L(\hat{x}, \hat{p})$; otherwise, Eq. (\ref{Ch2_Equations_for_Liouvillian}) immediately leads to the contradiction: $0 = \hat{p}$ and $0 =  -U'(\hat{x})$.

Therefore, we add into consideration two new operators $\hat{\lambda}$ and $\hat{\theta}$ such that
\begin{align}\label{Complete_classical_algebra}
	\commut{ \hat{x}, \hat{\lambda} } = \commut{ \hat{p}, \hat{\theta} } = i, 
\end{align}
and the other commutators among $\hat{x}$, $\hat{p}$, $\hat{\lambda}$, and $\hat{\theta}$ vanish. (A curious fact is that the operators $\hat{\lambda}$ and $\hat{\theta}$ were first discovered in the phase space representation of quantum mechanics and known as the Bopp operators \cite{Bopp1956}.) The need to introduce auxiliary operators arises in classical dynamics because all the observables commute;  hence, the notion of an individual trajectory can be introduced. Moreover, the choice of the commutation relationships (\ref{Complete_classical_algebra}) is unique. Equation (\ref{Complete_classical_algebra}) can be considered as an additional axiom. Now we seek the generator $\hat{L}$ in the form $\hat{L} = L(\hat{x}, \hat{\lambda}, \hat{p}, \hat{\theta})$. Utilizing Eq. (\ref{Weyl_derivative_of_function}), we convert the commutator equations (\ref{Ch2_Equations_for_Liouvillian}) into the differential equations
\begin{align}
	m L'_{\lambda} (x, \lambda, p, \theta) = p, \qquad L'_{\theta} (x, \lambda, p, \theta) = -U'(x),
\end{align}
from which, the generator of classical dynamics $\hat{L}$ is found to be
\begin{align}\label{KvNGenerator}
	\hat{L} = \hat{p} \hat{\lambda} / m - U'(\hat{x}) \hat{\theta} + f(\hat{x}, \hat{p}),
\end{align}
where $f(x,p)$ is an arbitrary real-valued function.
Equations (\ref{KvNEquation}), (\ref{Complete_classical_algebra}), and (\ref{KvNGenerator}) represent classical dynamics in an abstract form. 

Let us find the equation of motion for $|\langle p \, x \ket{\Psi(t)}|^2$ by rewriting Eq. (\ref{KvNEquation}) in the $xp$-representation, in which
\begin{align}\label{ClassicalXPRepresentation}
	\langle x \, p | \hat{x} |\Psi \rangle = x \langle x \, p |\Psi \rangle,
	\quad
	\langle x \, p | \hat{\lambda}  |\Psi \rangle = -i\frac{\partial}{\partial x} \langle x \, p |\Psi \rangle,
    \notag\\
	\quad
	\langle x \, p | \hat{p} |\Psi \rangle = p \langle x \, p |\Psi \rangle,
	\quad
	\langle x \, p | \hat{\theta}  |\Psi \rangle = -i\frac{\partial}{\partial p} \langle x \, p |\Psi \rangle.
\end{align}
[see Eq. (\ref{EqSummaryCommutatorSection}) for justification],
\begin{align}\label{PreLiouvillianEq}
		\left[ i\frac{\partial }{\partial t} +i\frac{p}{m} \frac{\partial}{\partial x} - i U'(x) \frac{\partial}{\partial p} - f(x,p) \right]  \langle x \, p \ket{\Psi(t)} = 0.
\end{align}
Let's explicitly work it out the following derivative
\begin{align}\label{FromKvNToLiouvilleDerivation}
	\frac{\partial }{\partial t} |\langle x \, p \ket{\Psi(t)}|^2 
	=& \langle x \, p | \Psi(t) \rangle^* \frac{\partial }{\partial t} \langle x \, p | \Psi(t) \rangle  + \langle x \, p | \Psi(t) \rangle \left( \frac{\partial }{\partial t}  \langle x \, p | \Psi(t) \rangle \right)^*
	\notag\\
	=&  \langle x \, p | \Psi(t) \rangle^* \left[-\frac{p}{m} \frac{\partial}{\partial x} + U'(x) \frac{\partial}{\partial p} -i f(x,p) \right]  \langle x \, p \ket{\Psi(t)}
	\notag\\
	& + \langle x \, p | \Psi(t) \rangle \left[-\frac{p}{m} \frac{\partial}{\partial x} + U'(x) \frac{\partial}{\partial p} +i f(x,p) \right]  \langle x \, p \ket{\Psi(t)}^*
	\notag\\
	=&  - \frac{p}{m} \langle x \, p | \Psi(t) \rangle^* \frac{\partial}{\partial x}  \langle x \, p | \Psi(t) \rangle -  \frac{p}{m} \langle x \, p | \Psi(t) \rangle  \frac{\partial}{\partial x}  \langle x \, p | \Psi(t) \rangle^* 
	\notag\\
	&+ U'(x) \langle x \, p | \Psi(t) \rangle^* \frac{\partial}{\partial p} \langle x \, p | \Psi(t) \rangle + U'(x) \langle x \, p | \Psi(t) \rangle \frac{\partial}{\partial p} \langle x \, p | \Psi(t) \rangle^* 
	\notag\\
	=& - \frac{p}{m} \frac{\partial}{\partial x} \left| \langle x \, p | \Psi(t) \rangle \right|^2 + U'(x) \frac{\partial}{\partial p} \left| \langle x \, p | \Psi(t) \rangle \right|^2. 
\end{align}
Therefore, Eq.~(\ref{PreLiouvillianEq}) yields the well known classical Liouville equation for the probability distribution in phase-space $\rho(x,p;t) = |\langle p \, x \ket{\Psi(t)}|^2$,
\begin{align}\label{Liouville_Eq}
	\frac{\partial }{\partial t} \rho(x,p;t)  = \left[ -\frac{p}{m} \frac{\partial}{\partial x} +  U'(x) \frac{\partial}{\partial p}  \right]  \rho(x,p;t).
\end{align}
Koopman and von Neumann \cite{Koopman1931, Neumann1932} (see Ref. \cite{DaniloMauro2002} for a great introduction) pioneered the recasting of classical mechanics in a form similar to quantum mechanics by introducing classical complex valued wave functions and representing associated physical observables by means of commuting self-adjoint operators (for modern developments and applications see Refs. \cite{Gozzi1988, Gozzi1989, Wilkie1997, Wilkie1997a, Gozzi2002, DaniloMauro2002, Deotto2003, Deotto2003a, Abrikosovjr2005, Blasone2005, Brumer2006, Carta2006, Gozzi2010, Gozzi2011, Cattaruzza2011}).

Thus, we have deduced the classical Liouville equation (\ref{Liouville_Eq}) along with the Koopman-von Neumann theory [given by Eqs. (\ref{KvNEquation}) and (\ref{KvNGenerator})] from Eq. (\ref{EhrenfestThs}) by assuming that the classical momentum and coordinate operators commute. Note that the function $f(x, p)$ in Eqs. (\ref{KvNGenerator}) and (\ref{PreLiouvillianEq}) disappeared when we were calculating the equation of motion for the density [see Eq. (\ref{FromKvNToLiouvilleDerivation})]. This suggests that $f(x, p)$ is in fact a phase of the Koopman--von Neumann wave function, moreover it is physically unnecessary in dynamics [Eq. (\ref{Liouville_Eq})] as well as kinematics
\begin{align}
	\langle O (x,p) \rangle = \langle \Psi | O (\hat{x}, \hat{p}) | \Psi \rangle = \int dxdp \langle \Psi | x\, p \rangle O (x, p) \langle x\, p | \Psi \rangle =  \int dxdp \, \rho(x,p) O(x,p).
\end{align}

\begin{figure}
    \centering
	\includegraphics[width=0.5\hsize]{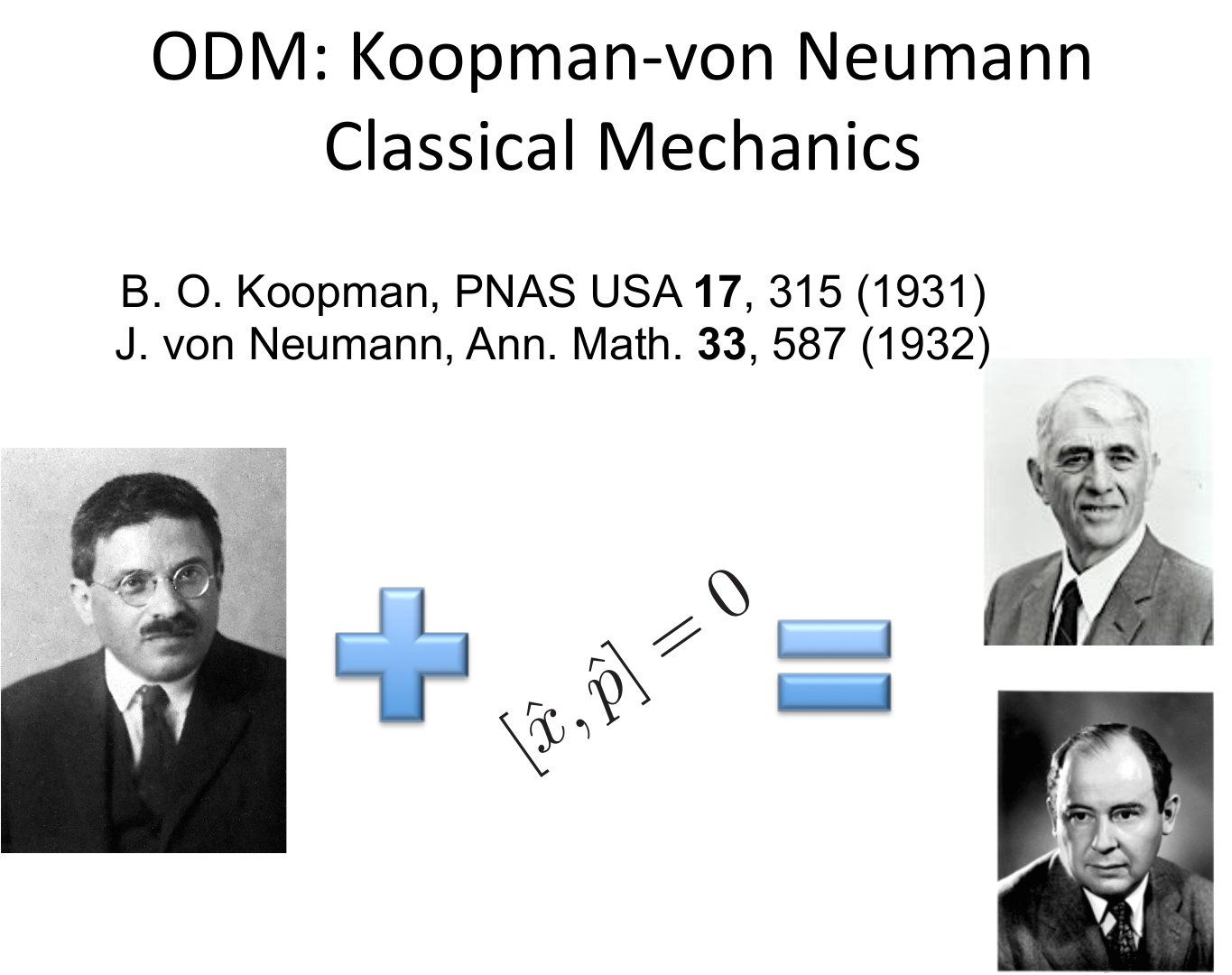}
    \caption{How we derive the Koopman-von Neumann classical mechanics in Sec. \ref{ODMLiouvilleSec}}
	\includegraphics[width=0.6\hsize]{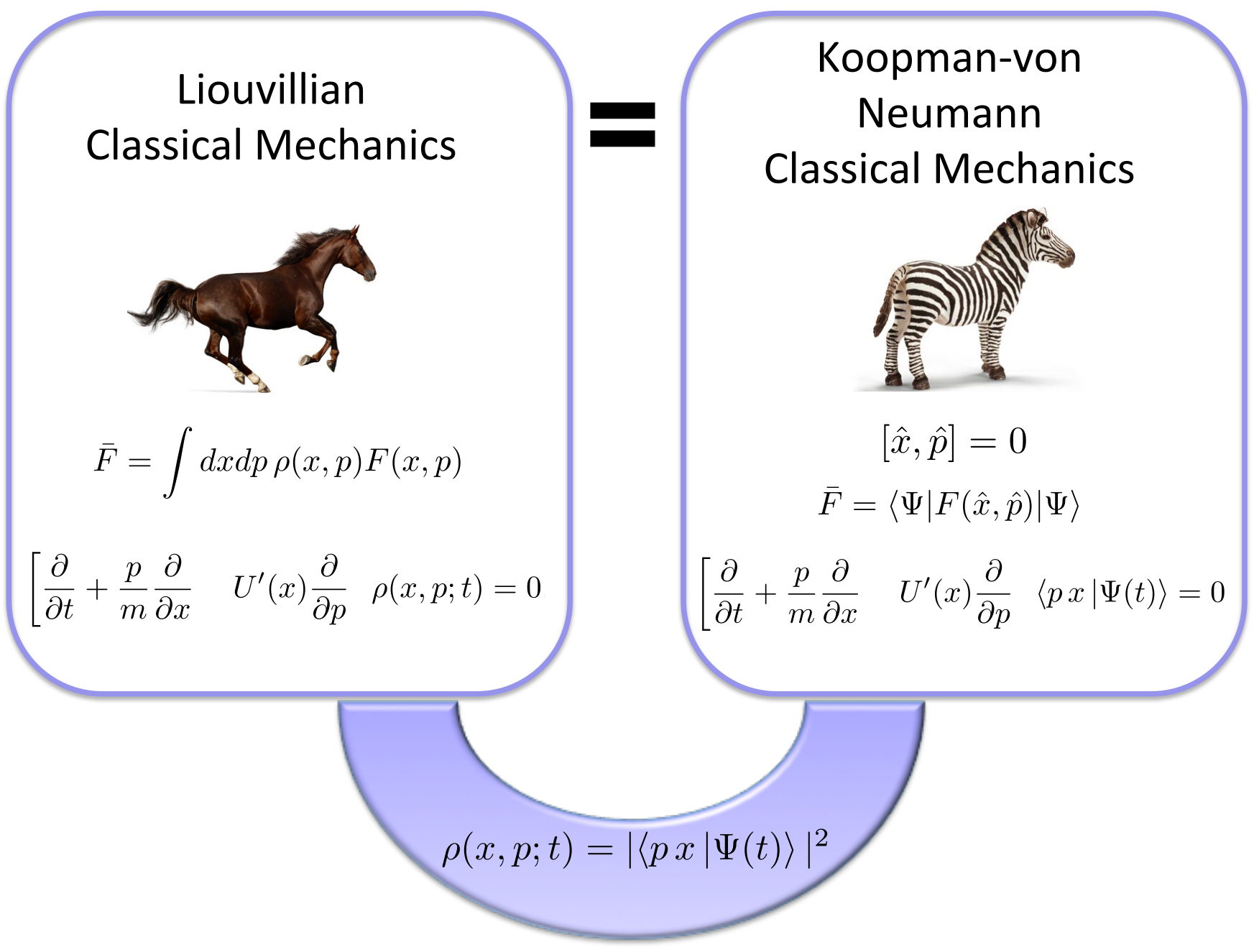}
    \caption{Liouville and  Koopman-von Neumann  mechanics are like a horse and a zebra: at first, they look very different, but after a closer inspection, they are physically the same up to coloring.}
    \includegraphics[width=0.8\hsize]{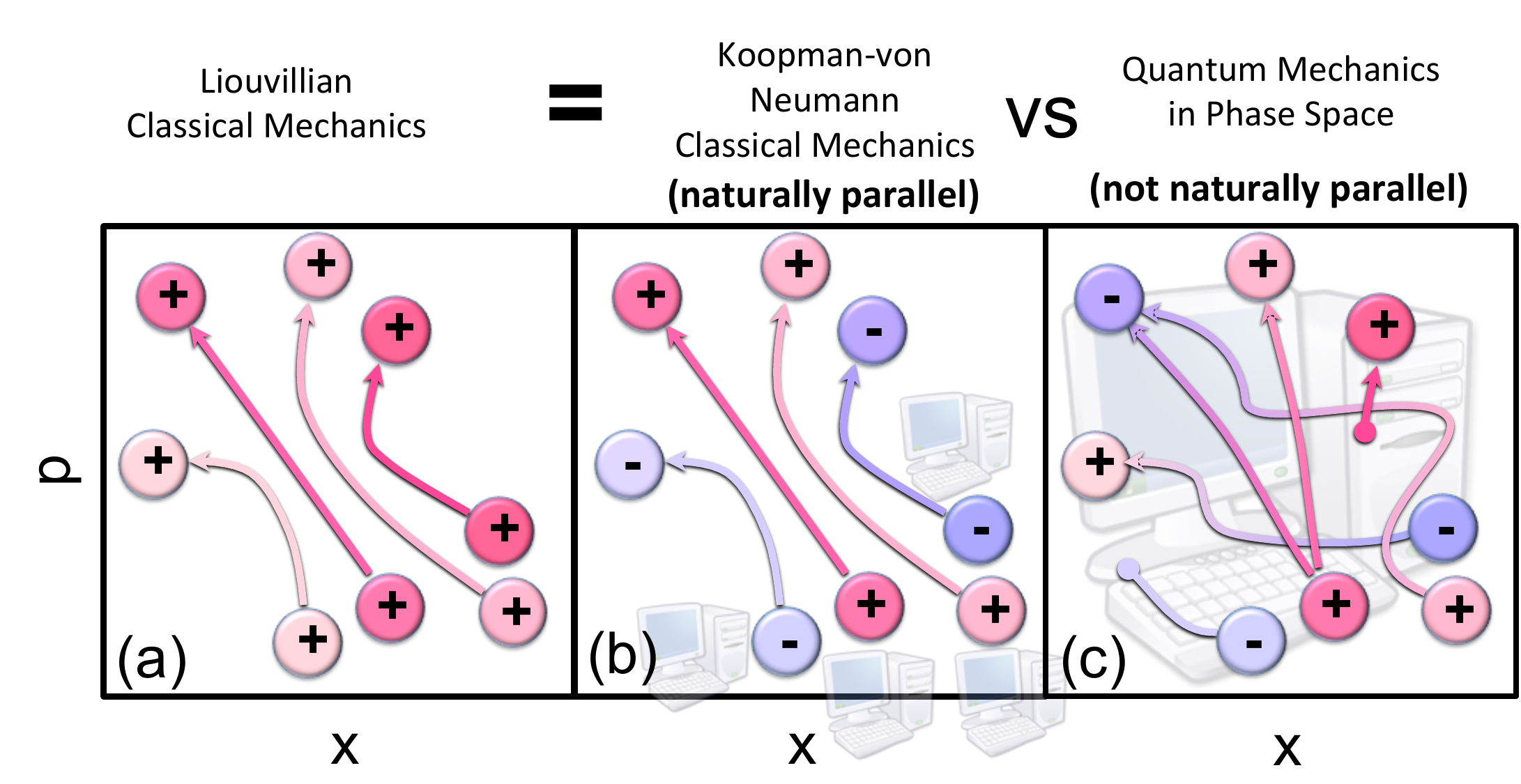}
    \caption{Comparison of (a) classical Liouville mechanics,  (b) Koopman-von Neumann classical mechanics, and (c) the Wigner phase space representation of quantum mechanics.}\label{Fig_L_KvN_QM}
\end{figure}

\section{From Liouville to Newton and back}\label{Sec:Liouville2Newton}

Let us derive Newton's equations from the Liouville equation (\ref{Liouville_Eq}). Assume a classical ensemble has only one particle moving in the phase space along a trajectory $x_c(t)$ and $p_c(t)$. Then, the Liouville probability distribution for such a degenerate ensemble reads
\begin{align}\label{ClassProbDistributionSingleParticle}
	\rho(x,p; t) = \delta[ x - x_c(t)] \delta[p - p_c(t)].
\end{align}
Substituting Eq. (\ref{ClassProbDistributionSingleParticle}) into Eq. (\ref{Liouville_Eq}), we get 
\begin{align}
	\frac{\partial}{\partial t} \rho(x,p; t) =& -\delta'[ x - x_c(t)] \delta[p - p_c(t)] \frac{dx_c(t)}{dt} -  \delta[ x - x_c(t)] \delta'[p - p_c(t)] \frac{dp_c(t)}{dt}
	\notag\\
	=&  -\frac{p}{m} \delta'[ x - x_c(t)] \delta[p - p_c(t)]  + U'(x) \delta[ x - x_c(t)] \delta'[p - p_c(t)].
\end{align}
Now using Eq. (\ref{DeltFunctPropert3})
\begin{align}
	 \left(\frac{dx_c(t)}{dt}  -\frac{p_c(t)}{m}\right) \delta'[ x - x_c(t)] \delta[p - p_c(t)]  
	 + \left(\frac{dp_c(t)}{dt} + U'[x_c(t)]\right)  \delta[ x - x_c(t)] \delta'[p - p_c(t)] =0.
\end{align}
The fundamental property of the Dirac delta function (\ref{MostGeneralFormofDistribution}) immediately leads us to Newton's equations
\begin{align}\label{NewtonEqs}
	 \frac{dx_c(t)}{dt}  = \frac{p_c(t)}{m}, \qquad  \frac{dp_c(t)}{dt} = -U'[x_c(t)].
\end{align}

One can go back from Newton's equations to Liouville equation (\ref{Liouville_Eq}) by assuming that
\begin{align}
	\rho(x,p; t)  = \sum_c w_c \delta[ x - x_c(t)] \delta[p - p_c(t)],
\end{align}
where $w_c$ are weights, and each trajectory $x_c(t)$, $p_c(t)$ obeys Eq. (\ref{NewtonEqs}).

\section{Symplectic integrators for classical dynamic propagation}\label{Sec:SymplecticPropagator}

This section closely follows Sec. \ref{Sec:SplitOpSchrodinger}, where the split-operator technique for propagating Schr\"{o}\-din\-ger equation was developed. 

First, let us consider slightly more general classical dynamics, characterized by the following Newton's equations
\begin{align}\label{TImeDependentClassSys}
	 \frac{dx(t)}{dt}  = K'[p(t), t], \qquad  \frac{dp(t)}{dt} = -U'[x(t), t].
\end{align}
Note that both right hand sides are time-dependent. We first convert it to the time-independent dynamics by increasing the number of variables (see, e.g., Ref. \cite{howland1980two}).  In particular, we introduce variables $s_x$ (a fictitious coordinate) and $s_p$ (a  fictitious momentum) such that
\begin{align}
	\frac{d}{dt} \left( x(t) \atop s_x(t) \right) = \left( K'[p(t), s_p(t)] \atop 1 \right),  \label{EqClassicalExtendedX}\\
	\frac{d}{dt} \left( p(t) \atop s_p(t) \right) = \left( -U'[x(t), s_x(t)] \atop 1 \right).  \label{EqClassicalExtendedP}
\end{align}
with the initial condition $s_{x,p}(0) = 0$. Formally speaking, the right hand sides of equations (\ref{EqClassicalExtendedX}) and (\ref{EqClassicalExtendedP})  do not explicitly depend on time $t$; there is an implicit time dependence through the coordinates and momenta.  The variable $s_p$ acts as a new momentum variable, whereas $s_x$ can be interpreted as associated coordinate variable. In other words, the system (\ref{TImeDependentClassSys}) with an explicitly time-dependent r.h.s. has been reduced to the system (\ref{EqClassicalExtendedX}) and (\ref{EqClassicalExtendedP}) with a time independent r.h.s.. Note that the equations (\ref{EqClassicalExtendedX}) and (\ref{EqClassicalExtendedP})  have the same structure as Eq.~(\ref{TImeDependentClassSys}): The r.h.s. of the  equations of motion for coordinate [Eq. (\ref{EqClassicalExtendedX})] exclusively depends on the momenta, whereas the equations of motion for momenta [Eq. (\ref{EqClassicalExtendedP})] solely depend on the coordinates. 

Given the fact that any time-dependent system can be transformed to the time-independent one, we shall focus on the development of numerical methods exclusively for the latter case.

The following Newton's equations
\begin{align}\label{NewtonEqForNumMethod}
	 \frac{dx(t)}{dt}  = K'[p(t)], \qquad  \frac{dp(t)}{dt} = -U'[x(t)].
\end{align}
can be represented by the Koopman-von Neumann equation
\begin{align}\label{KvNEqForNumMethod}
	i \frac{d}{dt} | \Psi(t) \rangle =  \hat{L} | \Psi(t) \rangle, \qquad \hat{L} = K'(\hat{p}) \hat{\lambda} - U'(\hat{x}) \hat{\theta}.
\end{align}
Equations (\ref{NewtonEqForNumMethod}) can be obtained from Eq. (\ref{KvNEqForNumMethod}) using the method presented in Sec. \ref{Sec:Liouville2Newton}. Let us present a more simple derivation: Eq. (\ref{NewtonEqForNumMethod}) can be obtained as the Heisenberg dynamical equation generated by $\hat{L}$,
\begin{align}\label{NewtonEqAsHeisenbergKvNEq}
	 \frac{d\hat{x}}{dt}  = i[ \hat{L}, \hat{x}] =   K'(\hat{p}), \qquad  \frac{d\hat{p}}{dt} = i[ \hat{L}, \hat{p}] = -U'(\hat{x}).
\end{align}

Note that this dynamics has the following conserved property, known as the total energy,
\begin{align}
	\hat{H} = K(\hat{p}) + U(\hat{x}), \qquad
	[\hat{H}, \hat{L}] = iU'(\hat{x}) K'(\hat{p}) -i K'(\hat{p}) U'(\hat{x}) = 0.
\end{align}
Equation (\ref{DisplacementOp}) represents the displacement operator.

Similarly to the time-independent Schr\"{o}dinger equation, the solution of Eq. (\ref{KvNEqForNumMethod}) reads 
\begin{align}
	 | \Psi(t) \rangle = e^{-i\hat{L} t}  | \Psi(0) \rangle 
	 \Longrightarrow  | \Psi(t) \rangle = e^{i[K'(\hat{p}) \hat{\lambda} - U'(\hat{x})\hat{\theta}]\delta t}  | \Psi(t+\delta t) \rangle.
\end{align}
Using Eqs. (\ref{SecondOrderSplitting}), (\ref{ClassicalXPRepresentation}), and (\ref{DisplacementOp}), we obtain
\begin{align}
	\langle x \, p | \Psi(t) \rangle &=  \langle x \, p | e^{-i\delta t U'(\hat{x})\hat{\theta}/2} e^{i\delta t K'(\hat{p}) \hat{\lambda}} e^{-i\delta t U'(\hat{x})\hat{\theta}/2} | \Psi(t+\delta t) \rangle + O\left( \delta t^3 \right)  
	\notag\\
	&=  e^{-\frac{\delta t}{2} U'(x)\frac{\partial}{\partial p}} \langle x \, p | e^{i\delta t K'(\hat{p}) \hat{\lambda}} e^{-i\delta t U'(\hat{x})\hat{\theta}/2} | \Psi(t+\delta t) \rangle + O\left( \delta t^3 \right)
	\notag\\
	&=  \langle x \, p_1 | e^{i\delta t K'(\hat{p}) \hat{\lambda}} e^{-i\delta t U'(\hat{x})\hat{\theta}/2} | \Psi(t+\delta t) \rangle + O\left( \delta t^3 \right)
	\notag\\
	&=  e^{\delta t K'(p_1) \frac{\partial}{\partial x}} \langle x \, p_1 | e^{-i\delta t U'(\hat{x})\hat{\theta}/2} | \Psi(t+\delta t) \rangle + O\left( \delta t^3 \right)
	\notag\\
	&=  \langle x_1 \, p_1 | e^{-i\delta t U'(\hat{x})\hat{\theta}/2} | \Psi(t+\delta t) \rangle + O\left( \delta t^3 \right)
	\notag\\
	&=  e^{-\frac{\delta t}{2} U'(x_1)\frac{\partial}{\partial p_1}} \langle x_1 \, p_1 | \Psi(t+\delta t) \rangle + O\left( \delta t^3 \right)
	\notag\\
	&=  \langle x_1 \, p_2 | \Psi(t+\delta t) \rangle + O\left( \delta t^3 \right), \label{VerletNumMethodEq2}
\end{align}
where
\BoxedEquation{\label{VerletNumMethodEq1}
	p_1 = p - U'(x) \delta t /2,
	\qquad
	x_1 = x + K'(p_1) \delta t,
	\qquad
	p_2 = p_1 - U'(x_1) \delta t /2.
}
Equations (\ref{VerletNumMethodEq1}) and (\ref{VerletNumMethodEq2}) constitute the Verlet integrator, which is a second-order symplectic (i.e., Hamiltonian preserving) numerical method. The derived propagator is easily connected to the cartoon of Koopman--von Neumann classical mechanics depicted in Fig. \ref{Fig_L_KvN_QM}(b). Equation (\ref{VerletNumMethodEq1}) specifies the shape of individual trajectories; whereas, Eq. (\ref{VerletNumMethodEq2}) guarantees the conservation of tags along the trajectories.

Numerical scheme (\ref{VerletNumMethodEq1}) could be easily generalized to many dimensional case. In this situation, a single particle state is given by the coordinate $\vec{x}$ and momentum $\vec{p}$ vectors. Then, the second-order time propagation from $\vec{x}(t) = \vec{x}$, $\vec{p}(t) = \vec{p}$ to $\vec{x}(t + \delta t) = \vec{x}_1$, $\vec{p}(t + \delta t) = \vec{p}_2$ is given by
\BoxedEquation{
	\vec{p}_1 = \vec{p} - \vec{\nabla} U(\vec{x}) \delta t /2,
	\qquad
	\vec{x}_1 = \vec{x} + \vec{\nabla} K(\vec{p}_1) \delta t,
	\qquad
	\vec{p}_2 = \vec{p}_1 - \vec{\nabla} U(\vec{x}_1) \delta t /2.
}
Regarding a Python implementation of the Verlet integrator to propagate an ensemble of classical particles in an arbitrary number of spatial dimensions, see 
\footnote{\url{https://github.com/dibondar/QuantumClassicalDynamics/blob/master/verlet_classical_integrator.py}}. Further details about simulations of classical mechanics can be found in \cite{blanes2016concise}.

\chapter{7. Open quantum systems}\label{Chapter:7}

Assume we performed a dynamical thought experiment as described in Sec. \ref{Sec:DynamicThoughtExp} and obtained the following Ehrenfest theorems
\begin{align}\label{SimpleDissipativeEhrenfestThrms}
	\frac{d}{dt} \langle x \rangle = \frac{1}{m} \langle p \rangle,
	\qquad
	\frac{d}{dt} \langle p \rangle = \langle -U'(x) \rangle  -\gamma \langle p \rangle.
\end{align}
These Ehrenfest theorems describe dissipative dynamics. Note that the force acting on the system is velocity-dependent and acts against the direction of motion. This is a friction force. Following the derivation presented in Sec. \ref{Sec:ODMSchrodinger}, we show that this dynamics yielding the Ehrenfest theorems (\ref{SimpleDissipativeEhrenfestThrms}) cannot be modeled in the same fashion as the Schr\"{o}dinger equation.

Assuming that the empirical averages are represented by the Dirac bra-ket sandwich
\begin{align}
	\langle A \rangle(t) = \langle \Psi(t) | \hat{A} | \Psi(t) \rangle,
\end{align}
where $ | \Psi(t) \rangle$ evolves according to the unitary dynamics, hence obeying Stone's theorem (see Sec.~\ref{Sec:Stones_Th}),
\begin{align}
	i\hbar \frac{d}{dt} | \Psi(t) \rangle = \hat{G} | \Psi(t) \rangle, \qquad  \hat{G}^{\dagger} =  \hat{G},
\end{align}
we obtain the commutator equations for unknown generator of motion $\hat{G}$
\begin{align}
	im [\hat{G}, \hat{x}]  = \hbar \hat{p},
	\qquad
	i [\hat{G}, \hat{p}] = -\hbar U'(\hat{x}) -\gamma \hbar \hat{p}.
\end{align}
Assuming $G = G(\hat{x}, \hat{p})$ and $[ \hat{x}, \hat{p} ] = i\hbar$, we obtain from Theorem Eq. (\ref{Weyl_derivative_of_function}) the following set of differential equations for the scalar function $G(x,p)$
\begin{align}
	\frac{\partial G(x,p)}{\partial p} = \frac{p}{m},
	\qquad
	\frac{\partial G(x,p)}{\partial x} = U'(x) + \gamma p.
\end{align}
However, it is easy to show that this equation has no solution since
\begin{align}
	\frac{\partial^2 G(x,p)}{\partial x \partial p} = 0,
	\qquad
	\frac{\partial^2 G(x,p)}{\partial p \partial x} = \gamma.
\end{align}
As a result, we need to think about alternative mathematical language to represent dissipative dynamics.

We would like to construct the quantum mechanical formalism that allows for
\begin{itemize}
	\item smooth limit $\hbar \to 0$ to recover the classical mechanics (that will be done in Sec. \ref{Sec:PhaseSpaceQM}),
	\item enable to model more general dynamics obeying the Heisenberg uncertainty principle abeit not describable by unitary evolution.
\end{itemize}

\section{Quantum mechanics in phase-space}\label{Sec:PhaseSpaceQM}

We begin with a curious observation that linear combinations of the Bopp operators ($\hat{\theta}$, $\hat{\lambda}$) and classical position $\hat{x}$ and momentum $\hat{p}$ (\ref{Ch2_ClassicalCommutator}) and (\ref{Complete_classical_algebra})
\BoxedEquation{
	\hat{\bs{x}} = \hat{x} - \hbar \hat{\theta} /2,
	\qquad
	\hat{\bs{p}} = \hat{p} + \hbar \hat{\lambda} /2, 	
}
obey the canonical commutation relationship
\begin{align} 
	[\hat{\bs{x}}, \hat{\bs{p}}] = i\hbar. 	
\end{align}
Let's re-derive the close system quantum evolution, which obeys the Ehrenfest theorems
\begin{align} 
	\frac{d}{dt} \langle \Psi(t) | \hat{\bs{x}} | \Psi(t) \rangle = \langle \Psi(t) | \hat{\bs{p}} | \Psi(t) \rangle / m,
	\qquad
	\frac{d}{dt} \langle \Psi(t) | \hat{\bs{p}} | \Psi(t) \rangle = \langle \Psi(t) | -U'(\hat{\bs{x}}) | \Psi(t) \rangle.
\end{align}
or in an expanded form
\begin{align} 
	&\frac{d}{dt} \langle \Psi(t) | \hat{x} - \frac{\hbar}{2} \hat{\theta} | \Psi(t) \rangle = \langle \Psi(t) | \hat{p} + \frac{\hbar}{2} \hat{\lambda} | \Psi(t) \rangle / m,
    \notag
	\\
	&\frac{d}{dt} \langle \Psi(t) | \hat{p} + \frac{\hbar}{2} \hat{\lambda} | \Psi(t) \rangle = \langle \Psi(t) | -U'\left(\hat{x} - \frac{\hbar}{2} \hat{\theta} \right) | \Psi(t) \rangle.
\end{align}
Assuming the equation of motion
\begin{align}\label{StonesTheInQPhaseSpace}
	i \frac{d}{dt} | \Psi(t) \rangle  = G(\hat{x}, \hat{\lambda}, \hat{p}, \hat{\theta} ) | \Psi(t) \rangle,
\end{align}
we obtain
\begin{align}
	m \frac{\partial G}{\partial \lambda} + \frac{m\hbar}{2} \frac{\partial G}{\partial p} = p + \frac{\hbar}{2} \lambda,
	\qquad
	\frac{\partial G}{\partial \theta} - \frac{\hbar}{2}\frac{\partial G}{\partial x} = - U'\left( x - \frac{\hbar}{2}\theta \right).
\end{align}
Using the method of characteristics, one can find the most general solution for $G$
\begin{align}\label{GGeneratorEq1}
	\hat{G} = \frac{\hat{p}\hat{\lambda}}{m} + \frac{1}{\hbar}\left[ 
		U\left( \hat{x} - \frac{\hbar}{2}\hat{\theta} \right) + F\left( \hat{p} - \frac{\hbar}{2} \hat{\lambda}, \hat{x} + \frac{\hbar}{2}\hat{\theta} \right) 
	\right],
\end{align}
where $F$ is an arbitrary real-valued function. Equation (\ref{GGeneratorEq1}) should be compared with the classical generator of motion (\ref{KvNGenerator}). Now we can use the freedom of choice of $F$ to make the classical limit $\hbar \to 0$ smooth. This expression $F\left( \hat{p} - \hbar \hat{\lambda}/2, \hat{x} + \hbar\hat{\theta}/2 \right)  = -U\left( \hat{x} + \hbar\hat{\theta}/2 \right)$ does the job. Recalling the central finite difference approximation (\ref{CentralFinitDiffApprox}), it is easy to show that
\begin{align}
	\lim_{\hbar \to 0} \hat{G} = \hat{p}\hat{\lambda} / m - U'(\hat{x})\hat{\theta}
\end{align}
where
\begin{align}\label{GGenerator}
	\hat{G} = \frac{\hat{p}\hat{\lambda}}{m} + \frac{1}{\hbar}\left[ 
		U\left( \hat{x} - \frac{\hbar}{2}\hat{\theta} \right) - U\left( \hat{x} + \frac{\hbar}{2}\hat{\theta} \right) 
	\right].
\end{align}
We see that $\hbar$ acts as a finite action approximation of the classical dynamics. The generator of motion (\ref{GGenerator}) can be generalized as
\begin{align}\label{GGeneralGenerator}
	\hat{G} = \frac{1}{\hbar} H\left( \hat{x} - \frac{\hbar}{2}\hat{\theta}, \hat{p} + \frac{\hbar}{2} \hat{\lambda} \right)
		- \frac{1}{\hbar} H\left( \hat{x} + \frac{\hbar}{2}\hat{\theta}, \hat{p} - \frac{\hbar}{2} \hat{\lambda} \right), 
\end{align}
where $H(\bs{\hat{x}}, \bs{\hat{p}}) = K(\bs{\hat{p}}) + U(\bs{\hat{x}})$ is a Hamiltonian in a generalized form.

Let us introduce some convenient notation: Arrows on top of an operator to denote the direction of action the operator, e.g.,
\begin{align}\label{DirectionalDerivatives}
	g(y)\overrightarrow{\frac{\partial}{\partial y}} f(y) = f(y) \overleftarrow{\frac{\partial}{\partial y}} g(y) = g(y) \frac{\partial f(y)}{\partial y},
\end{align}
the displacement operator (\ref{DisplacementOp}) can be redefined in this context
\begin{align}\label{DirectionalDisplacementOp}
	g(y) e^{a\overrightarrow{\frac{\partial}{\partial y}}} f(y) = f(y) e^{a\overleftarrow{\frac{\partial}{\partial y}}} g(y) = g(y) f(y + a).
\end{align}

Projecting dynamical equation (\ref{StonesTheInQPhaseSpace}) with (\ref{GGeneralGenerator}) onto $\langle x \, p|$, the common (left) eigenvector of the commuting self-adjoint operators $\hat{x}$ and $\hat{p}$ yields equation for 
\BoxedEquation{
	\langle x \, p | \Psi(t) \rangle / \sqrt{2\pi} = W(x,p; t)
}
\begin{align}\label{PreWignerEq1}
	\frac{\partial}{\partial t}  W(x,p; t)  
	&= \frac{1}{i\hbar} \langle x \, p| H\left( \hat{x} - \frac{\hbar}{2}\hat{\theta}, \hat{p} + \frac{\hbar}{2} \hat{\lambda} \right) | \Psi(t) \rangle
		- \frac{1}{i\hbar} \langle x \, p| H\left( \hat{x} + \frac{\hbar}{2}\hat{\theta}, \hat{p} - \frac{\hbar}{2} \hat{\lambda} \right) | \Psi(t) \rangle.
\end{align}
Employing newly introduced notation (\ref{DirectionalDerivatives}) along with Eq. (\ref{ClassicalXPRepresentation}), we obtain a representation that will later play an important role below (Sec. \ref{Sec:Molecular2StateModel})
\begin{align}
	i\hbar \frac{\partial}{\partial t}  W(x,p; t)  
	&=  H\left( x - \frac{\hbar}{2} \overrightarrow{\hat{\theta}}, p + \frac{\hbar}{2} \overrightarrow{\hat{\lambda}} \right) W(x,p; t)
		- W(x,p; t) H\left(x + \frac{\hbar}{2} \overleftarrow{\hat{\theta}}, p - \frac{\hbar}{2} \overleftarrow{\hat{\lambda}} \right).
\end{align}
Returning to Eq. (\ref{PreWignerEq1}) and using identities (\ref{ClassicalXPRepresentation}) and (\ref{DirectionalDisplacementOp}), one yields
\begin{align}
	&\frac{\partial}{\partial t}  W(x,p; t)  
	= \frac{1}{i\hbar} H\left( x + i\frac{\hbar}{2} \overrightarrow{\frac{\partial}{\partial p}},
				p -i\frac{\hbar}{2} \overrightarrow{\frac{\partial}{\partial x}} \right) W(x,p; t)
		- \frac{1}{i\hbar}  H\left( x - i\frac{\hbar}{2} \overrightarrow{\frac{\partial}{\partial p}},
				p + i\frac{\hbar}{2} \overrightarrow{\frac{\partial}{\partial x}} \right) W(x,p; t)
	\notag\\
	&= \frac{1}{i\hbar} \left[
		H\left( x + i\frac{\hbar}{2} \overrightarrow{\frac{\partial}{\partial p}}, p \right)
		\exp\left( -i\frac{\hbar}{2} \overrightarrow{\frac{\partial}{\partial x}} \overleftarrow{\frac{\partial}{\partial p}} \right)
		-H\left( x - i\frac{\hbar}{2} \overrightarrow{\frac{\partial}{\partial p}}, p \right)
		\exp\left( i\frac{\hbar}{2} \overrightarrow{\frac{\partial}{\partial x}} \overleftarrow{\frac{\partial}{\partial p}} \right)
	\right] W(x,p; t)
	\notag\\
	&= \frac{1}{i\hbar} \left[ H(x,p) \exp\left(
			-i\frac{\hbar}{2} \overrightarrow{\frac{\partial}{\partial x}} \overleftarrow{\frac{\partial}{\partial p}} 
			+i\frac{\hbar}{2} \overrightarrow{\frac{\partial}{\partial p}} \overleftarrow{\frac{\partial}{\partial x}} 
		\right)
		- H(x,p) \exp\left(
			 i\frac{\hbar}{2} \overrightarrow{\frac{\partial}{\partial x}} \overleftarrow{\frac{\partial}{\partial p}}
			 - i\frac{\hbar}{2} \overrightarrow{\frac{\partial}{\partial p}} \overleftarrow{\frac{\partial}{\partial x}} 
		\right)
	\right] W(x,p; t)
	\notag\\
	& = \frac{1}{i\hbar} \left[ H(x,p) \exp\left(
			-i\frac{\hbar}{2} \overrightarrow{\frac{\partial}{\partial x}} \overleftarrow{\frac{\partial}{\partial p}} 
			+i\frac{\hbar}{2} \overrightarrow{\frac{\partial}{\partial p}} \overleftarrow{\frac{\partial}{\partial x}} 
		\right) W(x,p;t)\right.
    \notag\\
    &\qquad		\left.
            - W(x,p;t) \exp\left(
			i\frac{\hbar}{2} \overleftarrow{\frac{\partial}{\partial x}} \overrightarrow{\frac{\partial}{\partial p}}
            - i\frac{\hbar}{2} \overleftarrow{\frac{\partial}{\partial p}} \overrightarrow{\frac{\partial}{\partial x}} 
		\right) H(x,p)
	\right].
\end{align}
Finally,
\BoxedEquation{
	\frac{\partial}{\partial t}  W(x,p; t)  &= \{\{ H(x,p), W(x,p; t) \}\}, \label{MoyalEq}\\
	\{\{ f, g \}\} &= \frac{f \star g - g \star f}{i\hbar}
			= \frac{2}{\hbar} f \sin\left(
			 	\frac{\hbar}{2} \overleftarrow{\frac{\partial}{\partial x}} \overrightarrow{\frac{\partial}{\partial p}}
				 - \frac{\hbar}{2} \overleftarrow{\frac{\partial}{\partial p}} \overrightarrow{\frac{\partial}{\partial x}} 
			\right)g, \label{MoyalBracket}\\
	\star &= \exp\left(
			 \frac{i\hbar}{2} \overleftarrow{\frac{\partial}{\partial x}} \overrightarrow{\frac{\partial}{\partial p}}
			 - \frac{i\hbar}{2} \overleftarrow{\frac{\partial}{\partial p}} \overrightarrow{\frac{\partial}{\partial x}} 
		\right). \label{MoyalStar}
}
Equations (\ref{MoyalEq}), (\ref{MoyalBracket}), and (\ref{MoyalStar}) are known as the Wigner-Moyal equation \cite{moyal1949quantum, zachos2005quantum, curtright2012quantum}, Moyal bracket,  and Moyal star \cite{PhysRevD.58.025002, zachos2001, curtright2014concise}, respectively; whereas, $W(x,p; t)$ is the Wigner phase-space quasi distribution \cite{Wigner1932}. These constructions play paramount role in the phase space representation of quantum mechanics  \cite{zachos2005quantum, curtright2012quantum, curtright2014concise}.

The identity 
$
	\sin(y) = \left( e^{iy} - e^{-iy} \right)/(2i)
$
was used to derive Eq. (\ref{MoyalBracket}). Note that the Wigner function is real valued function according to Eq. (\ref{MoyalBracket}).

The Wigner function is connected with the density matrix
\BoxedEquation{\label{WignerFuncDeffEq}
	W(x,p; t) = \frac{1}{2\pi} \int  
		\left\langle  x - \frac{\hbar}{2}\theta \right| \hat{\bs{\rho}}(t)  \left| x + \frac{\hbar}{2}\theta \right\rangle 
		e^{i p \theta}  d \theta
}
where $\langle  \bs{x} | \hat{\bs{\rho}}  | \bs{x'}  \rangle$ denotes a density matrix in the coordinate representation, as we shall show now.

%%%%%%%%%%%%%%%%%%%%%%%%%%%%%%%%%%%%%%%%%%%%%%%%%%%%
\begin{figure}
    \centering
	\includegraphics[width=0.6\hsize]{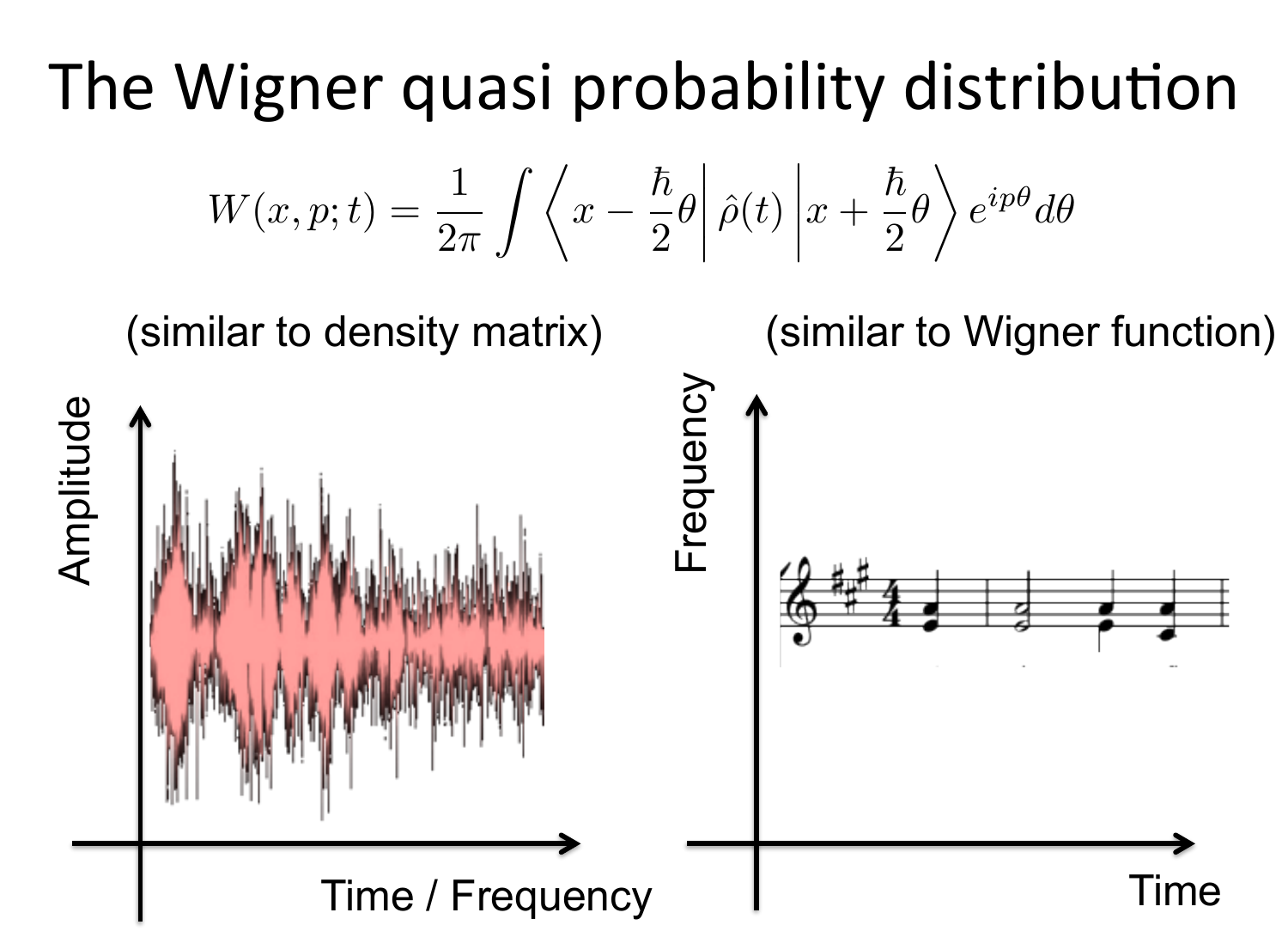}
    \caption{The Wigner function (\ref{WignerFuncDeffEq}) as musical score of quantum mechanics.}
\end{figure}
%%%%%%%%%%%%%%%%%%%%%%%%%%%%%%%%%%%%%%%%%%%%%%%%%%%%

The commutator relations \eqref{Complete_classical_algebra} and Eq. \eqref{EqSummaryCommutatorSection} imply
\begin{align}\label{ClassicalXThetaRepresentation}
	\langle x \, \theta | \hat{x} |\Psi \rangle = x \langle x \, \theta |\Psi \rangle,
	\quad
	\langle x \, \theta | \hat{\lambda}  |\Psi \rangle = -i\frac{\partial}{\partial x} \langle x \, \theta |\Psi \rangle,
	%\quad
    \notag\\
	\langle x \, \theta | \hat{p} |\Psi \rangle = i\frac{\partial}{\partial \theta} \langle x \, \theta |\Psi \rangle,
	\quad
	\langle x \, \theta | \hat{\theta}  |\Psi \rangle = \theta \langle x \, \theta |\Psi \rangle.
\end{align}
Projecting dynamical equation (\ref{StonesTheInQPhaseSpace}) with (\ref{GGeneralGenerator}) onto $\langle x \, \theta |$, the common (left) eigenvector of the commuting self-adjoint operators $\hat{x}$ and $\hat{\theta}$ yields equation for 
\BoxedEquation{
	\langle x \, \theta | \Psi(t) \rangle = B(x,\theta; t),
}
\begin{align}
	i\hbar \frac{\partial}{\partial t}  B(x,\theta; t)  
	&= \langle x \, \theta | H\left( \hat{x} - \frac{\hbar}{2}\hat{\theta}, \hat{p} + \frac{\hbar}{2} \hat{\lambda} \right) | \Psi(t) \rangle
		- \langle x \, \theta | H\left( \hat{x} + \frac{\hbar}{2}\hat{\theta}, \hat{p} - \frac{\hbar}{2} \hat{\lambda} \right) | \Psi(t) \rangle,
\end{align}
$B(x,\theta; t)$ is known as the \emph{Blokhintsev function}.
 
Employing newly introduced notation (\ref{DirectionalDerivatives}) along with Eq. (\ref{ClassicalXThetaRepresentation}), we obtain 
\begin{align}
	i\hbar \frac{\partial}{\partial t}  B(x,\theta; t)  
	&= H\left( x - \frac{\hbar}{2}\theta, 
			i\overrightarrow{\frac{\partial}{\partial \theta}} -i \frac{\hbar}{2} \overrightarrow{\frac{\partial}{\partial x}}
		\right) B(x,\theta; t)  
	 - B(x,\theta; t) H\left( x + \frac{\hbar}{2}\theta, 
	 	i\overleftarrow{\frac{\partial}{\partial \theta}} +i \frac{\hbar}{2} \overleftarrow{\frac{\partial}{\partial x}}
	\right).
\end{align}
Using new variables
\begin{align}
	\bs{x} = x - \hbar \theta  /2, \qquad \bs{x'} = x + \hbar \theta  /2 
	\quad \Longrightarrow \quad
	x = (\bs{x'} + \bs{x})/2, \qquad  \theta = (\bs{x'} - \bs{x})/\hbar
	\notag\\
	\frac{\partial}{\partial \theta} = \frac{\partial \bs{x'}}{\partial \theta}\frac{\partial}{\partial \bs{x'}}  
		+ \frac{\partial \bs{x}}{\partial \theta}\frac{\partial}{\partial \bs{x}}
		= \frac{\hbar}{2} \left( \frac{\partial}{\partial \bs{x'}} - \frac{\partial}{\partial \bs{x}} \right),
	\qquad
	\frac{\partial}{\partial x} = \frac{\partial \bs{x'}}{\partial x}\frac{\partial}{\partial \bs{x'}}  
		+ \frac{\partial \bs{x}}{\partial x}\frac{\partial}{\partial \bs{x}}
		= \frac{\partial}{\partial \bs{x'}} + \frac{\partial}{\partial \bs{x}}, 
\end{align}
the equation for $\rho(\bs{x}, \bs{x'}; t) = B\left( \frac{\bs{x'} + \bs{x}}{2}, \frac{\bs{x'} - \bs{x}}{\hbar}; t\right)$ reads
\begin{align}
	i\hbar \frac{\partial}{\partial t}  \rho(\bs{x}, \bs{x'}; t) 
	&= H\left( \bs{x},  -i\hbar\overrightarrow{\frac{\partial}{\partial \bs{x}}} \right)  \rho(\bs{x}, \bs{x'}; t)
		- \rho(\bs{x}, \bs{x'}; t) H\left( \bs{x'},  i\hbar\overleftarrow{\frac{\partial}{\partial \bs{x'}}} \right).
\end{align}
The latter equation is the \emph{von Neumann} equation for the density matrix in the coordinate representation. Indeed, if we introduce the operators $\hat{\bs{x}}$ and $\hat{\bs{p}}$ such that
\begin{align}
	[\hat{\bs{x}}, \hat{\bs{p}}] = i \hbar 
	\quad \Longrightarrow \quad
	\hat{\bs{x}} | \bs{x} \rangle =  \bs{x} | \bs{x} \rangle, 
	\qquad
	\langle \bs{x} | \hat{\bs{p}} | \phi \rangle = -i\hbar \overrightarrow{\frac{\partial}{\partial \bs{x}}} \langle \bs{x} | \phi \rangle,
	\qquad
	\langle \phi | \hat{\bs{p}} | \bs{x} \rangle = i\hbar \langle \phi | \bs{x} \rangle \overleftarrow{\frac{\partial}{\partial \bs{x}}},
\end{align}
and $\rho(\bs{x}, \bs{x'}; t) = \langle \bs{x} | \hat{\bs{\rho}}(t) | \bs{x'} \rangle $ then
\begin{align}
	i\hbar \frac{\partial}{\partial t}  \langle \bs{x} | \hat{\bs{\rho}}(t) | \bs{x'} \rangle
	&= H\left( \bs{x},  -i\hbar\overrightarrow{\frac{\partial}{\partial \bs{x}}} \right) \langle \bs{x} | \hat{\bs{\rho}}(t) | \bs{x'} \rangle
		- \langle \bs{x} | \hat{\bs{\rho}}(t) | \bs{x'} \rangle H\left( \bs{x'},  i\hbar\overleftarrow{\frac{\partial}{\partial \bs{x'}}} \right)
	\notag\\
	& =  \langle \bs{x} | H\left( \hat{\bs{x}},  \hat{\bs{p}} \right) \hat{\bs{\rho}}(t) | \bs{x'} \rangle
		- \langle \bs{x} | \hat{\bs{\rho}}(t) H\left( \hat{\bs{x}},  \hat{\bs{p}} \right) | \bs{x'} \rangle \notag
\end{align}
Therefore, we can write the {von Neumann} equation as
\BoxedEquation{\label{EqVonNeumannDensityMatrxi}
	i\hbar \frac{\partial}{\partial t}  \hat{\bs{\rho}}(t)  &= [ H\left( \hat{\bs{x}},  \hat{\bs{p}} \right), \hat{\bs{\rho}}(t)  ].
}
Substituting $\hat{\bs{\rho}} (t) = \ket{\varphi (t)}\bra{\varphi (t)}$ into Eq. (\ref{EqVonNeumannDensityMatrxi}), we obtain
\begin{align}
	i\hbar \left( \frac{\partial}{\partial t}  \ket{\varphi} \right) \bra{\varphi} + 
	i\hbar  \ket{\varphi} \left( \frac{\partial}{\partial t} \bra{\varphi} \right) 
	= \hat{H} \ket{\varphi}\bra{\varphi} - \ket{\varphi}\bra{\varphi} \hat{H} \Longrightarrow \notag\\
	 \ket{\eta} \bra{\varphi} = -\ket{\varphi} \bra{\eta}, \label{EqEtaPhiMinusPhusEta}
\end{align}
where $\ket{\eta} = i\hbar \frac{\partial}{\partial t}  \ket{\varphi} - \hat{H} \ket{\varphi}$. Sandwiching Eq. \eqref{EqEtaPhiMinusPhusEta} between $\bra{\varphi}$ and $\ket{\varphi}$ and recalling $\langle  \varphi \ket{\varphi} = 1$ (implying that $\ket{\varphi}$ is not a null vector), one gets
$ \langle \varphi \ket{\eta} \bra{\varphi} \varphi \rangle = - \langle \varphi | \varphi \rangle \langle \eta | \varphi \rangle \Longrightarrow   \langle \varphi | \eta\rangle = - \langle \varphi | \eta \rangle^*$. Combining the latter with the definition of $\ket{\eta}$, we further obtain
\begin{align}
	\bra{\varphi}  i\hbar \frac{\partial}{\partial t}  \ket{\varphi} - \bra{\varphi} \hat{H} \ket{\varphi}
	= i\hbar \left[ \bra{\varphi}  \frac{\partial}{\partial t}  \ket{\varphi} \right]^* + \bra{\varphi} \hat{H} \ket{\varphi}, \notag\\
	\left[ \bra{\varphi}  \frac{\partial}{\partial t}  \ket{\varphi} \right]^*
	=  \left( \frac{\partial}{\partial t} \bra{\varphi} \right)   \ket{\varphi} 
	= \frac{\partial}{\partial t} [ \langle \varphi | \varphi \rangle] - \bra{\varphi} \frac{\partial}{\partial t}  \ket{\varphi} 
	= - \bra{\varphi} \frac{\partial}{\partial t}  \ket{\varphi} \Longrightarrow \notag\\
	\bra{\varphi(t)} \left(  i\hbar  \frac{\partial}{\partial t}  - \hat{H} \right) \ket{\varphi(t)} = 0. \label{EqVarphiAvgSchrodEq}
\end{align}
Since $\ket{\varphi(t)}$ is a unitary transformation of the initial condition $\ket{\varphi(t=0)}$ and Eq. \eqref{EqVarphiAvgSchrodEq} should be valid for any initial state, Eq. \eqref{EqVarphiAvgSchrodEq} reduces to the Schr\"{o}dinger equation
\begin{align}
	 i\hbar  \frac{\partial}{\partial t} \ket{\varphi} = \hat{H} \ket{\varphi}.
\end{align}

Therefore, Eq. (\ref{EqVonNeumannDensityMatrxi}) is an equation of motion for a state that can go beyond the wave function. From the wave function normalization, we obtain
\begin{align}
	1 = \bra{\varphi} \varphi \rangle = \int d\bs{x} \, \bra{\varphi} \bs{x} \rangle\langle \bs{x} \ket{\varphi}
	 = \int d\bs{x} \, \langle \bs{x} \ket{\varphi} \bra{\varphi} \bs{x} \rangle
	 = \int d\bs{x} \, \langle \bs{x} | \hat{\bs{\rho}} | \bs{x} \rangle. \notag
\end{align}
This gives us the normalization condition
\BoxedEquation{
	\Tr \hat{\bs{\rho}} = 1.
}
Furthermore,
\begin{align}
	\Tr \hat{\bs{\rho}} = 1 \Longrightarrow & \int d\bs{x} \rho(\bs{x}, \bs{x}) = \int d\bs{x}d\bs{x'} \delta(\bs{x'} - \bs{x}) \rho(\bs{x}, \bs{x'})
	= \hbar \int dx d\theta \delta(\hbar \theta) B(x,\theta) 
	\notag\\
	&=  \int dx d\theta \delta(\theta) B(x,\theta) = \int dx d\theta \delta(\theta) \langle x \, \theta | \Psi \rangle
	= \int dx d\theta \langle 1 | x \, \theta \rangle \langle x \, \theta | \Psi \rangle = \langle 1 | \Psi \rangle = 1.
\end{align}
Here, we have introduced the bra vector $\langle 1 |$ such that $\langle 1 | x \, \theta \rangle = \delta(\theta)$, corresponding to the identity density matrix $\hat{\bs{\rho}}=\hat{1}$. From Eq. (\ref{EqSummaryCommutatorSection}), we  get
\begin{align}
	\langle x \, p | x' \, \theta \rangle = \delta(x-x')  e^{ip\theta} / \sqrt{2\pi}.
\end{align}
\begin{align}\notag
	1 =  \langle 1 | \Psi \rangle = \int dx d\theta dx' dp \langle 1 | x' \, \theta \rangle\langle x' \, \theta | x \, p \rangle\langle  x \, p | \Psi \rangle 
	= \int dx d\theta dx' dp \, \delta(\theta) \delta(x-x')  e^{-ip\theta} W(x,p) 
	\Longrightarrow 
\end{align}
\BoxedEquation{
	\int  W(x,p) dx dp = 1.
}
\begin{align}
	W(x,p) = \int \frac{dx' d\theta}{\sqrt{2\pi}} \langle x \, p | x' \, \theta \rangle\langle x' \, \theta | \Psi \rangle 
	= \int \frac{dx' d\theta}{2\pi} \delta( x - x') e^{i p \theta}  B (x, \theta)
    \notag\\
	=  \int \frac{d \theta}{2\pi} \left\langle  x - \frac{\hbar}{2}\theta \right| \hat{\bs{\rho}} \left| x + \frac{\hbar}{2}\theta \right\rangle e^{i p \theta}. 
\end{align}
Thus, we have recovered Eq. (\ref{WignerFuncDeffEq}). Table \ref{Table:WignerVsDensityMatrixRep} is a summary of this section.

\begin{table}
\begin{tabular}{|c|c|c|}
	\hline
		& Density matrix formalism & ~~~~~~Phase space formalism~~~~~~ \\
	\hline \hline
	System state & $\hat{\bs{\rho}}^{\dagger} = \hat{\bs{\rho}}$ & $W(x,p) = W^* (x,p)$ 
    \\		     & non-negative             & real valued function  
    \\           & self-adjoint operator    & obtained via Eq.~(\ref{WignerFuncDeffEq})
    \\
	\hline
	Normalization condition & $\Tr \hat{\bs{\rho}} = 1$ & $\int W(x,p) dxdp = 1$ \\
	\hline
	Momentum distribution & $\langle \bs{p} | \hat{\bs{\rho}} | \bs{p} \rangle$, $\hat{\bs{p}} | \bs{p} \rangle = \bs{p} | \bs{p} \rangle$ & $\int W(x,p)dx$ \\
	\hline
	Coordinate distribution & $\langle \bs{x} | \hat{\bs{\rho}} | \bs{x} \rangle$, $\hat{\bs{x}} | \bs{x} \rangle = \bs{x} | \bs{x} \rangle$ & $\int W(x,p)dp$ \\
	\hline
	Observable multiplication & $\hat{A}\hat{B}$ & $A \star B$ \\
		& matrix product & Moyal star \\
	\hline
	Algebra of observables & $[ \hat{\bs{x}}, \hat{\bs{p}} ] = i\hbar$ & $x \star p - p \star x = i\hbar$ \\
	\hline
	  Observable averages $\langle A \rangle$ & $\Tr[A(\hat{\bs{x}}, \hat{\bs{p}}) \hat{\bs{\rho}}]$ &  $\int A(x,p) W(x,p) dx dp$ \\
	 \hline
	Equation of motion & $i\hbar \frac{d}{dt} \hat{\bs{\rho}} = [ \hat{H}, \hat{\bs{\rho}} ]$ & $i\hbar \frac{d}{d t}  W  = H \star W - W \star H$  \\
	 \hline
	 	& $\Tr(\hat{A}\hat{B}) = \Tr(\hat{B}\hat{A})$ & $\int A(x,p) \star B(x,p) dx dp $ \\
		& cyclic property of trace    & $= \int B(x,p) \star A(x,p) dx dp$
  \\    &                             & $=  \int A(x,p) B(x,p) dx dp$
  \\    &                             & ``lone star'' identity 
  \\
	\hline
\end{tabular}\caption{Quantum mechanics in the density matrix formalism and phase-space representation.}\label{Table:WignerVsDensityMatrixRep}
\end{table}

\section{Hamilton formulation of classical mechanics}

The classical limit $\hbar\to 0$ of the Moyal bracket (\ref{MoyalBracket}) is the celebrated \emph{Poisson bracket} 
\begin{align}
	\llbracket f, g \rrbracket = \lim_{\hbar\to 0} \{\{ f, g \}\} 
	=  \lim_{\hbar\to 0} \frac{2}{\hbar} f \sin\left(
			 	\frac{\hbar}{2} \overleftarrow{\frac{\partial}{\partial x}} \overrightarrow{\frac{\partial}{\partial p}}
				 - \frac{\hbar}{2} \overleftarrow{\frac{\partial}{\partial p}} \overrightarrow{\frac{\partial}{\partial x}} 
			\right)g
	= f \left( \overleftarrow{\frac{\partial}{\partial x}} \overrightarrow{\frac{\partial}{\partial p}}
		-  \overleftarrow{\frac{\partial}{\partial p}} \overrightarrow{\frac{\partial}{\partial x}} \right) g
		\Longrightarrow
\end{align}
\BoxedEquation{
	\llbracket f, g \rrbracket = \frac{\partial f}{\partial x} \frac{\partial g}{\partial p} -  \frac{\partial f}{\partial p} \frac{\partial g}{\partial x}.
}
The Poisson bracket plays a paramount role in classical mechanics. In fact, the adjective ``symplectic'' in the title of Sec. \ref{Sec:SymplecticPropagator} means the Poisson bracket structure; hence, symplectic integrators preserve the Poisson bracket structure.

Newton's equation (\ref{NewtonEqForNumMethod}) can be rewritten in the form known as \emph{Hamilton's equations}
\begin{align}\label{EqHamiltonEq}
	\frac{d x}{dt} = \llbracket x, H \rrbracket,
	\qquad
	\frac{d p}{dt} = \llbracket p, H \rrbracket,
	\qquad
	H = K(p) + U(x).
\end{align}
Note that this equation resemble Eq.~(\ref{NewtonEqAsHeisenbergKvNEq}). Additionally, the Liouville equation (\ref{Liouville_Eq}) can be rewritten as 
\begin{align}
	\frac{d \rho}{dt} = \llbracket H, \rho \rrbracket,
	\qquad
	H = \frac{p^2}{2m} + U(x).
\end{align}
All classical mechanics can be recast via the Poisson bracket.

\section{Heisenberg uncertainty principle in the density matrix formalism}\label{Sec:UncertaintyPrincDensMatrix}

By revisiting the Heisenberg uncertainty principle in the density matrix formalism, we will understand what more general form (beyond Stone's theorem) of quantum dynamics is allowed (see Sec.~\ref{Sec:WaveFuncHeisUncertaintyPrinc} regarding the derivation of the uncertainty principle in the wave function formalism).

We begin by noting that this expression 
\BoxedEquation{
	\langle \hat{A} | \hat{B} \rangle = \Tr( \hat{A}^{\dagger} \hat{B} ),
}
obey all the axioms for the scalar product (see Sec. \ref{Sec:DiracBraKetFormalism}). This scalar product between two operators $\hat{A}$ and $\hat{B}$ are known as the Hilbert-Schmidt product. Having defined $\langle \hat{A} | \hat{B} \rangle$, the Cauchy-Schwarz inequality (\ref{CauchySchwarzIneq}) yields
\begin{align}
	\langle \hat{A} | \hat{A} \rangle \langle \hat{B} | \hat{B} \rangle  \geq | \langle \hat{A} | \hat{B} \rangle |^2
	\Longrightarrow
	 \Tr( \hat{A}^{\dagger} \hat{A} )   \Tr( \hat{B}^{\dagger} \hat{B} ) \geq  \left| \Tr( \hat{A}^{\dagger} \hat{B} ) \right|^2.
\end{align}
Let's work out the standard deviation for the coordinate ($\sigma_x$) and momentum ($\sigma_p$) is similar fashion to Sec. \ref{Sec:WaveFuncHeisUncertaintyPrinc}
\begin{align}\label{HilbSchmCauchySchwarzIneq}
	\sigma_x^2 &= \Tr\left[ (\hat{x} - \langle \hat{x} \rangle )^2 \hat{\rho} \right] 
	= \Tr\left[ (\hat{x} - \langle \hat{x} \rangle )^2 \sqrt{\hat{\rho}} \sqrt{\hat{\rho}}^{\dagger} \right]
	\notag\\
	&= \Tr\left[ \sqrt{\hat{\rho}}^{\dagger} (\hat{x} - \langle \hat{x} \rangle ) (\hat{x} - \langle \hat{x} \rangle )  \sqrt{\hat{\rho}} \right] 
	= \Tr\left[ \left( (\hat{x} - \langle \hat{x} \rangle ) \sqrt{\hat{\rho}} \right)^{\dagger} (\hat{x} - \langle \hat{x} \rangle )  \sqrt{\hat{\rho}} \right];
\end{align}
likewise,
\begin{align}
	\sigma_p^2 & 
	= \Tr\left[ \left( (\hat{p} - \langle \hat{p} \rangle ) \sqrt{\hat{\rho}} \right)^{\dagger} (\hat{p} - \langle \hat{p} \rangle )  \sqrt{\hat{\rho}} \right].
\end{align}

Here is the \emph{key point}: For this derivation to work, a square root of the density matrix $\hat{\rho}$ must exists! This is only possible if $\hat{\rho}$ is non-negative self-adjoint operator (i.e., an operator whose eigenvalues are zeros or positive real numbers).

Note a square root $\sqrt{\hat{A}}$ of a non-negative operator $\hat{A}$ acting on a complex Hilbert space is defined as a solution of the equation: $\sqrt{\hat{A}}^{\dagger} \sqrt{\hat{A}} = \hat{A}$. A square root of a self-adjoint operator need not be self-adjoint; moreover, it is not unique since $\hat{U} \sqrt{\hat{A}}$ is also a valid square root for an arbitrary unitary $\hat{U}$. However, there is a unique  non-negative operator $\sqrt{\hat{A}}$ such that  $\sqrt{\hat{A}}^2 = \hat{A}$.

Using Eqs. (\ref{AuxiliaryComplexIneq}) and (\ref{HilbSchmCauchySchwarzIneq}), we obtain
\begin{align}
	\sigma_x^2 \sigma_p^2  
		&\geq \left| \Tr\left[ \left( (\hat{x} - \langle \hat{x} \rangle ) \sqrt{\hat{\rho}} \right)^{\dagger} (\hat{p} - \langle \hat{p} \rangle )  \sqrt{\hat{\rho}} \right] \right|^2 
		= \Big| \Tr\left[ (\hat{x} - \langle \hat{x} \rangle ) (\hat{p} - \langle \hat{p} \rangle ) \hat{\rho} \right] \Big|^2 
		\notag\\
		&\geq \left| \frac{1}{2} \Tr\left[ (\hat{x} - \langle \hat{x} \rangle ) (\hat{p} - \langle \hat{p} \rangle ) \hat{\rho} - (\hat{p} - \langle \hat{p} \rangle ) (\hat{x} - \langle \hat{x} \rangle ) \hat{\rho} \right] \right|^2  = \left| \frac{1}{2} \Tr( [\hat{x}, \hat{p}]\hat{\rho} )\right|^2 = \frac{\hbar^2}{4}.
\end{align}

In conclusion, the non-negativity of the density matrix implies the uncertainty principle. Thus, we would like to find equation of motion preserving not only density matrix's positivity, but also keeping its trace a unity.

\section{Choi's theorem}\label{Sec:ChoiThrm}

\begin{figure}
    \centering
	\includegraphics[width=0.55\hsize]{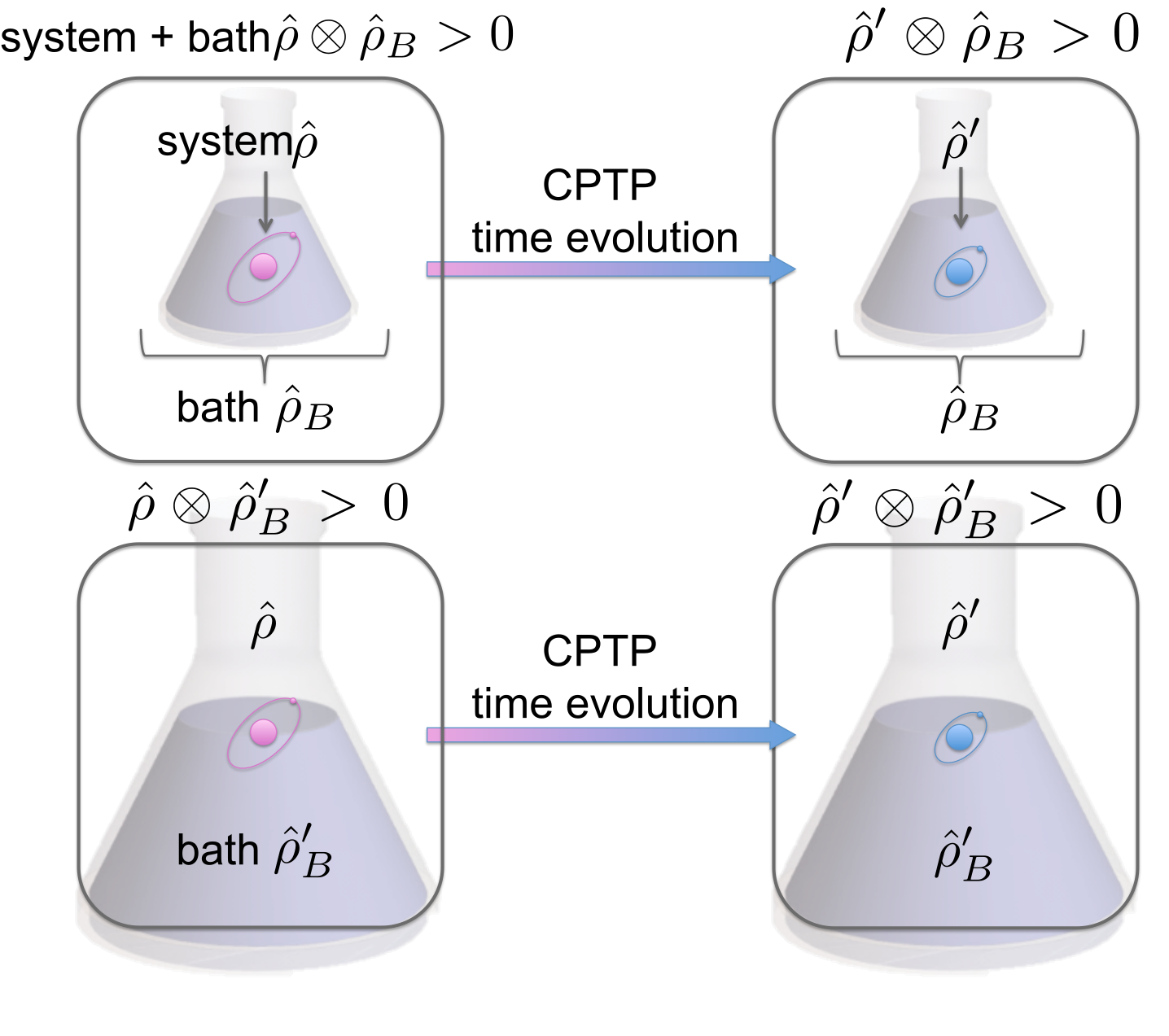}
	\caption{The illustration of the notion of Complete Positivity Trace Preserving Evolution (CPTP). The density matrix of the combined system and bath remains non-negative defined irrespectively of the bath size.}\label{Fig:CPTP}
\end{figure}

In Sec. \ref{Sec:UncertaintyPrincDensMatrix}, we argued that quantum evolution must be modeled by a positivity and trace preserving linear transformation of the density matrix. There is something else needed.

A quantum system under study, characterized by the initial density matrix $\hat{\rho}$, may be coupled to another system with the density matrix $\hat{\rho}_B$, such that the state of the total system is characterized by the initial density matrix $\hat{\rho} \otimes \hat{\rho}_B > 0$ (see Fig.~\ref{Fig:CPTP}).   Hence, when our system evolves from $\hat{\rho}$ to $\hat{\rho}'$ (without affecting $\hat{\rho}_B$), the final density matrix of the total system, $\hat{\rho}' \otimes \hat{\rho}_B > 0$, must still be positive. The linear transformation $\hat{\rho} \to \hat{\rho}'$ obeying this property for an arbitrary large $\hat{\rho}_B$ (i.e., any dimensionality of the matrix $\hat{\rho}_B$) is called \emph{completely positive trace preserving map} (the abbreviation CPTP is also extensively used).

%{\bf(FIX THIS SECTION)}

Choi's theorem tells us how to represent a generalized quantum evolution modeled by a completely positive trace preserving map, thereby generalizing Stone's theorem for close system evolution. 

\emph{Choi's theorem}. Every completely positive trace preserving map $\mathcal{L}(\hat{\rho}) = \hat{\rho}'$ is represented as (if and only if)
\begin{align}
	&\mathcal{L}(\hat{\rho}) = \sum_n \hat{V}_n \hat{\rho} \hat{V}_n^{\dagger}, \label{ChoiRep}
	\\
	&\sum_n \hat{V}_n^{\dagger} \hat{V}_n  = \hat{1}. \label{ChoiTP}
\end{align}

Let us demonstrate the theorem. If $\hat{\rho}$ is positive, it can be represented as 
\begin{align}
	\hat{\rho} = \sum_k | \psi_k \rangle \langle \psi_k |, 
\end{align}
where $| \psi_k \rangle$ are arbitrary vectors (not necessary a basis). Using Eq. (\ref{ChoiRep})
\begin{align}
	\mathcal{L}(\hat{\rho}) = \sum_{k, n} \hat{V}_n | \psi_k \rangle  \langle \psi_k | \hat{V}_n^{\dagger}  
		= \sum_{k, n} | \hat{V}_n \psi_k \rangle  \langle \psi_k \hat{V}_n |,
\end{align}
it is evident that $\mathcal{L}(\hat{\rho})$ is also positive operator. Let us verify the trace for an arbitrary $\hat{\rho}$
\begin{align}
	\Tr \mathcal{L}(\hat{\rho}) = \sum_{n} \Tr \left[ \hat{V}_n \hat{\rho} \hat{V}_n^{\dagger}  \right]
		= \sum_{n} \Tr \left[ \hat{V}_n^{\dagger} \hat{V}_n \hat{\rho} \right] \Longrightarrow
	\sum_n \hat{V}_n^{\dagger} \hat{V}_n  = \hat{1}.	
\end{align}
Hence, Eq. (\ref{ChoiTP}) is recovered.

\section{Lindblad master equation}\label{Sec:LindbladMasterEq}

Choi's theorem (Sec. \ref{Sec:ChoiThrm}) describes the mathematical structure for a linear map connecting the density matrices at the initial and final times. In this section, we would like to find a general form of master equation, specifying system's evolution in infinitesimal time increments ($\delta t$),  compatible with Choi's theorem.

Let us first recast the unitary evolution in terms of Choi's theorem (Sec. \ref{Sec:ChoiThrm}).
\begin{align}\label{UnitaryChoiRep1}
	\hat{\rho}(t + \delta t) =
    %&
    \left( 1 - \frac{i}{\hbar} \hat{H}\delta t \right) \hat{\rho}(t) \left( 1 - \frac{i}{\hbar} \hat{H}\delta t \right)^{\dagger} 
	%\\
	=
    %&
    \hat{\rho}(t) + \frac{i}{\hbar} [ \hat{\rho}(t), \hat{H} ] \delta t + O\left( \delta t^2 \right). 
		%\Longrightarrow \notag\\ 
\end{align}
This yields
\begin{align}        
	\frac{d \hat{\rho}(t)}{dt} &= \frac{i}{\hbar} [ \hat{\rho}(t), \hat{H} ], \label{vonNeumanRhoEq}
\end{align}
as well as
\begin{align}\label{UnitaryChoiRep2}
	\left( 1 - \frac{i}{\hbar} \hat{H}\delta t \right)^{\dagger} \left( 1 - \frac{i}{\hbar} \hat{H}\delta t \right) = 1 +  O\left( \delta t^2 \right) .
\end{align}
Note that condition (\ref{UnitaryChoiRep2}) [i.e., (\ref{ChoiTP})] must be obeyed up to the second-order terms with respect to the time increment $\delta t$ because time evolution (\ref{UnitaryChoiRep1}) yields corrections linear with respect to $\delta t$.  

According to Eqs. (\ref{UnitaryChoiRep1}) and (\ref{UnitaryChoiRep2}), equation of motion (\ref{vonNeumanRhoEq}) obeys Choi's theorem as $\delta t \to 0$. 

Let us generalize Eq. (\ref{UnitaryChoiRep1})
\begin{align}
	\hat{\rho}(t + \delta t) =& \delta t \sum_k \hat{A}_k \hat{\rho}(t) \hat{A}_k^{\dagger} 
		+ \left( 1 - \frac{i}{\hbar} \hat{H} \delta t + \sum_k \hat{B}_k \delta t \right) \hat{\rho}(t) 
			 \left( 1 - \frac{i}{\hbar} \hat{H} \delta t + \sum_k \hat{B}_k \delta t \right)^{\dagger}
	\\
	=& \hat{\rho}(t) + \frac{i}{\hbar} [ \hat{\rho}(t), \hat{H}] \delta t
	+ \sum_{k} \left[ \hat{A}_k \hat{\rho}(t) \hat{A}_k^{\dagger} + \hat{\rho}(t) \hat{B}_k^{\dagger} + \hat{B}_k \hat{\rho}(t)  \right] \delta t + O\left( \delta t^2 \right). %\Longrightarrow
\end{align}
Therefore,
\begin{align}
	\frac{d\hat{\rho}}{dt} =& \frac{i}{\hbar} [ \hat{\rho}(t), \hat{H}]
		+ \sum_{k} \left[ \hat{A}_k \hat{\rho}(t) \hat{A}_k^{\dagger} + \hat{\rho}(t) \hat{B}_k^{\dagger} + \hat{B}_k \hat{\rho}(t)  \right]. 
\end{align}
The operators are constrained using the trace-preserving condition (\ref{ChoiTP})
\begin{align}
	&\delta t \sum_k \hat{A}_k^{\dagger} \hat{A}_k +
		\left( 1 - \frac{i}{\hbar} \hat{H} \delta t + \sum_k \hat{B}_k \delta t \right)^{\dagger}
		\left( 1 - \frac{i}{\hbar} \hat{H} \delta t + \sum_k \hat{B}_k \delta t \right) 
	\notag\\
	&= 1 + \delta t \sum_k \left( \hat{A}_k^{\dagger} \hat{A}_k + \hat{B}_k + \hat{B}_k^{\dagger} \right) + O\left( \delta t^2 \right)
	\Longrightarrow 
	\hat{B}_k  = -\frac{1}{2} \hat{A}_k^{\dagger} \hat{A}_k.
\end{align}

%\fixme{\bf(SHOW THAT ANTI HERMITIAN PART of B modifies the hamiltonian)}

Finally, the most general for of master equation describing completely positive trace preserving dynamics reads
\BoxedEquation{\label{LindbladMasterEq}
	\frac{d\hat{\rho}(t)}{dt} =& \frac{i}{\hbar} [ \hat{\rho}(t), \hat{H}]
		+ \sum_{k} \left( \hat{A}_k \hat{\rho}(t) \hat{A}_k^{\dagger} 
			-\frac{1}{2} \hat{\rho}(t) \hat{A}_k^{\dagger} \hat{A}_k -\frac{1}{2} \hat{A}_k^{\dagger} \hat{A}_k \hat{\rho}(t)  \right)
		\notag\\
		=& \frac{i}{\hbar} [ \hat{\rho}(t), \hat{H}] + \frac{1}{2} \sum_{k} \left( 
			\left[ \hat{A}_k \hat{\rho}(t), \hat{A}_k^{\dagger} \right] + \left[ \hat{A}_k, \hat{\rho}(t) \hat{A}_k^{\dagger}\right]
		\right).
}
If all $\hat{A}_k$ does not depend on time, Eq. (\ref{LindbladMasterEq}) is known as the \emph{Lindblad master equation}, descrinbing Markovian dynamics. However, any non-Markovian dynamics can be described within Eq. (\ref{LindbladMasterEq}) by making $\hat{A}_k = \hat{A}_k(t)$ time dependent. This form is known as \emph{time-convolutionless master equation}.

A special case of the Lindblad master equation is worth noting
\BoxedEquation{\label{LindbladMasterSelfAdjointEq}
	\hat{A}_k = \hat{A}_k^{\dagger} \quad \Longrightarrow \quad
	\frac{d\hat{\rho}(t)}{dt} = \frac{i}{\hbar} [ \hat{\rho}(t), \hat{H}] 
		-\frac{1}{2} \sum_{k} [ \hat{A}_k, [\hat{A}_k, \hat{\rho}(t)]].
}

\section{Physical interpretation of Lindblad and classical Markov jumps}

First, let us rewrite Eq. (\ref{LindbladMasterEq}) in the following form
\begin{align}\label{LindbladAlickiMasterEq}
	\frac{d\hat{\rho}(t)}{dt} = \frac{i}{\hbar} [ \hat{\rho}(t), \hat{H}]
		+ & \sum_{k > j}  \gamma_{k \to j} \left( \hat{A}_{k \to j} \hat{\rho}(t) \hat{A}_{k \to j}^{\dagger} 
			-\frac{1}{2} \hat{\rho}(t) \hat{A}_{k \to j}^{\dagger} \hat{A}_{k \to j} -\frac{1}{2} \hat{A}_{k \to j}^{\dagger} \hat{A}_{k \to j} \hat{\rho}(t)\right)
		\notag\\
		+ & \sum_{k \geq j}  \gamma_{j \to k} \left( \hat{A}_{k \to j}^{\dagger} \hat{\rho}(t) \hat{A}_{k \to j} 
			-\frac{1}{2} \hat{\rho}(t) \hat{A}_{k \to j} \hat{A}_{k \to j}^{\dagger} -\frac{1}{2} \hat{A}_{k \to j} \hat{A}_{k \to j}^{\dagger} \hat{\rho}(t)\right),
	\notag\\
	\gamma_{k \to j} \geq 0, \quad \gamma_{j \to k} \geq 0.
\end{align}

Let $| n \rangle$ be eigenstate of the Hamiltonian $\hat{H}$ and the density matrix $\hat{\rho}$ be diagonal in the basis of $| n \rangle$
\begin{align}\label{DiagonalHamiltDensityMatrix}
	\hat{H}| n \rangle  = E_n | n \rangle,
	\qquad
	\langle n | \hat{\rho}(t)| n \rangle =  p_{n}(t),
	\qquad
	\langle n' | n \rangle = \delta_{n', n},
\end{align}
where $E_n$ denotes energy levels, and $p_n(t)$ denotes the population of the state $| n \rangle$,
\begin{align}\label{PopulationProbabilitiesCond}
	0 \leq p_n(t) \leq 1,
	\qquad
	\sum_n p_n(t) = 1.
\end{align}
Note that $\hat{\rho}(t)$ and $\hat{H}$ share common eigenvectors, hence they commute for all times $t$.

Furthermore, we assume
\begin{align}\label{LindbladJumpOperators}
	\hat{A}_{k \to j} = | k \rangle \langle j |,
	\qquad
	\hat{A}_{k \to j}^{\dagger} = | j \rangle \langle k |.
\end{align}

Sandwiching Eq. (\ref{LindbladAlickiMasterEq}) between $\langle n |$ and $| n \rangle$ with help of Eqs. (\ref{DiagonalHamiltDensityMatrix})  and (\ref{LindbladJumpOperators}), we obtain
\begin{align}
	\frac{d}{dt} \langle n | \hat{\rho}(t) | n \rangle &= \frac{d p_n(t)}{dt} =
	\sum_{k > j}  \gamma_{k \to j} \langle n | \left( | k \rangle \langle j | \hat{\rho}(t) | j \rangle \langle k | 
			-\frac{1}{2} \hat{\rho}(t) | j \rangle \langle k | k \rangle \langle j | -\frac{1}{2} | j \rangle \langle k | k \rangle \langle j | \hat{\rho}(t)\right) | n \rangle
		\notag\\
		& + \sum_{k \geq j}  \gamma_{j \to k} \langle n | \left( | j \rangle \langle k | \hat{\rho}(t) | k \rangle \langle j | 
			-\frac{1}{2} \hat{\rho}(t) | k \rangle \langle j | j \rangle \langle k | -\frac{1}{2} | k \rangle \langle j | j \rangle \langle k | \hat{\rho}(t) \right) | n \rangle,
		\notag\\
		=&
	\sum_{k > j}  \gamma_{k \to j} \left( \delta_{n,k} p_j(t)  -\frac{1}{2} p_j(t) \delta_{n,j}  -\frac{1}{2} p_n(t) \delta_{n,j} \right) 
		\notag\\
		& + \sum_{k \geq j}  \gamma_{j \to k} \left(  \delta_{n,j}  p_k(t)  -\frac{1}{2} p_k(t) \delta_{n,k} -\frac{1}{2}  p_n(t) \delta_{n,k}  \right),
		\notag\\
		=&
	\sum_{n > j}  \gamma_{n \to j}  p_j(t)  -  p_n(t)  \sum_{k > n}  \gamma_{k \to n} 
		 + \sum_{k \geq n}  \gamma_{n \to k}  p_k(t)  - p_n(t) \sum_{n \geq j}  \gamma_{j \to n}.
\end{align}
Finally, we recover the master equation for \emph{continuous-time Markov process} (in classical probability) also known as the \emph{Pauli master equation}
\begin{align}\label{MarkovJumpMasterEq}
	\frac{d p_n(t)}{dt} = \sum_{j} \left[ \gamma_{n \to j}  p_j(t)  -  \gamma_{j \to n} p_n(t) \right].
\end{align}

Thus, Lindblad equation (\ref{LindbladAlickiMasterEq}) models transitions between energy levels  of the Hamiltonian, as shown in Fig. \ref{Fig:LindbladInterpretation}. $ \gamma_{n \to j}$ denotes the transition rate from the $n$ to $j$ levels.

%%%%%%%%%%%%%%%%%%%%%%%%%%%%%%%%%%%%%%%%%%%%%%%%%%%%
\begin{figure}
    \centering
	\includegraphics[width=0.3\hsize]{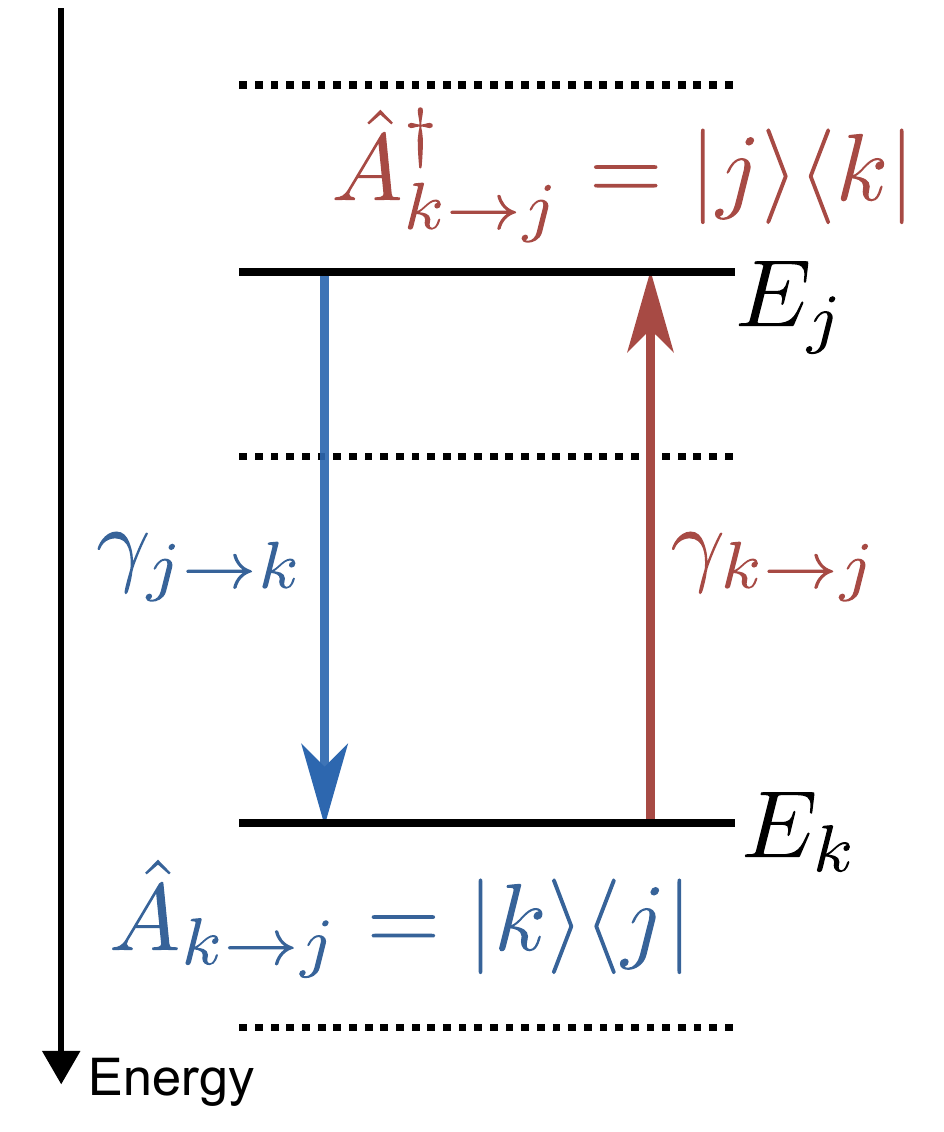}
	\caption{Interpretation of the Lindblad equation (\ref{LindbladAlickiMasterEq}).}\label{Fig:LindbladInterpretation}
\end{figure}
%%%%%%%%%%%%%%%%%%%%%%%%%%%%%%%%%%%%%%%%%%%%%%%%%%%%

\section{Lindblad dissipator and transition rates}

The dissipator in Eqs.~(\ref{LindbladAlickiMasterEq}) and (\ref{LindbladJumpOperators}) can be written as 
\begin{align}
	\mathcal{L}[\odot] = \sum_{j,k} \gamma_{k \to j} \left( 
		| k \rangle\langle j | \odot | j \rangle\langle k |
		-\frac{1}{2} \odot | j \rangle\langle j |
		-\frac{1}{2}  | j \rangle\langle j | \odot
	\right). 
\end{align}
We note that
\begin{align}\label{GammaLindbladAlicki}
	\gamma_{n \to l} = \langle n | \mathcal{L}\left( | l \rangle\langle l | \right) | n \rangle, \qquad l \neq n.
\end{align}
To introduce the transition rate for a general dissipator of the form
\begin{align}\label{GenericLindbladDissipator}
 	\mathcal{D}[\odot] = \hat{A} \odot \hat{A}^{\dagger} - \frac{1}{2} \odot \hat{A}^{\dagger}\hat{A} - \frac{1}{2} \hat{A}^{\dagger}\hat{A} \odot,
\end{align}
we calculate
\begin{align}\label{GenericLindbladLikeGamma}
 	\langle n | \mathcal{D}\left( | l \rangle\langle l | \right) | n \rangle = \left| \langle n | \hat{A} | l \rangle \right|^2, \qquad l \neq n.
\end{align}
Comparison of Eqs.~(\ref{GammaLindbladAlicki}) and (\ref{GenericLindbladLikeGamma}) suggests 
\BoxedEquation{\label{GammaAndA_FermiGoldenRule}
 	\gamma_{n \to l} = \left| \langle n | \hat{A} | l \rangle \right|^2, \qquad l \neq n,
}
to interpret $\hat{A}$ in the generic dissipator (\ref{GenericLindbladDissipator}) as a coupling operator inducing the transition $n \to l$. Note that Eq. (\ref{GammaAndA_FermiGoldenRule}) resembles \emph{the Fermi golden rule.}

\section{Random unitary channel}

The \emph{random unitary channel} (see, e.g., Ref. \cite{audenaert2008random}) is a mapping from density matrices to density matrices,
\begin{align}\label{EqDeffRandomUnitary}
	\mathcal{L}(\hat{\rho}) = \sum_k p_k \hat{U}_k \hat{\rho} \hat{U}_k^{\dagger},
	\qquad
	\hat{U}_k^{\dagger} \hat{U}_k = \hat{U}_k \hat{U}_k^{\dagger}  =1, \quad
	\sum_k p_k = 1, \quad 0 \leq p_k < 1,
\end{align}
where $\hat{U}_k$ are unitary operators and $p_k$ has the meaning of probability, such that the unitary transformation  $\hat{U}_k$ is applied onto the state $\hat{\rho}$ with probability $p_k$. Notice that $p_k$ cannot be exactly equal to unity because in that case the map $\mathcal{L}$ reduces to a unitary transformation. By construction, the map \eqref{EqDeffRandomUnitary} satisfies Choi's theorem \ref{Sec:ChoiThrm}. Because of the unitarity, $\hat{U}_k = e^{-i\hat{G}_k}$ for self-adjoint $\hat{G}_k$. The Hausdorff expansion \eqref{EqBCH} %{HausdorffExpansionEq} 
leads to
\begin{align}
	\mathcal{L}(\hat{\rho}) = \sum_k p_k e^{-i\hat{G}_k} \hat{\rho} e^{i\hat{G}_k} 
		=  \hat{\rho} + \sum_k p_k \sum_{n=1}^{\infty} \frac{(-i)^n}{n!} \underbrace{[ \hat{G}_k, [ \hat{G}_k, [\cdots [ \hat{G}_k}_{\mbox{$n$-times}}, \hat{\rho}] \cdots ].
\end{align}

\begin{align}
	\hat{\rho}(t + \delta t)  =  \hat{\rho}(t) -i \sum_k p_k [ \hat{G}_k, \hat{\rho}(t)]
		+ \sum_k p_k  \frac{(-i)^2}{2} [ \hat{G}_k, [ \hat{G}_k,  \hat{\rho}(t)]] + \left( \mbox{3rd order term w.r.t. } \hat{G}_k \right).
\end{align}

\section[The environmental coupling]{The environmental coupling introduces not only dissipation, but also a Hamiltonian correction}

We will show that the inclusion of the Lindblad dissipator may modify the unitary dynamics in addition to inducing dissipative non-unitary evolution. This is a generalization of the following physical example: the vacuum non-only causes a line broadening via spontaneous emission of atomic spectral lines (modeled as dissipative dynamics), but also the Lamb shift (modeled a small correction to the Hamiltonian).

Assume a closed quantum system with the Hamiltonian $\hat{H}$ coupled to the environment modeled by the following Lindblad equation
\begin{align}
	\frac{d \hat{\rho}}{dt} = \mathcal{G} [\hat{\rho}]  +  \mathcal{D}[\hat{\rho}], 
	\qquad
	 \mathcal{G}[ \odot] = \frac{i}{\hbar} [  \odot, \hat{H} ], \qquad
	 \mathcal{D}[\odot] = \hat{A} \odot \hat{A}^{\dagger} - \frac{1}{2} \odot \hat{A}^{\dagger}\hat{A} - \frac{1}{2} \hat{A}^{\dagger}\hat{A} \odot. 
\end{align}
Note that the generator of the unitary $\mathcal{G}$ and Lindblad $\mathcal{D}$ dynamics are linear operators acting on the Hilbert space of operators endowed with the Hilbert-Schmidt scalar product $\langle \hat{A} | \hat{B} \rangle = \Tr( \hat{A}^{\dagger} \hat{B} )$.

In this terminology, $\mathcal{G}$ is an anti self-adjoint linear operator, i.e., $\mathcal{G} = -\mathcal{G}^{\dagger}$. Given the Hilbert-Schmidt scalar product, the adjoint $\mathcal{G}^{\dagger}$ operator  to $\mathcal{G}$ is defined as 
\begin{align}
	\Tr \left(\hat{\xi} \mathcal{G}[\hat{\eta}]\right) = \Tr \left(\mathcal{G}^{\dagger}[\hat{\xi}] \hat{\eta}\right),
	\qquad
	\forall \hat{\xi}, \, \hat{\eta}.
\end{align}
Using the cyclic property of trace $\Tr(\hat{A}\hat{B}) = \Tr(\hat{B}\hat{A})$, we obtain
\begin{align*}
	\Tr \left(\hat{\xi} \mathcal{G}[\hat{\eta}]\right) = \frac{i}{\hbar}  \Tr \left(\hat{\xi} \hat{\eta} \hat{H} - \hat{\xi} \hat{H} \hat{\eta} \right) = \frac{i}{\hbar}  \Tr \left(\hat{H} \hat{\xi} \hat{\eta}  - \hat{\xi} \hat{H} \hat{\eta} \right) = \frac{i}{\hbar}  \Tr \left([\hat{H}, \hat{\xi}] \hat{\eta} \right),
\end{align*}
therefore,
\begin{align}
    \mathcal{G}^{\dagger} [ \odot] = \frac{i}{\hbar} [ \hat{H},  \odot ] = -\mathcal{G} [ \odot].
\end{align}

Now, we partition the dissipator $\mathcal{D}$ as 
\begin{align}
	\mathcal{D} = \underbrace{\left( \mathcal{D} + \mathcal{D}^{\dagger} \right)/2}_{\mbox{self-adjoint part:} \atop \mbox{describes pure dissipation}} 
	+  \underbrace{\left( \mathcal{D} - \mathcal{D}^{\dagger} \right)/2.}_{\mbox{anti self-adjoint part:} \atop \mbox{correction to unitary dynamics}}
\end{align}
\begin{align}
	&\Tr \left(\hat{\xi} \mathcal{D}[\hat{\eta}]\right)
	= \Tr \left( \hat{\xi} \hat{A} \hat{\eta} \hat{A}^{\dagger} - \frac{1}{2} \hat{\xi} \hat{\eta} \hat{A}^{\dagger}\hat{A} - \frac{1}{2} \hat{\xi} \hat{A}^{\dagger}\hat{A} \hat{\eta} \right)
	= \Tr \left( \hat{A}^{\dagger} \hat{\xi} \hat{A} \hat{\eta} - \frac{1}{2} \hat{A}^{\dagger}\hat{A} \hat{\xi} \hat{\eta}  - \frac{1}{2} \hat{\xi} \hat{A}^{\dagger}\hat{A} \hat{\eta} \right)
	\notag\\
    &\qquad
    = \Tr \left(\mathcal{D}^{\dagger}[\hat{\xi}] \hat{\eta}\right) \Longrightarrow \notag\\
	& \mathcal{D}^{\dagger} [\odot] = \hat{A}^{\dagger} \odot \hat{A} - \frac{1}{2} \odot \hat{A}^{\dagger}\hat{A} - \frac{1}{2} \hat{A}^{\dagger}\hat{A} \odot.
\end{align}
Thus, we extract pure dissipative correction
\begin{align}
	\frac{1}{2} \left( \mathcal{D} + \mathcal{D}^{\dagger} \right) = 
	\frac{1}{2} \left(
		\hat{A} \odot \hat{A}^{\dagger}  + \hat{A}^{\dagger} \odot \hat{A} 
		-\odot \hat{A}^{\dagger}\hat{A} - \hat{A}^{\dagger}\hat{A} \odot 
	\right),
\end{align}
and purely unitary correction inducing an effective Hamiltonian
 \begin{align}
	\frac{1}{2} \left( \mathcal{D} - \mathcal{D}^{\dagger} \right) = 
	\frac{1}{2} \left( \hat{A} \odot \hat{A}^{\dagger}  - \hat{A}^{\dagger} \odot \hat{A} \right).
\end{align}
Note that if $\hat{A}$ is self-adjoint, then the environment offers no correction to the unitary evolution.

\section{When Lindblad equation converges to the Boltzmann Gibbs state}

System coupled to a bath should converge the Gibbs state, whose density matrix reads 
\begin{align}
	\hat{\rho}_{\beta} = \frac{1}{Z} e^{-\beta \hat{H}},
	\quad
	\beta = \frac{1}{kT} 
	\quad \Longrightarrow \quad
	\hat{\rho}_{\beta} | n \rangle = \frac{1}{Z} e^{-\beta E_n} | n \rangle.
\end{align}
Thus, we want the Lindblad equation (\ref{LindbladAlickiMasterEq}) to have the Gibbs state as a stationary state. When does it happen? From Eq. (\ref{MarkovJumpMasterEq}), we obtain
\begin{align}
		\frac{d p_n}{dt} = \sum_{j} \left[ \gamma_{n \to j}  \frac{1}{Z} e^{-\beta E_j}  -  \gamma_{j \to n} \frac{1}{Z} e^{-\beta E_n} \right] = 0.
\end{align}
For this equation to be true for all $\gamma_{n \to j} \geq 0$, 
\begin{align}
		 \gamma_{n \to j} /  \gamma_{j \to n}  = e^{-\beta (E_n - E_j)}.
\end{align}

\section{Lindblad equation converges to the Fermi-Dirac distribution}

The desired Fermi-Dirac stationary state reads
\begin{align}
	\hat{\rho}_{\beta} = \frac{1}{Z} \frac{1}{e^{\beta (\hat{H} - \mu)} + 1} \Longrightarrow p_n = \frac{1}{Z} \frac{1}{e^{\beta (E_n - \mu)} + 1}.
\end{align}
For the Lindblad equation (\ref{LindbladAlickiMasterEq})  to yield the Fermi-Dirac steady state,
\begin{align}
	\frac{d p_n}{dt} = \sum_{j} \left[ \gamma_{n \to j}  \frac{1}{Z} \frac{1}{e^{\beta (E_j - \mu)} + 1} -  \gamma_{j \to n} \frac{1}{Z} \frac{1}{e^{\beta (E_n - \mu)} + 1} \right] = 0.
\end{align}
Therefore,
\begin{align}
	\frac{\gamma_{n \to j}}{\gamma_{j \to n}} = \frac{e^{\beta(E_j - \mu)} + 1}{e^{\beta(E_n-\mu)} + 1}.
\end{align}

\section{A unitary evolution of the density matrix}

In this section, we shall develop a numerical algorithm to solve the von Neumann equation:
\begin{align}
	\frac{d \hat{\rho}(t)}{dt} =\frac{i}{\hbar} [ \hat{\rho}(t), \hat{H}(t) ], \qquad \hat{H}(t) = K(t, \hat{p}) + V(t, \hat{x}). 
\end{align}
Let us show that the following is a solution by direct substitution
\begin{align}
	\hat{\rho}(t) = \hat{\mathcal{T}} \exp\left[ -\frac{i}{\hbar} \int_{t'}^{t} \hat{H}(\tau) d\tau \right]  \hat{\rho}(t') \hat{\mathcal{T}} \exp\left[ \frac{i}{\hbar} \int_{t'}^{t} \hat{H}(\tau) d\tau \right].
\end{align}
Using Eq. (\ref{EqMotionForArbitraryG}), we get
\begin{align}
	i \hbar \frac{d}{dt} \hat{\rho}(t) =& \left\{ i\hbar\frac{\partial}{\partial t} \hat{\mathcal{T}} \exp\left[ -\frac{i}{\hbar} \int_{t'}^{t} \hat{H}(\tau) d\tau \right] \right\}  \hat{\rho}(t') \hat{\mathcal{T}} \exp\left[ \frac{i}{\hbar} \int_{t'}^{t} \hat{H}(\tau) d\tau \right] \notag\\
		& + \hat{\mathcal{T}} \exp\left[ -\frac{i}{\hbar} \int_{t'}^{t} \hat{H}(\tau) d\tau \right] \hat{\rho}(t')  \left\{ -i\hbar\frac{\partial}{\partial t} \hat{\mathcal{T}} \exp\left[ -\frac{i}{\hbar} \int_{t'}^{t} \hat{H}(\tau) d\tau \right] \right\}^{\dagger} \notag\\
		=& \hat{H}(t) \hat{\mathcal{T}} \exp\left[ -\frac{i}{\hbar} \int_{t'}^{t} \hat{H}(\tau) d\tau \right]  \hat{\rho}(t') \hat{\mathcal{T}} \exp\left[ \frac{i}{\hbar} \int_{t'}^{t} \hat{H}(\tau) d\tau \right] \notag\\
		& + \hat{\mathcal{T}} \exp\left[ -\frac{i}{\hbar} \int_{t'}^{t} \hat{H}(\tau) d\tau \right] \hat{\rho}(t')  \left\{ -\hat{H}(t) \hat{\mathcal{T}} \exp\left[ -\frac{i}{\hbar} \int_{t'}^{t} \hat{H}(\tau) d\tau \right] \right\}^{\dagger} \notag\\
		=& [\hat{H}(t), \hat{\rho}(t)].
\end{align}
Let us derive the propagator
\begin{align}
	\bra{x} \hat{\rho}(t& + \delta t ) \ket{x'} = \bra{x} \hat{\mathcal{T}} \exp\left[ -\frac{i}{\hbar} \int_{t}^{t + \delta t} \hat{H}(\tau) d\tau \right]  \hat{\rho}(t) \hat{\mathcal{T}} \exp\left[ \frac{i}{\hbar} \int_{t}^{t + \delta t} \hat{H}(\tau) d\tau \right] \ket{x'} \notag\\
	=& \bra{x} \exp\left[ -\frac{i}{\hbar} \int_{t}^{t + \delta t} \hat{H}(\tau) d\tau \right]  \hat{\rho}(t) \exp\left[ \frac{i}{\hbar} \int_{t}^{t + \delta t} \hat{H}(\tau) d\tau \right] \ket{x'} + O\left( \delta t^3 \right)
	~~\mbox{[using Eq. (\ref{EqTimeOrderingDropping})]} \notag\\
	=& \bra{x} \exp\left[ -\frac{i \delta t }{\hbar} \hat{H}\left(t + \frac{\delta t}{2}\right) \right]  \hat{\rho}(t) \exp\left[ \frac{i \delta t }{\hbar} \hat{H}\left(t + \frac{\delta t}{2}\right) \right] \ket{x'} + O\left( \delta t^3 \right)
	~~\mbox{[using Eq. (\ref{EqTrapezoidalIntegralRule})]} \notag\\
	=& \bra{x} 
		e^{ -\frac{i \delta t }{2\hbar} V\left(t + \frac{\delta t}{2}, \hat{x}\right) } 
		e^{ -\frac{i \delta t }{\hbar} K\left(t + \frac{\delta t}{2}, \hat{p}\right) }
		e^{ -\frac{i \delta t }{2\hbar} V\left(t + \frac{\delta t}{2}, \hat{x}\right) } 
		\hat{\rho}(t) \notag\\
	& \times	e^{ \frac{i \delta t }{2\hbar} V\left(t + \frac{\delta t}{2}, \hat{x}\right) } 
		e^{ \frac{i \delta t }{\hbar} K\left(t + \frac{\delta t}{2}, \hat{p}\right) }
		e^{ \frac{i \delta t }{2\hbar} V\left(t + \frac{\delta t}{2}, \hat{x}\right) } 
		 \ket{x'} + O\left( \delta t^3 \right)
	~~\mbox{[using Eq. (\ref{SecondOrderSplitting})]} \notag\\
	=&  e^{ \frac{i \delta t }{2\hbar} \left[ V\left(t + \frac{\delta t}{2}, x' \right) -V\left(t + \frac{\delta t}{2}, x\right) \right]} 
		\int dx'' dx''' dp dp'
		\bra{x}
		e^{ -\frac{i \delta t }{\hbar} K\left(t + \frac{\delta t}{2}, \hat{p}\right) } \ket{p}\bra{p}
		e^{ -\frac{i \delta t }{2\hbar} V\left(t + \frac{\delta t}{2}, \hat{x}\right) } 
		 \notag\\
	& \times	\ket{x''}\bra{x''} \hat{\rho}(t) \ket{x'''} \bra{x'''} e^{ \frac{i \delta t }{2\hbar} V\left(t + \frac{\delta t}{2}, \hat{x}\right) } 
		\ket{p'}\bra{p'} e^{ \frac{i \delta t }{\hbar} K\left(t + \frac{\delta t}{2}, \hat{p}\right) } 
		 \ket{x'} + O\left( \delta t^3 \right)
	\notag\\
	=&  e^{ \frac{i \delta t }{2\hbar} \left[ V\left(t + \frac{\delta t}{2}, x' \right) -V\left(t + \frac{\delta t}{2}, x\right) \right]} 
		\int dx'' dx''' dp dp'
		\bra{x} p \rangle \langle p' \ket{x'}
		e^{ \frac{i \delta t }{\hbar} \left[ K\left(t + \frac{\delta t}{2}, p'\right) - K\left(t + \frac{\delta t}{2}, p\right) \right]} 
		 \notag\\
	& \times
		\bra{p} x'' \rangle \langle x''' \ket{p'}
		e^{ \frac{i \delta t }{2\hbar} \left[V\left(t + \frac{\delta t}{2}, x'''\right) - V\left(t + \frac{\delta t}{2}, x''\right)\right] } 
		\bra{x''} \hat{\rho}(t) \ket{x'''}   + O\left( \delta t^3 \right).
\end{align} 
Using Eq. (\ref{EqContiniousFourierTransform}), we finally get the second-order unitary propagator for the density matrix
\BoxedEquation{\label{UnitaryPropDensityMatrxiEq}
	\bra{x} \hat{\rho}(t + \delta t ) \ket{x'} =& 
		e^{ \frac{i \delta t }{2\hbar} \left[ V\left(t + \frac{\delta t}{2}, x' \right) -V\left(t + \frac{\delta t}{2}, x\right) \right]} 
		\mathcal{F}_{p' \to x'} \mathcal{F}_{p \to x}^{-1} \left\{
		e^{ \frac{i \delta t }{\hbar} \left[ K\left(t + \frac{\delta t}{2}, p'\right) - K\left(t + \frac{\delta t}{2}, p\right) \right]} \times  \right. \notag\\
		& \left. \mathcal{F}_{x''' \to p'}^{-1} \mathcal{F}_{x'' \to p} \left\{  
			e^{ \frac{i \delta t }{2\hbar} \left[V\left(t + \frac{\delta t}{2}, x'''\right) - V\left(t + \frac{\delta t}{2}, x''\right)\right] } 
			\bra{x''} \hat{\rho}(t) \ket{x'''} 
		\right\}\right\} + O\left( \delta t^3 \right).
}

\section{A position dependent dissipator}

In this section, we develop a propagator for the following equation
\begin{align}\label{LindbladWithADependentOnXEq}
	\frac{d\hat{\rho}(t)}{dt} =& \frac{i}{\hbar} \left[ \hat{\rho}(t), K(t, \hat{p}) + V(t, \hat{x})\right]
		+ A(\hat{x}) \hat{\rho}(t) A(\hat{x})^{\dagger} 
		-\frac{1}{2} \hat{\rho}(t) A(\hat{x})^{\dagger} A(\hat{x}) -\frac{1}{2} A(\hat{x})^{\dagger} A(\hat{x}) \hat{\rho}(t). 
\end{align}
Sandwiching this equation between $\bra{x}$ and $\ket{x'}$, one gets
\begin{align}
	\frac{d}{dt} \bra{x}& \hat{\rho}(t) \ket{x'} = 
		\frac{i}{\hbar}  \left[ \bra{x} \hat{\rho}(t) K(t, \hat{p}) \ket{x'}  - \bra{x} K(t, \hat{p}) \hat{\rho}(t) \ket{x'} \right] + F(t; x,x') \bra{x} \hat{\rho}(t) \ket{x'} \notag\\
		&= \left[ \frac{i}{\hbar} K\left(t, i\hbar\frac{\partial}{\partial x'} \right) - \frac{i}{\hbar} K\left(t, -i\hbar\frac{\partial}{\partial x} \right) + F(t; x,x') \right] \bra{x} \hat{\rho}(t) \ket{x'}, 
			~~ \mbox{[using Eq. (\ref{EqSummaryCommutatorSection})]}  \\
	F(t; x,x') &= \frac{i}{\hbar} \left[V(t, x') - V(t,x) \right] + A(x) A(x')^* - \frac{1}{2} A(x')^* A(x') - \frac{1}{2} A(x)^* A(x).
\end{align}
Therefore, now Eq. (\ref{UnitaryPropDensityMatrxiEq}) can be readily generalized to propagate Eq. (\ref{LindbladWithADependentOnXEq}):
\BoxedEquation{
	\bra{x} \hat{\rho}(t + \delta t ) \ket{x'} =& 
		e^{ \frac{\delta t }{2} F\left(t + \frac{\delta t}{2}; x, x' \right) } 
		\mathcal{F}_{p' \to x'} \mathcal{F}_{p \to x}^{-1} \left\{
		e^{ \frac{i \delta t }{\hbar} \left[ K\left(t + \frac{\delta t}{2}, p'\right) - K\left(t + \frac{\delta t}{2}, p\right) \right]} \times  \right. \notag\\
		& \left. \mathcal{F}_{x''' \to p'}^{-1} \mathcal{F}_{x'' \to p} \left\{  
			e^{ \frac{\delta t }{2} F\left(t + \frac{\delta t}{2}; x'', x''' \right) }
			\bra{x''} \hat{\rho}(t) \ket{x'''} 
		\right\}\right\} + O\left( \delta t^3 \right).
}
Finally, using Eq. (\ref{EqCFTviaFFT}), one obtains 
\begin{align}
	\bra{x_l} \hat{\rho}(t + \delta t ) \ket{x_{l'}} \approx& 
		e^{ \frac{\delta t }{2} F\left(t + \frac{\delta t}{2}; x_l, x_{l'} \right) } (-)^{l + l'}
		FFT_{k' \to l'} iFFT_{k \to l} \left\{
		(-)^{k + k'}  e^{ \frac{i \delta t }{\hbar} \left[ K\left(t + \frac{\delta t}{2}, p_{k'}\right) - K\left(t + \frac{\delta t}{2}, p_k\right) \right]} (-)^{k + k'}  \right. \notag
        \\
		& \left. \times iFFT_{l''' \to k'} FFT_{l'' \to k} \left\{  
			(-)^{l'' + l'''} e^{ \frac{\delta t }{2} F\left(t + \frac{\delta t}{2}; x_{l''}, x_{l'''} \right) }
			\bra{x_{l''}} \hat{\rho}(t) \ket{x_{l'''}} 
		\right\}\right\} \notag
        \\
		=& e^{ \frac{\delta t }{2} F\left(t + \frac{\delta t}{2}; x_l, x_{l'} \right) } (-)^{l + l'}
		FFT_{k' \to l'} iFFT_{k \to l} \left\{
		e^{ \frac{i \delta t }{\hbar} \left[ K\left(t + \frac{\delta t}{2}, p_{k'}\right) - K\left(t + \frac{\delta t}{2}, p_k\right) \right]}  \right. \notag
        \\
		& \left.\times  iFFT_{l''' \to k'} FFT_{l'' \to k} \left\{  
			(-)^{l'' + l'''} e^{ \frac{\delta t }{2} F\left(t + \frac{\delta t}{2}; x_{l''}, x_{l'''} \right) }
			\bra{x_{l''}} \hat{\rho}(t) \ket{x_{l'''}}
		\right\}\right\}. 
\end{align}

\section{A random collision model}

\begin{align}\label{RandomCollisionRhoMasterEq}
	\frac{d\hat{\rho}}{dt} = \frac{1}{i\hbar} [\hat{H}, \hat{\rho}]  + \gamma(\hat{\rho}_{\beta} - \hat{\rho}). 
\end{align}

This model is a special case of Lindblad equation (\ref{LindbladAlickiMasterEq}) when
\begin{align}
	\gamma_{n \to j}  = \gamma e^{-\beta E_n} / Z.
\end{align}
Indeed, substituting this identity into Eq. (\ref{MarkovJumpMasterEq}) and using Eq. (\ref{PopulationProbabilitiesCond}), we get
\begin{align}
	\frac{d p_n(t)}{dt} = \sum_{j} \left[ \gamma_{n \to j}  p_j(t)  -  \gamma_{j \to n} p_n(t) \right]
		=  \sum_{j} \frac{\gamma}{Z} e^{-\beta E_n} p_j(t)  - \sum_{j} \frac{\gamma}{Z} e^{-\beta E_j} p_n(t)
		\notag\\
        = \gamma\left(  \frac{1}{Z} e^{-\beta E_n} - p_n(t) \right).
\end{align}
Thus, we recovered Eq. (\ref{RandomCollisionRhoMasterEq}).

\section{Time propagation of a random collision model}

Consider a generalized equation of motion (\ref{RandomCollisionRhoMasterEq}) with the time dependent hamiltonian
\begin{align}\label{EqTimeDepRandomCollisionMod}
	\frac{d\hat{\rho}}{dt} = \frac{1}{i\hbar} [\hat{H}(t), \hat{\rho}]  + \gamma(\hat{\rho}_{\beta} - \hat{\rho})
	= \left[ \hat{\mathcal{G}}(t) + \hat{\mathcal{D}} \right] \hat{\rho},
	\qquad 
	\hat{\mathcal{G}}(t) \odot = \frac{1}{i\hbar} [\hat{H}(t), \odot],
	\quad
	\hat{\mathcal{D}}  \odot = \gamma(\hat{\rho}_{\beta} - \odot).
\end{align}
Since Eq. (\ref{EqTimeDepRandomCollisionMod}) resembles Eq. (\ref{EqMotionForArbitraryG}), then the solution for the random collision  model reads according to Eq. (\ref{EqTExpGeneric})
\begin{align}
	\hat{\rho}(t + \delta t) &= \hat{\mathcal{T}} \exp\left[  \int_{t}^{t + \delta t} \left( \hat{\mathcal{G}}(\tau) + \hat{\mathcal{D}} \right) d\tau \right] \hat{\rho}(t)  \notag\\
		&= \exp\left[  \int_{t}^{t + \delta t} \hat{\mathcal{G}}(\tau)d\tau + \hat{\mathcal{D}} \delta t \right] \hat{\rho}(t)  + O\left( \delta t^3 \right) \qquad\qquad \mbox{[using Eq. (\ref{EqTimeOrderingDropping})]} \notag\\
		&= \exp\left[  \hat{\mathcal{G}}\left(t + \frac{\delta t}{2}  \right) \delta t + \hat{\mathcal{D}} \delta t \right] \hat{\rho}(t)  + O\left( \delta t^3 \right) \qquad\qquad \mbox{[using Eq. (\ref{EqTrapezoidalIntegralRule})]} \notag\\
		&=  \underbrace{e^{\hat{\mathcal{G}}\left(t + \frac{\delta t}{2}  \right) \frac{\delta t}{2}}}_{\mbox{unitary} \atop \mbox{evolution}}
			\underbrace{e^{\hat{\mathcal{D}} \delta t}}_{\mbox{dissipative} \atop \mbox{evolution}}
			\underbrace{e^{\hat{\mathcal{G}}\left(t + \frac{\delta t}{2}  \right) \frac{\delta t}{2}}}_{\mbox{unitary} \atop \mbox{evolution}} \, \hat{\rho}(t)  + O\left( \delta t^3 \right).  
			\qquad\qquad \mbox{[using Eq. (\ref{SecondOrderSplitting})]}		
\end{align}

\section{Wave-function Monte Carlo with Poisson noise}

The solution $\hat{\rho}(t)$ to the master equation
\begin{align}
	\frac{d}{dt} \hat{\rho}(t) =& -\frac{i}{\hbar} \left[ \hat{H}(t), \hat{\rho}(t) \right] 
		+ \sum_{k} \left( \hat{A}_k \hat{\rho}(t) \hat{A}_k^{\dagger} 
			-\frac{1}{2} \hat{\rho}(t) \hat{A}_k^{\dagger} \hat{A}_k -\frac{1}{2} \hat{A}_k^{\dagger} \hat{A}_k \hat{\rho}(t)  \right) \notag\\
		& + \sum_{k} \left( \hat{B}_k \hat{\rho}(t) \hat{B}_k^{\dagger} 
			-\frac{1}{2} \hat{\rho}(t) \hat{B}_k^{\dagger} \hat{B}_k -\frac{1}{2} \hat{B}_k^{\dagger} \hat{B}_k \hat{\rho}(t)  \right)
\end{align}
can be found as 
\begin{align}
	\hat{\rho}(t) = \mathbb{E}\bigg( | \psi(t) \rangle\langle \psi(t) | \bigg),
\end{align}
where $\mathbb{E}$ denotes averaging over a Poisson noise and  $| \psi(t) \rangle$ is a solution to the following stochastic procedure (See Sec. 4.3.4 of Ref. \cite{jacobs2014quantum}):
\begin{enumerate}
\item Generate two sequences of random numbers $\{r_{\hat{A}_k}\}_k$ and $\{r_{\hat{B}_k}\}_k$ each uniformly distributed on the interval $[0,1]$.
\item Solve numerically the Schr\"{o}dinger-like equation (\ref{EqWFuncMonteCarloPoisson})
\end{enumerate}

We do need to solve the master equation to find the density matrix evolution $\hat{\rho}(t)$. It can be found by solving an ensemble of the Schr\"{o}dinger equations

In between the projections we need to solve the following equation
\begin{align}\label{EqWFuncMonteCarloPoisson}
	\frac{d}{dt} | \psi(t) \rangle = -\left( \frac{i}{\hbar} \hat{H}(t) + \frac{1}{2} \sum_k \left[ 
			(\hat{A}_k^{\dagger}\hat{A}_k) (\hat{x}) - \lambda_{\hat{A}_k}(t)
			 + (\hat{B}_k^{\dagger}\hat{B}_k) (\hat{p}) - \lambda_{\hat{B}_k}(t) \rangle
		\right] \right) |\psi(t) \rangle, \\
	\lambda_{\hat{A}_k}(t) = \langle \psi(t) | (\hat{A}_k^{\dagger}\hat{A}_k) (\hat{x}) | \psi(t) \rangle, \qquad
	\lambda_{\hat{B}_k}(t) = \langle \psi(t) | (\hat{B}_k^{\dagger}\hat{B}_k) (\hat{x}) | \psi(t) \rangle, \notag
\end{align}
where $(\hat{A}_k^{\dagger}\hat{A}_k) (\hat{x})$ $\left[ (\hat{B}_k^{\dagger}\hat{B}_k) (\hat{p}) \right]$ denotes the fact that the product of two operators 
$\hat{A}_k^{\dagger}\hat{A}_k$ $\left[ \hat{B}_k^{\dagger}\hat{B}_k \right]$ is a function of the quantum position operator $\hat{x}$ (the quantum momentum operator $\hat{p}$) only. Note that the fact that $\hat{A}_k^{\dagger}\hat{A}_k$ solely dependents on $\hat{x}$ does not imply that $\hat{A}_k = A_k(\hat{x})$  is only a function of $\hat{x}$. In fact, $\hat{A}_k$ may depend on both $\hat{x}$ and $\hat{p}$. For example, $\hat{A} = e^{ik\hat{x}} f(\hat{p}) \Longrightarrow \hat{A}^{\dagger} \hat{A} = |f(\hat{p})|^2$.  

The evolution governed by Eq. (\ref{EqWFuncMonteCarloPoisson}) preserves the wave function's norm
 \begin{align}\label{EqNormalizationWfMcP}
 	\frac{d}{dt} \langle \psi(t) | \psi(t) \rangle = 0.
 \end{align}
According to Eqs. (\ref{EqTExpGeneric}) and (\ref{EqMotionForArbitraryG})
\begin{align}
	&| \psi(t + \delta t) \rangle = \hat{\mathcal{T}} \exp\left[ -\int_{t}^{t+\delta t} \left( \frac{i}{\hbar} \hat{H}(\tau) + \frac{1}{2} \sum_k \left[ 
			\hat{A}_k^{\dagger}\hat{A}_k -  \lambda_{\hat{A}_k}(\tau) 
			 + \hat{B}_k^{\dagger}\hat{B}_k  -  \lambda_{\hat{B}_k}(\tau)
		\right] \right) d\tau \right] | \psi(t) \rangle \notag\\
	&=  e^{\sum_k \int_{t}^{t+\delta t} \left[\lambda_{\hat{A}_k}(\tau)  + \lambda_{\hat{B}_k}(\tau) \right] d\tau }
		\exp\left[ -\left( \frac{i}{\hbar} \hat{H}\left(t + \frac{\delta t}{2}\right) + \frac{1}{2} \sum_k \left[ 
			\hat{A}_k^{\dagger}\hat{A}_k  + \hat{B}_k^{\dagger}\hat{B}_k 
		\right] \right) \delta t \right] | \psi(t) \rangle + O\left( \delta t^3 \right) 
    \notag\\
	&= N \underbrace{ \exp\Bigg[ -  \frac{i}{\hbar} \Bigg( \hat{H}\left(t + \frac{\delta t}{2}\right) 
			-\overbrace{ \frac{i\hbar}{2} \sum_k \hat{A}_k^{\dagger}\hat{A}_k }^{\parbox{7em}{correction to\\potential energy}}
			-\overbrace{ \frac{i\hbar}{2} \sum_k  \hat{B}_k^{\dagger}\hat{B}_k }^{\parbox{7em}{correction to\\kinetic energy}}
	\Bigg) \delta t \Bigg] }_{\mbox{resmebles Eq. (\ref{SecondOrderTexpPartEq})}} | \psi(t) \rangle + O\left( \delta t^3 \right),
\end{align}
where $N$ is the normalization constant chosen such that $\langle \psi(t + \delta t) | \psi(t + \delta t) \rangle = 1$ [according to Eq. (\ref{EqNormalizationWfMcP})].

In fact, the corrections to the kinetic and potential energies of the Hamiltonian $\hat{H}$ resemble the absorbing boundary (\ref{EqAbsBoundaryCoreection}) discussed in Sec. \ref{Sec_AbsorbibgBoundary}.
\begin{align}
	P_{\hat{A}_k, \hat{B}_k} (t) = \exp\left( -\int_0^t \lambda_{\hat{A}_k, \hat{B}_k}(\tau)d\tau \right).
\end{align}
\begin{align}
	P_{\hat{A}_k, \hat{B}_k} (t + \delta) = P_{\hat{A}_k, \hat{B}_k} (t) \exp\left(-\left[  \lambda_{\hat{A}_k, \hat{B}_k}(t + \delta t) + \lambda_{\hat{A}_k, \hat{B}_k}(t) \right] \delta t/2 \right) + O\left( \delta t^3 \right).
\end{align}

\subsection{Pauli equation}
Using the identity
\begin{align}
	e^{ia(n_{x} \sigma_{x} + n_y\sigma_{y} + n_z\sigma_{z})} = \cos(a) + i(n_{x} \sigma_{x} + n_y\sigma_{y} + n_z\sigma_{z})\sin(a),
	\quad
	n_{x}^2 + n_y^2 + n_z^2 = 1,
\end{align}
we define an auxiliary operator
\begin{align}
	\mathcal{P}(a; c_0, c_1, c_2, c_3) \Upsilon =&
	e^{ ia \sum_{k=0}^3 c_k \widehat{\sigma}_k } \left( \Upsilon_1 \atop \Upsilon_2 \right) =
		e^{ia c_0} \left( 
			\cos(ab)\Upsilon_1 + i \sin(ab)\left(c_3 \Upsilon_1 + [c_1 -ic_2]\Upsilon_2\right) / b  
		\atop 
			\cos(ab) \Upsilon_2 + i \sin(ab)\left([c_1 + ic_2]\Upsilon_1 - c_3 \Upsilon_2 \right) / b
		\right), \notag\\
	& b = \sqrt{ c_1^2 + c_2^2 + c_3^2 }.
\end{align}

\subsection{A model of molecule}\label{Sec:Molecular2StateModel}

We consider the two (electronic) state approximation for a molecule \cite{schwendner_photodissociation_1997}. The Schr\"{o}dinger equation describing a molecule reads
\begin{align}
	i\hbar\frac{\partial}{\partial t} \left( \psi_g(x, t) \atop \psi_e(x, t) \right) =& 
	\left(
	\begin{matrix}
		K(\hat{p}) + V_g(\hat{x}) & V_{eg}(\hat{x}, t) \\
		V_{eg}(\hat{x}, t) & K(\hat{p}) + V_e(\hat{x})
	\end{matrix}
	\right)
	\left( \psi_g(x, t) \atop \psi_e(x, t) \right)
	= \left[ K(\hat{p}) + U(\hat{x}, t) \right] \left( \psi_g(x, t) \atop \psi_e(x, t) \right), \\
	U(\hat{x}, t) =& \left[ V_g(\hat{x}) + V_e(\hat{x}) \right]/2 
		+ \sigma_{x} V_{eg}(\hat{x}, t)  +  \sigma_{z} \left[ V_g(\hat{x}) - V_e(\hat{x}) \right]/2 ,
\end{align}  
where $\psi_g(x, t)$ [$\psi_e(x, t)$] denotes the vibrational wave-packet residing in the ground (excited) electronic state, $K(\hat{p})$ is the kinetic energy, $V_g(\hat{x})$ [$V_e(\hat{x})$] is the ground (excited) state adiabatic potential curve and
\begin{align}
	V_{eg}(\hat{x}, t) = -\mu_{eg}(\hat{x}) E(t),
\end{align}
is the laser-molecule interaction with $E(t)$ denoting a laser pulse and $\mu_{eg}$ denoting the dipole moment.

Lifting this model to the Wigner phase-space representation, we obtain
\begin{align}
	i\hbar\frac{\partial}{\partial t}W(x,p; t) 
	=
	&\left[ K\left(p + \frac{\hbar}{2}\overrightarrow{\hat{\lambda}} \right) + U\left(x - \frac{\hbar}{2}\overrightarrow{\hat{\theta}}, t\right) \right]W(x,p; t)
	\notag\\
    &-
	W(x,p;t)\left[ K\left(p - \frac{\hbar}{2}\overleftarrow{\hat{\lambda}} \right) + U\left(x + \frac{\hbar}{2}\overleftarrow{\hat{\theta}}, t\right) \right],
\end{align}
where $W(x,p;t)$ is the Wigner matrix valued function,
\begin{align}
	W& (x, p; t + dt) = e^{-\frac{idt}{\hbar} \left[ K\left(p + \frac{\hbar}{2}\overrightarrow{\hat{\lambda}} \right) + U\left(x - \frac{\hbar}{2}\overrightarrow{\hat{\theta}}, t\right) \right]} 
	W(x,p; t) 
	e^{\frac{idt}{\hbar} \left[ K\left(p - \frac{\hbar}{2}\overleftarrow{\hat{\lambda}} \right) + U\left(x + \frac{\hbar}{2}\overleftarrow{\hat{\theta}}, t\right) \right]}
	\notag\\
	&= e^{-\frac{idt}{\hbar} K\left(p + \frac{\hbar}{2}\overrightarrow{\hat{\lambda}} \right)}
	e^{-\frac{idt}{\hbar} U\left(x - \frac{\hbar}{2}\overrightarrow{\hat{\theta}}, t\right) } 
	W(x,p; t)
	e^{\frac{idt}{\hbar} U\left(x + \frac{\hbar}{2}\overleftarrow{\hat{\theta}}, t\right)} 
	e^{\frac{idt}{\hbar} K\left(p - \frac{\hbar}{2}\overleftarrow{\hat{\lambda}} \right)} + O\left(dt^2\right)
	\notag\\
	&=
	\mathcal{F}_{\lambda \to x}^{-1} \left\{
	e^{-\frac{idt}{\hbar} K\left(p + \frac{\hbar}{2}\lambda \right)}
		\mathcal{F}_{x \to \lambda} \mathcal{F}_{\theta \to p}^{-1} \left[
			e^{-\frac{idt}{\hbar} U\left(x - \frac{\hbar}{2}\theta, t\right) } 
				\mathcal{F}_{p \to \theta} \left[ W(x,p; t) \right]
			e^{\frac{idt}{\hbar} U\left(x + \frac{\hbar}{2}\theta, t\right)} 
		\right]
	e^{\frac{idt}{\hbar} K\left(p - \frac{\hbar}{2}\lambda \right)} 
	\right\}
    \notag\\
    &\qquad\qquad+ O\left(dt^2\right)
	\notag\\
	&=
	\mathcal{F}_{\lambda \to x}^{-1} \left\{
		e^{\frac{idt}{\hbar} \left[ K(p_{-}) -K(p_{+}) \right]}
		\mathcal{F}_{x \to \lambda} \mathcal{F}_{\theta \to p}^{-1} \left[
			e^{\frac{idt}{2\hbar} \left[ V_{g}(x_{+}) - V_{g}(x_{-}) + V_{e}(x_{+}) - V_{e}(x_{-})\right]} 
				T(x_{-}) 
                \right.\right.
                \notag\\
                %\left.\left.
                &\qquad\qquad\qquad\qquad
                \qquad\qquad\qquad\qquad
                \times
                \mathcal{F}_{p \to \theta} \left[ W(x,p; t) \right] T^{\dagger} (x_{+})
		\Bigr]%\right]
	\Bigr\},
    %\right\},
\end{align}
where $p_{\pm} = p \pm \hbar\lambda/2$ and $x_{\pm} = x \pm \hbar\theta/2$.
Using the identity
\begin{align}
	e^{ia(n_{x} \sigma_{x} + n_y\sigma_{y} + n_z\sigma_{z})} = \cos(a) + i(n_{x} \sigma_{x} + n_y\sigma_{y} + n_z\sigma_{z})\sin(a),
	\quad
	n_{x}^2 + n_y^2 + n_z^2 = 1,
\end{align}
we obtain
\begin{align}
	&T(q) = e^{-\frac{idt}{\hbar} \left( \sigma_{x} V_{eg}(q, t)  +  \sigma_{z} \frac{1}{2}\left[ V_g(q) - V_e(q) \right] \right) } 
		= C - iL \sigma_{x} -iM\sigma_{z}
		= \left(
		\begin{matrix}
			C - iM & -iL \\
			-iL & C + iM
		\end{matrix}
		\right),
	\\
	&L = SV_{eg}(q, t), \quad
	M =  S\frac{V_{g}(q) - V_{e}(q)}{2}, \quad
	C = \cos \frac{D dt}{\hbar} , \quad
	S = \frac{1}{D} \sin \frac{D dt}{\hbar}, 
    \notag\\
	&D =  \sqrt{V_{eg}^2(q, t) + \frac{1}{4}[ V_{g}(q) - V_{e}(q)]^2}.
\end{align}
Note the following symmetry relations in the Wigner matrix
\begin{align}
	W(x,p; t) =
	\left(
		\begin{matrix}
		W_{g}(x,p; t) & W_{ge}(x,p; t) \\
		W_{ge}^*(x,p; t) & W_{e}(x,p; t)
		\end{matrix}
	\right),
\end{align}
where $W_{g}$ and $W_{e}$ are real valued functions and $W_{eg}$ is complex.

%\chapter{Examples II}
%\input{Chapters/misc}

\blankpage
\bibliographystyle{ieeetr}
\bibliography{references}
%\nocite{*}

\end{document}